\documentclass[12pt]{article}
\usepackage{amsmath}
\usepackage{graphicx}%
\usepackage{amsfonts}%
\usepackage{amssymb}
\usepackage{overpic}
\usepackage{subfigure}
\usepackage{rotating}
\usepackage{url}
\usepackage{fullpage}
\usepackage{comment}
\usepackage{enumerate}
\usepackage{multirow}
\usepackage{booktabs}
\usepackage{tabularx}
\usepackage{makecell}
\usepackage{booktabs}
\usepackage{natbib}
\usepackage[colorlinks = true,
            linkcolor = blue,
            urlcolor  = blue,
            citecolor = blue,
            anchorcolor = blue]{hyperref}
\usepackage[letterpaper,left=1in,top=1in,right=1in,bottom=1in,lines=25]{geometry}

\usepackage{setspace}
\newcommand{\bA}{ {\boldsymbol A} }

\newcommand{\bD}{ {\boldsymbol D} }

\newcommand{\bH}{ {\boldsymbol H} }

\newcommand{\bI}{ {\boldsymbol I} }

\newcommand{\bm}{ {\boldsymbol m} }
\newcommand{\bM}{ {\boldsymbol M} }

\newcommand{\bQ}{ {\boldsymbol Q} }

\newcommand{\bs}{ {\boldsymbol s} }

\newcommand{\bt}{ {\boldsymbol t} }

\newcommand{\bu}{ {\boldsymbol u} }
\newcommand{\bU}{ {\boldsymbol U} }

\newcommand{\bW}{ {\boldsymbol W} }
\newcommand{\bx}{ {\boldsymbol x} }
\newcommand{\bX}{ {\boldsymbol X} }
\newcommand{\by}{ {\boldsymbol y} }

\newcommand{\bgamma}{ {\boldsymbol \gamma} }
\newcommand{\bGamma}{ {\boldsymbol \Gamma} }

\newcommand{\bLambda}{ {\boldsymbol \Lambda} }
\newcommand{\bmu}{ {\boldsymbol \mu} }

\newcommand{\bomega}{ {\boldsymbol \omega} }
\newcommand{\bOmega}{ {\boldsymbol \Omega} }

\newcommand{\bSigma}{ {\boldsymbol \Sigma} }

\newcommand{\bzero}{ {\boldsymbol 0} }

\newcommand{\given}{\,|\,}

\newtheorem{theorem}{Theorem}[section]
\newtheorem{lemma}[theorem]{Lemma}

\newenvironment{proof}[1][Proof]{\begin{trivlist}
\item[\hskip \labelsep {\bfseries #1}]}{\end{trivlist}}

\newcommand{\qed}{\nobreak \ifvmode \relax \else
      \ifdim\lastskip<1.5em \hskip-\lastskip
      \hskip1.5em plus0em minus0.5em \fi \nobreak
      \vrule height0.75em width0.5em depth0.25em\fi}

\title{High Dimensional Bayesian Network Classification with Network Global-Local Shrinkage Priors}
\author{Sharmistha Guha $\&$ Abel Rodriguez}
\begin{document}
\maketitle
\begin{abstract}
This article proposes a novel Bayesian classification framework for networks with labeled nodes. While literature on statistical modeling of network data typically involves analysis of a single network, the recent emergence of complex data in several
biological applications, including brain imaging studies, presents a need to devise a network classifier for subjects. This article considers an application from a brain connectome study, where the overarching goal is to classify subjects into two separate groups based on their brain network data, along with identifying
influential regions of interest (ROIs) (referred to as nodes). Existing approaches either treat all edge weights as a long vector or summarize the network information with a few summary measures. Both these approaches ignore the full network structure, may
lead to less desirable inference in small samples and are not designed to identify significant network nodes. We propose a novel binary logistic regression framework with the network as the predictor and a binary response, the network predictor coefficient being modeled using a novel class \emph{global-local} shrinkage priors. The framework is able to accurately detect nodes and edges in the network influencing the classification. Our framework is implemented using an efficient Markov Chain Monte Carlo algorithm. Theoretically, we show asymptotically optimal classification for the proposed framework when the number of network edges grows faster than the
sample size. The framework is empirically validated by extensive simulation studies
and analysis of a brain connectome data.
\end{abstract}

\section{Introduction}
Of late, the statistical literature has paid substantial attention to the unsupervised analysis of a single network, thought to be generated from a variety of classic models, including random graph models \cite{erdos1960evolution}, exponential random graph models \cite{frank1986markov}, social space models \cite{hoff2002latent, hoff2005bilinear, hoff2009multiplicative} and stochastic block models \cite{nowicki2001estimation}. These models have found prominence in social networking applications where the nodes in the network are exchangeable. However, there are pertinent biological and physiological applications in which network nodes are labeled and a network is available corresponding to every individual. Section \ref{connectome} presents one such example from a brain connectome study, where brain networks are available for multiple individuals who are classified as subjects with high or low IQ (Intelligence Quotient). In this study, the human brain has been divided according to the Desikan Atlas \cite{desikan2006automated} that identifies 34 cortical regions of interest (ROIs) both in the left and the right hemispheres of the human brain, implying $68$ cortical ROIs in all. A \emph{brain network} for each subject is represented by a symmetric adjacency matrix whose rows and columns are labeled corresponding to different ROIs (shared among networks corresponding to all individuals) and entries correspond to estimates of the number of \emph{fibers} connecting pairs of brain regions. The scientific goal in this setting pertains to developing a predictive rule for classifying a new subject as having low or high IQ based on his/her observed brain network with labeled nodes. Additionally, it is of specific interest for neuroscientists to identify influential brain regions (nodes in the brain network) and significant connections between different brain regions predictive of IQ.

\cite{guha2020bayesian} discuss the network regression problem with a continuous response and an undirected network predictor. However, there are pertinent biological and physiological studies where a network along with a binary response is obtained for each subject. The goal of these studies is usually to classify the networks according to the binary response and predict the associated binary response from a network. We refer to this problem as the network or graph classification problem. Additionally, \cite{guha2020bayesian} focus on a specific network shrinkage prior, whereas this article generalizes the inference to a class of network global-local shrinkage priors, which includes the prior specification in \cite{guha2020bayesian} as a special case.

 Earlier literature on network or graph classification has been substantially motivated by the problem of classification of chemical compounds \cite{srinivasan1996theories}, \cite{helma2001predictive}, where a graph represents a compound's molecular structure. In such analyses, certain discriminative patterns in a graph are identified and used as features for training a standard classification method
\cite{deshpande2005frequent}, \cite{fei2010boosting}. Another type of method is based on graph kernels \cite{vishwanathan2010graph}, which defines a similarity measure between two networks. Both of these approaches are computationally feasible only for small networks, do not account for uncertainty, and do not facilitate influential network node identification. When the number of network nodes is moderately large, a common approach to network classification is to use a few summary measures (average degree, clustering coefficient, or average path length) from the network and then apply statistical procedures in the context of standard classification methods (see, for e.g., \cite{bullmore2009complex} and references therein). These procedures have been recently employed in exploring the relationship between the brain network and neuropsychiatric diseases, such as Parkinson's \cite{olde2013disrupted} and Alzheimer's \cite{daianu2013breakdown}, but the analyses are sensitive to the chosen network topological measures, with substantially different results obtained for different types of summary statistics. Indeed, global summary statistics collapse all local network information, which can affect the accuracy of classification. Furthermore, identification of the impact of specific nodes on the response, which is of clear interest in our setting, is not feasible. As with network regression problems, an alternate approach proceeds to vectorize the network predictor and treat edge weights together as a long vector followed by developing a high dimensional regression model with this long vector of edge weights as predictors \cite{richiardi2011decoding}; \cite{craddock2009disease}; \cite{zhang2012pattern}. This approach can take advantage of the recent developments in high dimensional binary regression, consisting of both penalized optimization \cite{tibshirani1996regression} and Bayesian shrinkage \cite{park2008bayesian}; \cite{carvalho2010horseshoe}; \cite{armagan2013generalized} perspectives. However,  as mentioned in \cite{guha2020bayesian}, this treats the links of the network as exchangeable, ignoring the fact that coefficients involving common nodes can be expected to be correlated a priori. In a related work, \cite{vogelstein2013graph} propose to look for a minimal set of nodes which best explains the difference between two groups of networks. This requires solving a combinatorial problem. Again, \cite{durante2017} propose a high dimensional Bayesian tensor factorization model for a population of networks that allows to test for local edge differences between two groups of subjects. Both of these approaches tend to focus mainly on classification and are not designed to detect important nodes and edges impacting the response.

Our goal in this article is to develop a high-dimensional Bayesian network classifier that additionally infers on influential nodes and edges impacting classification. To achieve this goal, we formulate a high dimensional logistic network regression model with the binary response regressed on the network predictor corresponding to each subject. The network predictor coefficient is assigned a prior from the class of \emph{Bayesian network global-local shrinkage priors} discussed in this article. The proposed prior imparts low-rank and near sparse structures a priori on the network predictor coefficient. The low-rank structure of the coefficient is designed to address the transitivity effect on the network predictor coefficient and captures the effect of network edge coefficients on classification due to the interaction between nodes. On the other hand, the near sparse structure accounts for the residual effect due to edges.

One important contribution of this article is a careful study of the asymptotic properties of the proposed binary network classification (BNC) framework. In particular, we focus on consistency properties for the posterior distribution of the BNC framework using a specific network global-local shrinkage prior, namely the  \emph{ Bayesian Network Lasso prior}. Theory of posterior contraction for high dimensional regression models has gained
traction lately, though the literature is less developed in shrinkage priors compared to point-mass priors. For example, \cite{castillo2012needles} and \cite{belitser2015needles} have established posterior concentration and variable selection
properties for certain point-mass priors in the normal-means models. The latter article also establishes asymptotically nominal coverage of Bayesian credible sets. Results on posterior concentration and variable selection in high dimensional linear models are also established by \cite{castillo2015bernstein} and \cite{martin2017empirical} for certain point-mass priors. In contrast, literature on posterior contraction properties for high dimensional Bayesian shrinkage priors is relatively limited. To this end, \cite{armagan2013posterior} were the first to show posterior consistency in the ordinary linear regression model with shrinkage priors for low-dimensional settings under the assumption that the number of covariates \emph{does not} exceed the number of observations. Using direct calculations, \cite{van2014horseshoe} show that the posterior based on the ordinary horseshoe prior concentrates at the optimal rate for normal-mean problems. Recently, \cite{song2017nearly} considers a general class of continuous shrinkage priors and obtains posterior contraction rates in ordinary high dimensional linear regression models. In the same vein, \cite{wei2017contraction} offers analysis of posterior concentration for logistic regression models with shrinkage priors on coefficients. While \cite{wei2017contraction} are the first to delineate a theoretical approach for ordinary high dimensional binary classification models with shrinkage priors, the study of posterior contraction properties for more structured binary network classification problems in the Bayesian paradigm has not appeared in the literature. In fact, developing the theory for Bayesian network classification with the Bayesian Network Lasso prior proposed in this article is faced with two major challenges. First, the novel Bayesian Network Lasso prior imparts a more complex prior structure (incorporating a low-rank structure in the prior mean of edge coefficients, as described in \cite{guha2020bayesian} than that in \cite{wei2017contraction}, introducing additional theoretical challenges. Second, we aim at proving a challenging but practically desirable result of asymptotically optimal classification when the number of edges in the network predictor grows at a super-linear rate as a function of the sample size. Both of these present obstacles which we overcome in this work. The theoretical results provide insights on how the number of nodes in the network predictor, or the sparsity in the true network predictor coefficients should vary with sample size $n$ to achieve asymptotically optimal classification. We must mention that developing a similar theory for the Bayesian Network Horseshoe prior proposed in this article faces more challenges due to complex prior structure in parameters. We plan to tackle that problem as part of future work.

Section~\ref{sec3} develops the model and the prior distributions. Section~\ref{post_con} discusses theoretical developments justifying the asymptotically desirable prediction from the proposed model. Section~\ref{post_comp} details posterior computation. Results from various simulation experiments and a brain connectome data analysis have been presented in Sections~\ref{simulation} and \ref{connectome} respectively. Finally, Section~\ref{conclusion} concludes the article with a brief discussion of the proposed methodology.

\section{Model Formulation}\label{sec3}

In the context of network classification, we propose the high dimensional logistic regression model of the binary response $y_i\in\{0,1\}$ on the undirected network predictor $\bA_i$ as
\begin{align}\label{initial_model_bin}
y_i \sim Ber \left[\frac{\exp(\psi_i)}{1+\exp(\psi_i)}\right],\:\:\psi_i=\mu+\langle \bA_i,\bGamma\rangle_F,
\end{align}
where $\bGamma$ is a $V\times V$ symmetric network coefficient matrix whose $(k,l)$th element is given by $\gamma_{k,l}/2$, with $\gamma_{k,k}=0$, for all $k=1,...,V$.

Model (\ref{initial_model_bin}) can be expressed in the form of a generalized linear model. To be more specific, $\langle\bA_i, \bGamma\rangle_F=\sum\limits_{1 \leq k < l \leq V} a_{i,k,l} \gamma_{k,l}$, so that
$\psi_i = \mu+\sum\limits_{1 \leq k < l \leq V} a_{i,k,l} \gamma_{k,l} $ and
the probability mass function of $y_i$ can be written as
\begin{align}\label{bin132}
p(y_i) = \frac{\exp(\psi_i)^{y_i}}{1+\exp(\psi_i)}
\end{align}
Note that, if $\bx_i=(a_{i,1,2},...,a_{i,(V-1),V})'\in\mathbb{R}^{V(V-1)/2}$ is the collection of all upper triangular elements of $\bA_i$, and $\bgamma=(\gamma_{1,2},...,\gamma_{(V-1),V})'\in\mathbb{R}^{V(V-1)/2}$ is the vector of corresponding upper triangular elements of $2\bGamma$, then (\ref{initial_model_bin}) can be written as
\begin{align}\label{bin_new}
y_i\sim Ber\left(f_{\bgamma}(\bx_i)\right),\:\: f_{\bgamma}(\bx_i)=\frac{\exp(\mu+\bx_i'\bgamma)}{(1+\exp(\mu+\bx_i'\bgamma))}.
\end{align}
Although the binary network regression model is proposed for the logit link, it assumes natural extension for any other link function.
The next section describes a class of network global-local shrinkage priors on network coefficients.

\subsection{Bayesian network global-local shrinkage prior on the network predictor coefficient}\label{sec223}

In this article, we propose  the network global-local shrinkage prior given by,
\begin{align}\label{global-local_loc}
\gamma_{k,l}|s_{k,l},\sigma^2\sim N(\bu_k'\bLambda\bu_l,\sigma^2s_{k,l}^2),\:\:\sigma\sim H_1(\cdot),\:\:s_{k,l}\sim H_2(\cdot).
\end{align}
Note that this framework a priori centers $\gamma_{k,l}$ at a low-rank decomposition and controls the spread of the prior distribution of $\gamma_{k,l}$ using a global-local shrinkage prior. The formulation includes a wide variety of network shrinkage priors by choosing different functions $H_1(\cdot)$ and $H_2(\cdot)$. For example, \cite{guha2020bayesian}
 have investigated a particular class of such prior distributions, obtained by choosing $H_1(\sigma)=\delta_{1}(\sigma),$ where $\delta_1(\sigma)$ is the Dirac-delta function that is defined as $\delta_1(\sigma)=1$ if $\sigma=1$, and $0$ otherwise; and  $H_2(s_{k,l}^2)$ as an exponential density, referred to as the \emph{Network Lasso prior}. To show the generality of (\ref{global-local_loc}), we additionally investigate performance of (\ref{global-local_loc}) in binary regression with $s_{k,l}\sim C^{+}(0,1)$ and $\sigma\sim C^{+}(0,1)$. The resulting prior is referred to as the \emph{Network Horseshoe prior}. The rest of the hierarchy on $\lambda_r$'s, $\bu_k$'s follows as in \cite{guha2020bayesian}.

\section{Posterior Contraction of the Binary Network Classification Model}\label{post_con}
This section establishes convergence results for (\ref{initial_model_bin}) with $\gamma_{k,l}$'s following the Bayesian Network Lasso shrinkage prior. From the hierarchical specification given in (\ref{global-local_loc}), the Bayesian Network Lasso shrinkage prior is given by $\gamma_{k,l}|s_{k,l}\sim N(\bu_k'\bLambda\bu_l,s_{k,l}^2),\:s_{k,l}^2\sim Exp(\theta_n/2)$. For the theoretical study, a common practice is to fix $\theta_n$ as a function of $n$ \cite{armagan2013generalized}. Our theoretical investigations will also fix $\theta_n$ (the exact expression is given in Condition (\ref{theta_1111}) in the next subsection) with the fixed values specified later.

Here we consider an asymptotic setting in which the number of nodes in the network predictor, $V_n$, grows with the sample size $n$. This paradigm attempts to capture the fact that the number of elements in $\bA_i$,
given by $V_n^2$ can be substantially larger than sample size. Since model (\ref{initial_model_bin}) is equivalent to model (\ref{bin_new}), the size of the coefficient $\bgamma$ in (\ref{bin_new}) is also a function of $n$, given by $q_n=\frac{V_n(V_n-1)}{2}$. This creates theoretical challenges, related to (but distinct from) those faced in showing posterior consistency for high dimensional continuous \cite{armagan2013generalized} and binary regressions \cite{wei2017contraction}.

Let $\by_n=(y_1,...,y_n)'$. Using the superscript $(0)$ to indicate true parameters, the true data generating model is given by
\begin{align}\label{initial_model_bin_tr}
y_i \sim Bernoulli \left[\frac{\exp(\psi_{i}^{(0)})}{1+\exp(\psi_{i}^{(0)})}\right],\:\:\psi_{i}^{(0)}=\langle \bA_i,\bGamma^{(0)}\rangle_F.
\end{align}
where $\bGamma^{(0)}$ is the true network coefficient. Let $\bgamma^{(0)}$ be the vectorized upper triangular part of $\bGamma^{(0)}$. We assume, $\gamma_{k,l}^{(0)}=\bu_k^{(0)'}\bLambda\bu_l^{(0)}+\gamma_{2,k,l}^{(0)}$, where
$\bu_k^{(0)}$ is a $R_0$ dimensional vector, $k=1,...,V$. $\bgamma_2^{(0)}$ is the vector of all $\gamma_{2,k,l}^{(0)}$, $k<l$, and we denote the number of nonzero elements of $\bgamma_2^{(0)}$ by $s_{2,n}^0$, i.e. $||\bgamma_2^{(0)}||_0=s_{2,n}^0$.

For any $\epsilon>0$, define
$\mathcal{A}_n=\left\{\bgamma:\frac{1}{n}\sum\limits_{i=1}^n|f_{\bgamma}(\bx_i)-f_{\bgamma^{(0)}}(\bx_i)|\leq\epsilon\right\}$ as a neighborhood around the true density.
Further suppose $\pi_n(\cdot)$ and $\Pi_n(\cdot)$ are the prior and posterior densities of $\bgamma$ with $n$ observations, so that
\begin{align*}
\Pi_n(\mathcal{A}_n^c)=\frac{\int_{\mathcal{A}_n^c}p_{\bgamma}(\by_n)\pi_n(\bgamma)}
                          {\int p_{\bgamma}(\by_n)\pi_n(\bgamma)},
\end{align*}
where $p_{\bgamma}(\by_n)$ denotes the likelihood of the $n$dimensional response vector $\by_n$.

\subsection{Main Results}\label{theorem-binary-new}
To show the posterior contraction results, we follow \cite{wei2017contraction} and \cite{armagan2013generalized}, with substantial modifications required due to the nature of our proposed network lasso prior distribution. In proving the results, we make a couple of simplifications. It is assumed that the dimension $R$ of $\bu_k$ is fixed and is the same as $R_0$, the dimension of $\bu_k^{(0)}$. Consequently, \emph{effective dimensionality} is not required to be estimated, and hence $\bLambda=\bI$ is a non-random matrix. Additionally, we assume $\bM$ to be non-random and $\bM=\bI$. We emphasize that both these assumptions are \emph{not} essential for the posterior contraction rate result to be true, and are only introduced for simplifying calculations.

For two sequences $\{C_{1,n}\}_{n\geq 1}$ and $\{C_{2,n}\}_{n\geq 1}$, $C_{1,n}=o(C_{2,n})$ if $C_{1,n}/C_{2,n}\rightarrow 0$, as $n\rightarrow\infty$. To begin with, we state the following assumptions under which posterior contraction will be shown.
\begin{enumerate}[(A)]
\item $\sup\limits_{r=1,..,R;k=1,..,V_n}|u_{k,r}^{(0)}|<\infty$;
\item $V_n=o(\frac{n}{\log(n)})$;
\item $||\bA_i||_{\infty}$ is bounded for all $i=1,..,$, w.l.o.g assume $||\bA_i||_{\infty}\leq 1$.
\item $s_{2,n}^0\log(q_n)= o(n)$
\item $||\bgamma_{2}^{(0)}||_{\infty}<\infty$;
\item $\theta_n=\frac{C}{q_nn^{\rho/2}\log(n)}$ for some $C>0$ and some $\rho\in(1,2)$.n \label{theta_1111}
\end{enumerate}

\textbf{Remark:} Conditions (A), (C) and (E) are technical conditions ensuring that each of the entries in the true network coefficient and the network predictor are bounded. Condition (B) puts an upper bound on the growth of the number of network nodes with sample size to achieve asymptotically optimal classification. Similarly, (D) puts a restriction on the number of nonzero elements of $\bgamma_2^{(0)}$ with respect to $n$.

The following theorem shows contraction of the posterior asymptotically under mild sufficient conditions on $V_n,s_{2,n}^0$.
The proof of the theorem is provided in Appendix F.
\begin{theorem}\label{theorem:main}
Under assumptions (A)-(F) for the Bayesian Network Lasso prior on $\bgamma$, $\Pi_n(\mathcal{A}_n)\rightarrow 0$ in $P_{\bgamma^{(0)}}$ as $n\rightarrow\infty$, for any $\epsilon>0$.
\end{theorem}

\section{Posterior Computation}\label{post_comp}
We have implemented both the Bayesian Network Lasso and Network Horseshoe shrinkage priors on $\bgamma$.
Using the result in \cite{polson2013bayesian}, the data augmented representation of the distribution of $y_i$ given in (\ref{bin132}) follows as below
\begin{align}
p(y_i|\omega_i)=2^{-b}\exp ({k_i \psi_i})\exp (-{\omega_i \psi_i^2}/2),\:\:\omega_i\sim PG(1,0),
\end{align}
where $k_i = y_i - 1/2$. Let $\bx_i=(a_{i,1,2},a_{i,1,3},...,a_{i,1,V}, a_{i,2,3}, a_{i,2,4},..., a_{i,2,V}, ...., a_{i,V-1,V})'$ be of dimension
$q\times 1$, where $q=\frac{V(V-1)}{2}$. Assume $\bX=(\bx_1:\cdots:\bx_n)'$ is an $n\times q$ matrix. Then the conditional likelihood of $\by=(y_1,...,y_n)'$ given $\bomega=(\omega_1,...,\omega_n)'$ and $\bgamma$ is given by
\begin{align*}
p(\by \given \bX, \bgamma, \bomega) &\propto \prod_{i=1}^{n} p(y_i \given \bx_i, \bgamma, \omega_i,...) \\
&\propto \prod_{i=1}^{n} \exp \left\{ (y_i - 0.5) (\mu+\bx'_i \bgamma)  - \omega_i (\mu+\bx'_i \bgamma)^2/2  \right\}\\
&\propto \prod_{i=1}^{n} \exp \left\{ -\frac{\omega_i}{2} \left[ \frac{(y_i - 0.5)}{\omega_i} - (\mu+\bx'_i \bgamma) \right]^2  \right\}
\end{align*}
In matrix notation, the likelihood may be written as
\begin{align*}
p(\by \given \bX, \bgamma, \bomega...) \propto N(\bt \given \mu{\boldsymbol 1}+\bX \bgamma, \bOmega^{-1})
\end{align*}
where $\bt = ((y_1-0.5)/\omega_1,...,(y_n-0.5)/\omega_n)'=(k_1/\omega_1,...,k_n/\omega_n)'$ and $\bOmega=diag(\omega_1,...,\omega_n)$.
While the full posterior distributions for the parameters are not in closed forms, they mostly belong to the standard families. Hence drawing posterior samples using MCMC can be readily implemented.
Appendix D and Appendix E describe full conditional distributions of parameters for Bayesian Network Lasso and Network Horseshoe priors on $\bgamma$, respectively.

Let $\bOmega^{(1)},...,\bOmega^{(L)}$, $\bGamma^{(1)},...,\bGamma^{(L)}$ and $\mu^{(1)},...,\mu^{(L)}$ be the $L$ post burn-in MCMC samples for $\bOmega$, $\bGamma$ and $\mu$ respectively after suitable thinning. To classify a newly observed network $\bM_{*}$ as a member of one of the two groups, we compute $S^{(l)}=\frac{\exp(\mu^{(1)}+\langle \bM_*,\bGamma^{(l)}\rangle)}{1+\exp(\mu^{(1)}+\langle \bM_*,\bGamma^{(l)}\rangle)}$ for $l=1,...,L$. $\bM_*$ is classified as a member of group `low' or `high' if $\frac{1}{L}\sum_{l=1}^L S^{(l)}$ is less than or greater than $0.5$, respectively. To judge sensitivity to the choice of the cut-off, the simulation section presents Area under Curve (AUC) of ROC curves with True Positive Rates (TPR) and False Positive Rates (FPR) of classification corresponding to a range of cut-off values.

Node $k$ is recognized to be influential in the classification process if $\frac{1}{L}\sum_{l=1}^{L}\xi_k^{(l)}>0.5$, where $\xi_k^{(1)},...,\xi_{k}^{(L)}$ are the $L$ post burn-in MCMC samples of $\xi_k$. Again, one of the goals of the proposed framework is to identify influential network edges impacting the response. We employ the algorithm described in Appendix C to identify influential edges. The algorithm takes care of multiplicity correction by controlling the false discovery rate (FDR) at 5\% level. 
Finally, we present an estimate of $P(R_{eff}=r\given Data)$ computed by $\frac{1}{L}\sum_{l=1}^{L}I(\sum_{m=1}^{R}\lambda_m^{(l)}=r)$, where $I(A)$ for an event $A$ is 1 if the event $A$ happens and $0$ otherwise, and $\lambda_m^{(1)},...,\lambda_m^{(L)}$ are the $L$ post burn-in MCMC samples of $\lambda_m$.

\section{Simulation Studies}\label{simulation}
This section evaluates the inferential and classification ability of our proposed Bayesian network classification (BNC) framework, along with a number of competitors, using synthetic networks generated under various simulation settings. Our proposed network classification approach with the Bayesian Network Lasso prior and the Bayesian Network Horseshoe prior are referred to as the Bayesian Network Lasso classifier (BNLC) and Bayesian Network Horseshoe classifier (BNHC), respectively. In each simulation, we assess the ability of the BNLC and BNHC approaches to correctly identify influential nodes and edges, to accurately estimate predictive edge coefficients and to classify a network with precise characterization of uncertainties.
Classification performance of both methods are assessed using the area under the Receiving Operating Characteristics (ROC) curve (AUC).

To study all competitors under various data generation schemes, we simulate the response from (\ref{initial_model_bin}) given by
\begin{align}\label{data_generate}
y_i\sim Ber\left(\frac{\exp(\mu_0+\langle \bA_i,\bGamma_0\rangle_F)}{1+\exp(\mu_0+\langle \bA_i,\bGamma_0\rangle_F)}\right),\:\:
\end{align}
where $\bGamma_{0}$ is a symmetric matrix with zero diagonal entries. The intercept $\mu_0$ is fixed at 2 in all simulation scenarios. We consider two different schemes of generating the network $\bA_i$, referred to as
\emph{Simulation 1} and \emph{Simulation 2}, respectively.\\

\noindent \underline{\textbf{Simulation 1.}} In \emph{Simulation 1}, the network edges (i.e., the elements of the matrix $\bA_i$) are simulated
from $\rm{N}(0,1)$. Thus, \emph{Simulation 1} assumes that the network predictor follows an Erdos-Renyi graph.

\noindent \underline{\textbf{Simulation 2.}} In \emph{Simulation 2}, the network predictor $\bA_i$ corresponding to the $i$th sample is generated from a stochastic blockmodel. Here nodes in a simulated network are organized into communities so that nodes in the same community tend to have stronger connections than nodes belonging to different communities. This simulation scenario simulates networks which closely mimic brain connectome networks \cite{bullmore2009complex}. To simulate networks with such community structures, we assign each node a community label, $f_k\in\{1,2,...,3\}$, $k=1,...,V$. The node assignments are the same for all networks in the population. Given the community labels, the $(k,k')$th element of $\bA$ is simulated from $N(m_{f_k,f_{k'}},\sigma_0^2)$, where $m_{k,l}=0.5$ when $k=l$. When $k\neq l$, i.e., the concerned edges connect nodes belonging to different clusters, we sample a fixed number of edge locations randomly and simulate the values from $N(0,1)$, assigning the values at the remaining locations to be $0$. We set $\sigma_0^2=1$ and the three clusters with $8$, $9$ and $8$ nodes respectively, in the three communities. We note that the network predictors are simulated from a stochastic blockmodel in \emph{Simulation 2} which also ensures transitivity in the network predictor.

\noindent\underline{\emph{Simulating the network predictor coefficient $\bGamma_{0}$.}}  In both Simulations 1 and 2, the network predictor coefficient $\bGamma_{0}$ is constructed as the sum of two matrices $\bGamma_{0,1}$ and $\bGamma_{0,2}$. We provide the details of constructing the two matrices as below.

In both Simulations 1 and 2, we draw $V$ latent variables $\bu_{k,0}$, each of dimension $R_{g}$, from a mixture distribution given by
\begin{align}\label{wk}
\bu_{k,0} \sim \pi N_{R_{g}} (\bu_{m,g}, u_{s,g}^2) + (1 - \pi) \delta_{{\boldsymbol 0}}; \: k \in \{1,...,V\},
\end{align}
where $\delta_{{\boldsymbol 0}}$ is the Dirac-delta function and $\pi$ is the probability of any $\bu_{k,0}$ being nonzero. Define a symmetric matrix $\bGamma_{0,1}$ whose $(k,l)$th element is given by $\frac{\bu_{k,0}'\bu_{l,0}}{2}$, $k<l$ and $=0$ if $k=l$. Note that if $\bu_{k,0}$ is zero, then the $k$th node has no contribution to the mean function in (\ref{data_generate}), i.e., the $k$th node becomes non-influential in predicting the response. Since $(1-\pi)$ is the probability of a node being inactive, it is referred to as the \emph{node sparsity} parameter in the context of the data generation mechanism under \emph{Simulations 1} and \emph{2}. 
All elements of $\bu_{m,g}$ are taken to be $0.5$ and $u_{s,g}$ is taken to be $1$.

We also construct another symmetric sparse matrix $\bGamma_{0,2}$ to add additional edge effects corresponding to edges connecting a few randomly selected nodes. Let $\pi_2$ be the proportion of nonzero elements of $\bGamma_{0,2}$, set randomly at either $0.05$ or $0.1$. We randomly choose $\pi_2$ proportion of locations from the set of all $(k,l)$. The nonzero entries are drawn using one of the three following strategies:\\
\textbf{Strategy 1:} Nonzero entries are simulated from $\rm{N}(1,0.1)$.\\
\textbf{Strategy 2:} Nonzero entries are simulated from $\rm{N}(0.5,0.1)$.\\
\textbf{Strategy 3:} All nonzero entries are fixed at $0.5$.\\
The quantity $(1-\pi_2)$ is referred to as the \emph{residual edge sparsity}. 

Note that the specification of true edge coefficients largely preserves the transitivity property in $\bGamma_{0}$. To see this, note that $\bGamma_{0,2}$ is highly sparse, so that $\gamma_{0,1,k,l}=\gamma_{0,k,l}$ for most pairs $(k,l)$, $k<l$. For those pairs, $\gamma_{0,k,l}\neq 0$ and $\gamma_{0,l,l'}\neq 0$ imply that $\bu_{k,0}\neq\bzero$, $\bu_{l,0}\neq \bzero$ and $\bu_{l',0}\neq \bzero$. Thus it follows that $\gamma_{0,k,l'}=\frac{\bu_{k,0}'\bu_{l',0}}{2}\neq 0.$

For a comprehensive picture of \emph{Simulation 1} and \emph{Simulation 2}, we consider $4$ different cases each in both simulations as summarized in Table~\ref{Tab1} and \ref{Tab2} respectively. In each of these cases, the network predictor coefficient and the response are generated by changing the node sparsity $(1-\pi)$, the residual edge sparsity $(1-\pi_2)$ and the true dimension $R_{g}$ of the latent variables $\bu_{k,0}$'s.
The table also presents the maximum fitted dimension $R$ of the latent variables $\bu_k$ for the logistic regression model (\ref{bin132}). Note that the various cases also allow model mis-specification with unequal choices of $R$ and $R_{g}$.

\begin{table}[!th]
\begin{center}

\begin{tabular}
[c]{cccccccc}
\hline
Cases & $R_{g}$ & $R$ & Node & Residual Edge & Strategy \\
 &  &  & Sparsity ($1-\pi$) & Sparsity ($1-\pi_2$) & \\
\hline
Case - 1 & 2 & 2 & 0.5 & 0.95 & Strategy 1\\
Case - 2 & 3 & 5 & 0.6 & 0.95 & Strategy 1\\
Case - 3 & 2 & 5 & 0.5 & 0.90 & Strategy 2\\
Case - 4 & 2 & 5 & 0.4 & 0.90 & Strategy 3\\
\hline
\end{tabular}
\caption{Table presents different cases for \textbf{\emph{Simulation 1}}. The true dimension $R_{g}$ is the dimension of vector object $\bu_{k,0}$ using which data has been generated. The maximum dimension $R$ is the dimension of vector object $\bu_k$ using which the model has been fitted. Node sparsity and residual edge sparsity are described in the text.}\label{Tab1}
\end{center}
\end{table}

\begin{table}[!th]
\begin{center}

\begin{tabular}
[c]{cccccccc}
\hline
Cases & $R_{g}$ & $R$ & Node & Residual Edge & Strategy \\
           &               &         & Sparsity ($1-\pi$) & Sparsity ($1-\pi_2$) & \\
\hline
Case - 1 & 2 & 2 & 0.5 & 0.95 & Strategy 1\\
Case - 2 & 2 & 4 & 0.5 & 0.95 & Strategy 1\\
Case - 3 & 2 & 3 & 0.7 & 0.95 & Strategy 1\\
Case - 4 & 2 & 5 & 0.4 & 0.90 & Strategy 3\\
\hline
\end{tabular}
\caption{Table presents different cases for \textbf{\emph{Simulation 2}}. The true dimension $R_{g}$ is the dimension of vector object $\bu_{k,0}$ using which data has been generated. The maximum dimension $R$ is the dimension of vector object $\bu_k$ using which the model has been fitted. Node sparsity and residual edge sparsity are described in the text.}\label{Tab2}
\end{center}
\end{table}

\noindent As competitors, we use generic variable selection and shrinkage methods that treat edges between nodes together as a long predictor vector to run high dimensional regression, thereby ignoring the relational nature of the predictor. More specifically, we use Lasso \cite{tibshirani1996regression}, which is a popular penalized optimization scheme, and the Bayesian Lasso (BLasso for short)\cite{park2008bayesian} and Bayesian Horseshoe (BHS for short) priors \cite{carvalho2010horseshoe}, which are popular Bayesian shrinkage regression methods, all three under the logistic regression framework. We use the \texttt{glmnet} package in \texttt{R} \cite{friedman2010regularization} to implement the frequentist Lasso, while we write our own codes for BLasso and BHS. A comparison with these methods will indicate any relative advantage of exploiting the structure of the network predictor.
Additionally, we compare our methods to a frequentist approach that develops network classification in the presence of a network predictor and a binary response \cite{relion2017network}. We refer to this approach as \emph{Reli\'on.}


All Bayesian competitors are allowed to draw $50,000$ MCMC samples, out of which the first $30,000$ are discarded as burn-ins. Convergence is assessed by comparing different simulated sequences of representative parameters starting at different initial values \cite{gelman2014understanding}. All posterior inference is carried out based on the rest $20,000$ MCMC samples after suitably thinning the post burn-in chain. We monitor the auto-correlation plots and effective sample sizes of the iterates, and they are found to be satisfactorily uncorrelated. In all of our simulations, we set $V=25$ nodes and $n=250$ samples.

We present analysis for $\nu=20$, $a_{\Delta}=b_{\Delta}=1$. For BNLC, there are two additional hyper-parameters $\iota$ and $\zeta$, both of which are set to 1. Note that the choice of  $a_{\Delta}=b_{\Delta}=1$ ensures that the prior on models is such that we have a uniform distribution on the number of active nodes, and conditional on the size of the model, a uniform distribution on all possible models of that size.
 The choice of $\nu=20$ ensures that the prior distribution of $\bM$ is concentrated around a scaled identity matrix. Since model is invariant to rotations of the latent positions, so we want the prior on $\bu_k$'s to also be invariant under rotation.  That requires that we center $\bM$ around a matrix that is proportional to the identity. Our choice of $\iota$ and $\zeta$ set the prior mean of $s_{k,l}$ at $0.5$ which is the suggested prior mean for the local parameters proposed in \cite{park2008bayesian}.
Sensitivity to the choice of hyper-parameters is discussed later, both for simulation studies and for the real data analysis.

\subsection{Identification of Influential Nodes}
Figures \ref{node_select_B_sim1} and \ref{node_select_B_sim2} show the posterior probability of the $k$-th node being detected as
influential, i.e., $P(\xi_k=1|Data)$, by BNLC and BNHC for each node and each case within \emph{Simulations 1} and \emph{2}, respectively. Some interesting observations
emerge from the results. 
We find that both methods work well with lower node sparsity and higher residual edge sparsity.  Decreasing the residual edge sparsity and increasing the node sparsity have adverse effects on the performance. 
In general, BNLC shows relatively better performance than BNHC in cases with higher node sparsity and/or lower residual edge sparsity. We provide a brief discussion below to support these observations.

For BNHC, case 2 exhibits a few false positives, and the separation of posterior probabilities for truly active and truly inactive nodes is much more stark in case 1 than in case 2. BNLC does a better job of node identification than BNHC in case 2.
Residual edge effect does have an impact on the probabilities, which is evident by comparing cases 1 and 3. For BNHC, case 3 (Simulation 1) displays poor performance with a higher number of both false positives and false negatives. Performance of BNLC appears to be better than BNHC in case 3. Fixing the residual edge sparsity and increasing the node sparsity has a negative impact on node identification, as seen by comparing performances in cases 3 and 4 (Simulation 1). For Simulation 2, both competitors perform quite well in cases 1 a nd 2. Again, case 3 (Simulation 2) represents a higher node sparsity, so that both BNHC and BNLC do not perform well in this case. Similar to Simulation 1, BNHC shows inferior performance to BNLC in case 3. While BNHC offers a few false positives and false negatives in case 4 (Simulation 2), the performance appears to be much better than in case 3. Notice that case 3 has both higher node sparsity and residual edge sparsity than case 4. While they have opposing effects, it appears that higher node sparsity demonstrates more of an adverse effect here compared to a small perturbation in the residual edge sparsity.  Recall that \cite{relion2017network} is the only other competitor which is designed to detect influential nodes. It detects all nodes to be influential in all simulation cases.


\begin{figure}[!ht]
   \begin{center}
   \subfigure[BNLC]{\includegraphics[width=7 cm, height = 7 cm]{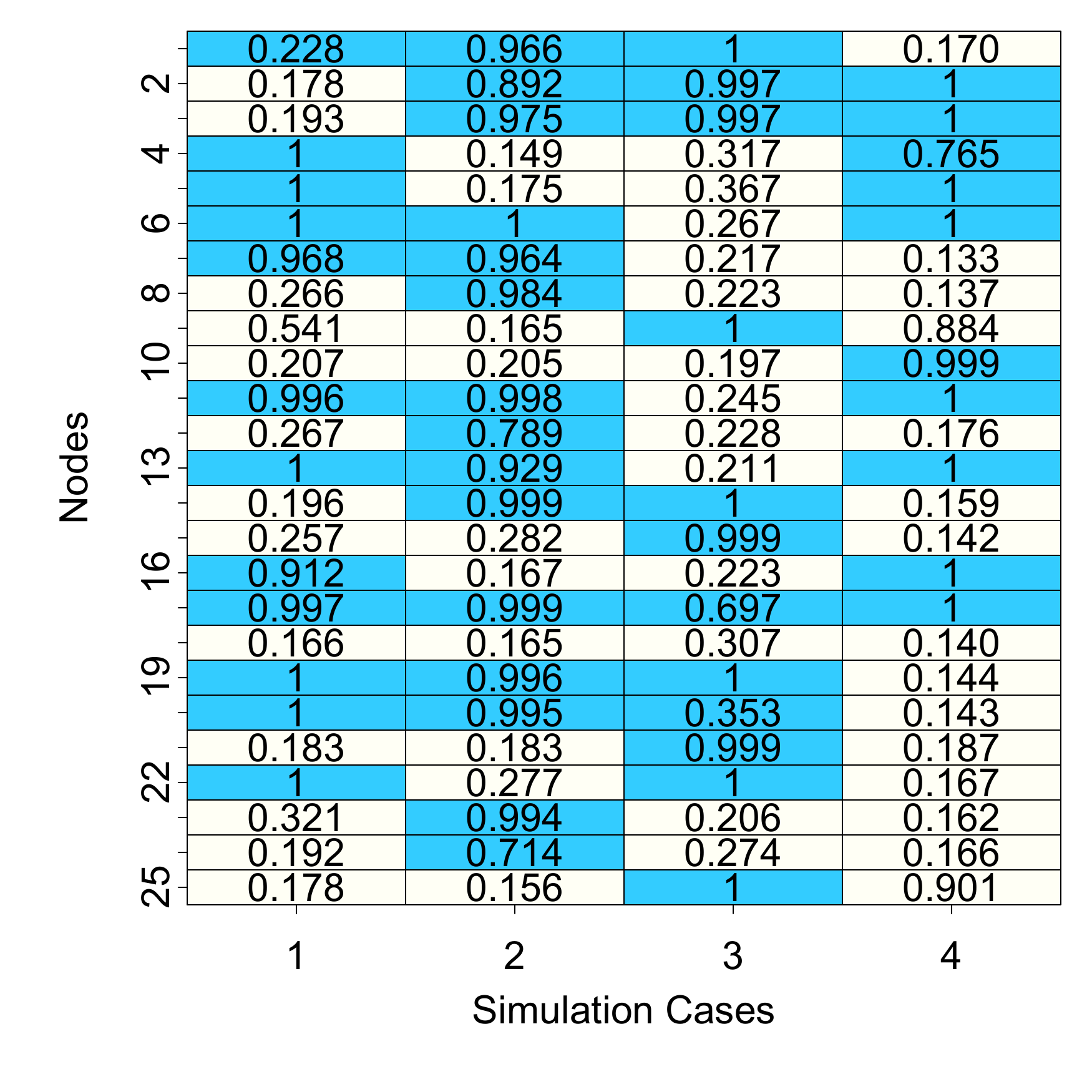}\label{sim_node_bnlc}}
   \subfigure[BNHC]{\includegraphics[width=7 cm, height = 7 cm]{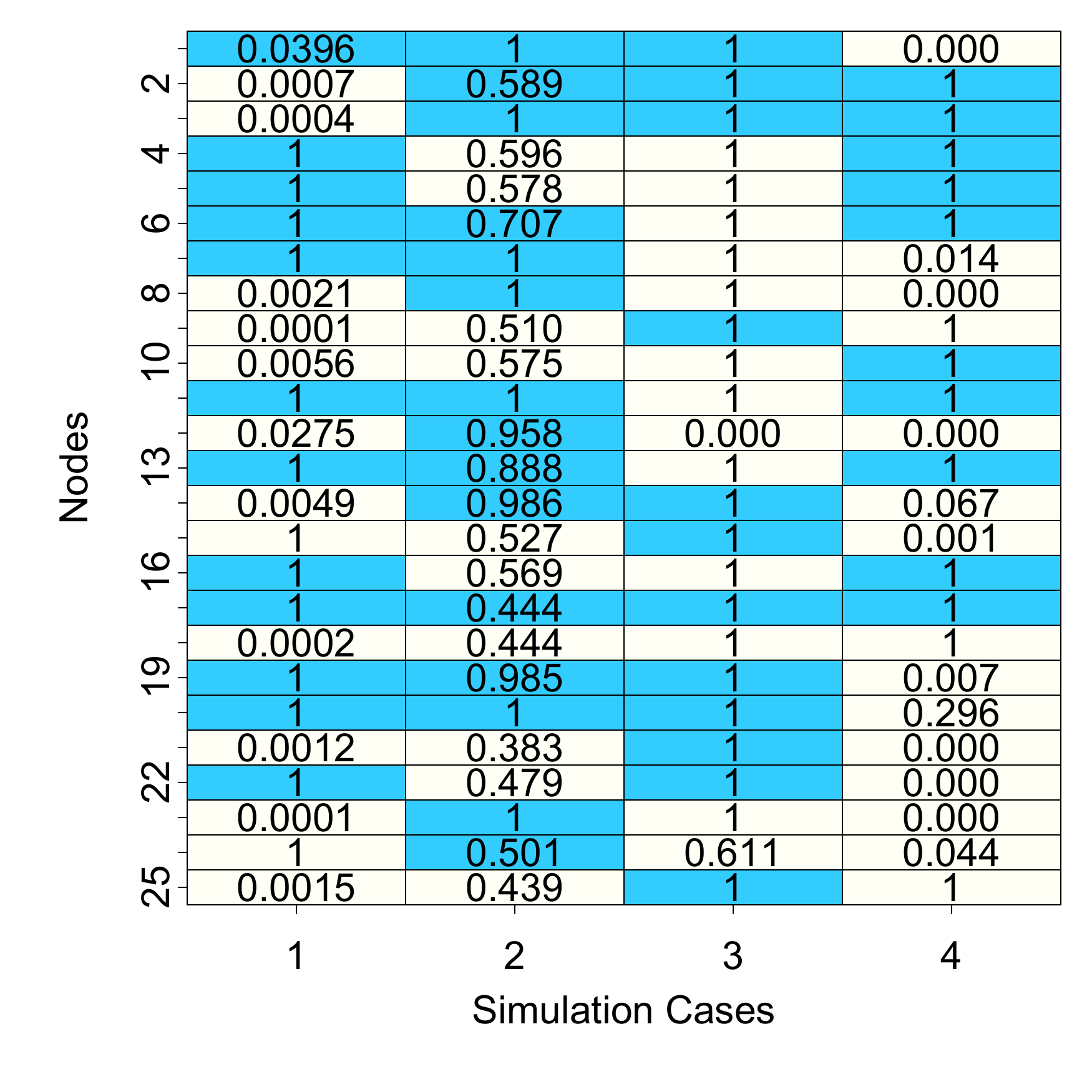}\label{sim_node_hs}}
    \end{center}
\caption{Simulation 1: clear background denotes \emph{uninfluential} and dark background denotes \emph{influential} nodes in the truth for BNLC and BNHC models. Note that there are $25$ rows (corresponding to $25$ nodes) and $4$ columns corresponding to $4$ different cases in Simulation 1. The model-detected posterior probability of being influential has been super-imposed onto the corresponding node.}\label{node_select_B_sim1}
\end{figure}

\begin{figure}[!ht]
   \begin{center}
   \subfigure[BNLC]{\includegraphics[width=7 cm, height = 7 cm]{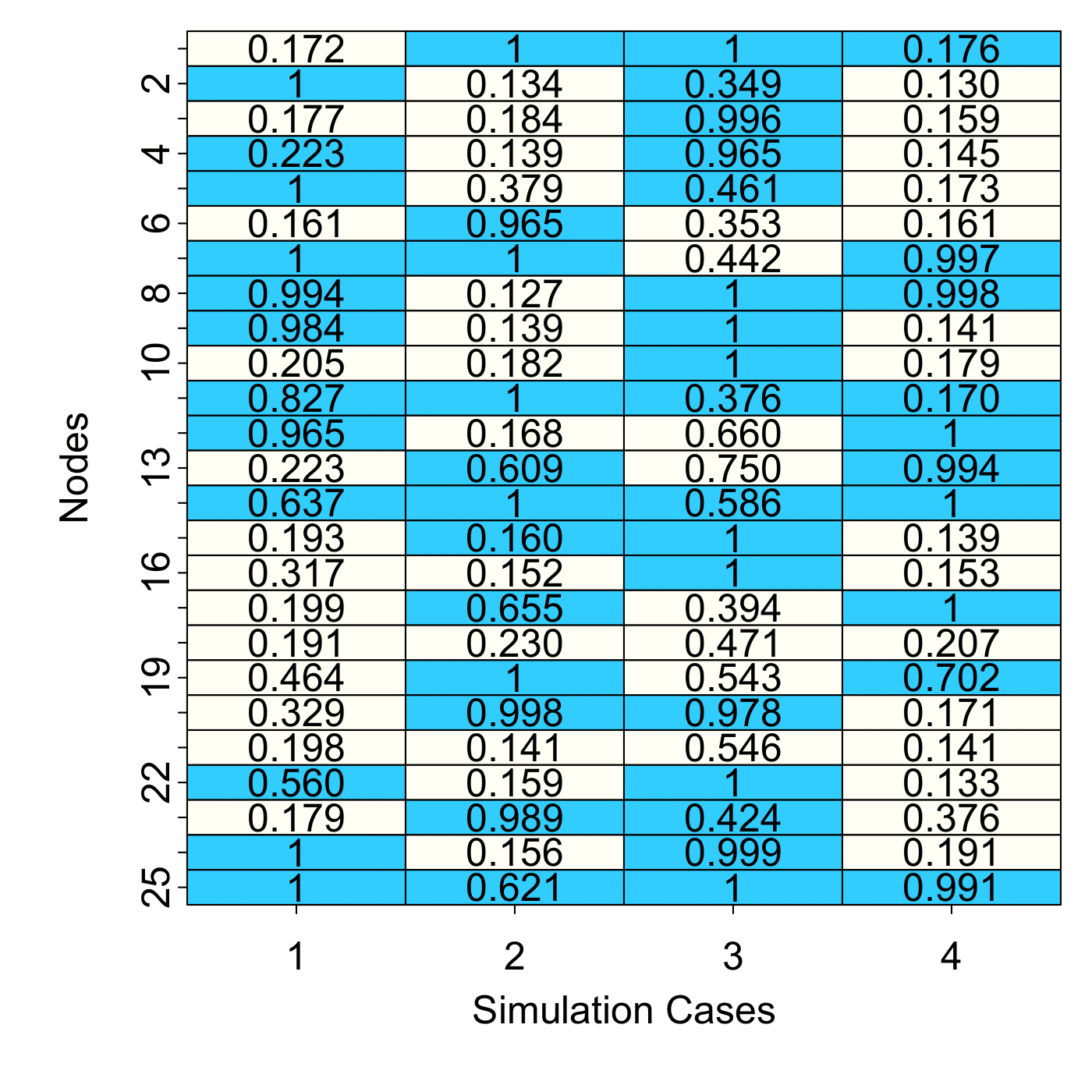}\label{sim2_node_bnlc}}
   \subfigure[BNHC]{\includegraphics[width=7 cm, height = 7 cm]{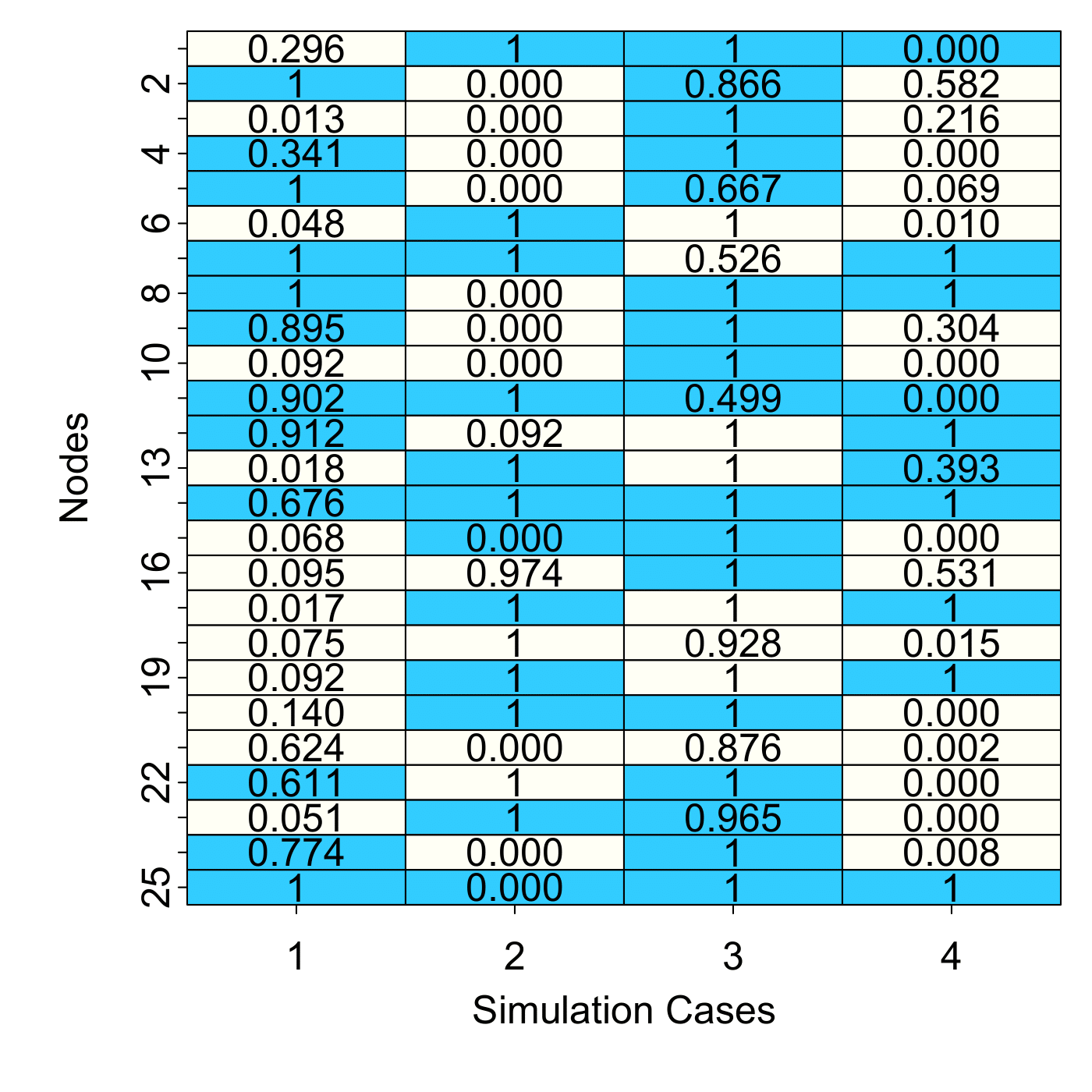}\label{sim2_node_hs}}
    \end{center}
\caption{Simulation 2: clear background denotes \emph{uninfluential} and dark background denotes \emph{influential} nodes in the truth for BNLC and BNHC models. Note that there are $25$ rows (corresponding to $25$ nodes) and $4$ columns corresponding to $4$ different cases in Simulation 2. The model-detected posterior probability of being influential has been super-imposed onto the corresponding node.}\label{node_select_B_sim2}
\end{figure}

\subsection{Identification of Influential Edges}
We apply the algorithm with a mixture of skewed t-distributions described in Appendix C to detect influential edges from the post burn-in MCMC samples of the edge coefficients using a threshold of $t=0.05$. The proposed approach controls FDR below a threshold of $0.05$ to account for multiplicity correction. Tables \ref{Tab_TPR_Edges_sim3_proj2} and \ref{Tab_TPR_Edges_sim3} provide the true positive rates (TPR) and false positive rates (FPR) in detecting important edges for Simulations 1 and 2 for the competitors, respectively. It is observed that when node sparsity is moderate and residual edge sparsity is high (cases 1 and 2), both BNLC and BNHC offer moderate performance in terms of identifying true positives, and include very few false positives. In these cases, BNHC generally exhibits a little higher FPR than BNLC. 
In the case of high node sparsity (e.g., case 3, Simulation 2) both these methods unfortunately show much lower true positive rates. Again, lower edge sparsity (case 3, Simulation 1) has almost no effect on FPR of BNLC, but decreases TPR substantially. For BNHC, both TPR and FPR increase when residual edge sparsity is reduced. Nevertheless, both of them perform significantly better than Lasso in almost all cases. The competitor in \cite{relion2017network} appears to have suboptimal performance, as it identifies all edges as important in all the simulation scenarios, resulting in high FPRs.
\begin{table}[!th]
\begin{center}
\begin{tabular}
[c]{c|cc|cc|cc|cc}
\hline
\multicolumn{1}{c|}{} & \multicolumn{2}{c|}{BNLC} & \multicolumn{2}{c|}{BNHC}  & \multicolumn{2}{c|}{Lasso}  & \multicolumn{2}{c}{Reli\'{o}n (2017)} \\
\hline
Cases & TPR & FPR & TPR & FPR & TPR & FPR & TPR &  FPR\\
\hline
Case - 1 & 0.65 & 0.01 & 0.72 & 0.12 & 0.50 & 0.22 & 1 & 1\\
Case - 2 & 0.64 & 0.00 & 0.63 & 0.02 & 0.40 & 0.14 & 1 & 1\\
Case - 3 & 0.45 & 0.00 & 0.86 & 0.40 & 0.42 & 0.22 & 1 & 1\\
Case - 4 & 0.72 & 0.09 & 0.70 & 0.12  & 0.54 & 0.16 & 1 & 1\\
\hline
\end{tabular}
\caption{True Positive Rates (TPR) and False Positive Rates (FPR) for edges for cases in \emph{Simulation 1}.}\label{Tab_TPR_Edges_sim3_proj2}
\end{center}
\end{table}

\begin{table}[!th]
\begin{center}

\begin{tabular}
[c]{c|cc|cc|cc|cc}
\hline
\multicolumn{1}{c|}{} & \multicolumn{2}{c|}{ BNLC} & \multicolumn{2}{c|}{ BNHC}  & \multicolumn{2}{c|}{Lasso}  & \multicolumn{2}{c}{Reli\'{o}n(2017)} \\
\hline
Cases & TPR & FPR & TPR & FPR & TPR & FPR & TPR &  FPR\\
\hline
Case - 1 & 0.63 & 0.00 & 0.84 & 0.08 & 0.44 & 0.20 & 1 & 1\\
Case - 2 & 0.56 & 0.00 & 0.63 & 0.12  & 0.53 & 0.22 & 1 & 1\\
Case - 3 & 0.46 & 0.02 & 0.59 & 0.08 &  0.31 & 0.16 & 1 & 1\\
Case - 4 & 0.68 & 0.03 & 0.75 & 0.06 & 0.34 & 0.12 & 1 & 1\\
\hline
\end{tabular}
\caption{True Positive Rates (TPR) and False Positive Rates (FPR) for edges for cases in \emph{Simulation 2}.}\label{Tab_TPR_Edges_sim3}
\end{center}
\end{table}

The results in Tables~\ref{Tab_TPR_Edges_sim3_proj2} and \ref{Tab_TPR_Edges_sim3} indicate higher number of edges identified as influential by BNHC than BNLC in all simulations. Digging a bit deeper, we report 
the ratio of the number of edges in the intersection of both methods to the number of  total edges identified by each method independently in Table~\ref{Tab_edge_incisive}.  In all simulation cases, almost all edges identified as influential by BNLC are also identified as influential by BNHC. In cases 2 and 4 (Simulation 1), the fractions corresponding to BNLC and BNHC are very similar, indicating similar edge identification by both of them. However, this fraction appears to be lower in BNHC for cases 1 and 3 (Simulation 1). This again shows that the edges identified by BNLC are also identified by BNHC, with BNHC identifying more edges. The discrepancy turns out to be more in case 3 (Simulation 1) where BNHC has identified many more edges. Simulation 2 shows a similar trend. We further track the top 10, 20 and 30 edges identified from BNLC and record how many of these edges belong to the top 10, 20 and 30 edges identified from BNHC. Table~\ref{Tab_edge_incisive} shows a high level of intersection among the top edges identified by these two methods.

A number of interesting observations emerge from the analysis. First of all, as mentioned earlier, the edges identified by BNLC are generally also identified by BNHC. BNHC tends to identify more edges, leading to higher TPR and FPR. Broadly, in presence of higher node sparsity, the discrepancy is greater, with BNHC having much higher TPR and FPR. Interestingly, the absolute values of the edge coefficients follow very similar rankings for BNHC and BNLC, which leads to high intersections among the top edges selected by these methods. Perhaps the difference in shrinkage mechanism imposed by BNHC and BNLC is responsible for their difference in tail behavior, leading to differences in edge selection. 

\begin{table}[!th]
\begin{center}

\begin{tabular}
[c]{c|ccccc|ccccc}
\hline
\multicolumn{1}{c|}{} & \multicolumn{5}{c|}{Simulation 1} & \multicolumn{5}{c}{Simulation 2} \\
\hline
Cases & $\frac{N_{BL,BH}}{N_{BL}}$ & $\frac{N_{BL,BH}}{N_{BH}}$ & \multicolumn{3}{c|}{Top} & $\frac{N_{BL,BH}}{N_{BL}}$ & $\frac{N_{BL,BH}}{N_{BH}}$ &  \multicolumn{3}{c}{Top}\\
           &                                               &                                               &    10    &    20     &   30       &                                              &                                               &    10     &     20    &  30  \\
\hline
1 & 0.94 & 0.61 & 9 & 19 & 27 & 1.00 & 0.58 & 7 & 17 & 26\\
2 & 0.85 & 0.83 & 8 & 14 & 21 & 1.00 & 0.46 & 8 & 17 & 26\\
3 & 1.00 & 0.25 & 9 & 13 & 24 & 0.97 & 0.70 & 9 & 18 & 28\\
4 & 0.91 & 0.87 & 8 & 18 & 27 & 0.91 & 0.75 & 8 & 17 & 27\\
\hline
\end{tabular}
\caption{$N_{BL,BH}$ represents the number of edges identified by both BNLC and BNHC. Similarly, $N_{BL}$ and $N_{BH}$ represent the number of edges identified by BNLC and BNHC, respectively. Top 10 represents the number of edges common among the top ten edges identified by BNLC and BNHC. Top 20 and Top30 are defined analogously.}\label{Tab_edge_incisive}
\end{center}
\end{table}

\subsection{Estimation of Edge Coefficients and Classification Accuracy}
The mean squared errors (MSE) associated with the point estimation of edge coefficients for different competitors are presented in Tables~\ref{Tab_sim1_B} and \ref{Tab_sim2_B}, corresponding to Simulations 1 and 2, respectively. For the Bayesian competitors, point estimates are computed using the posterior means of the edge coefficients. In all cases, BNLC and BNHC consistently outperform all other competitors, with the binary Bayesian Lasso exhibiting the next best performance. In all simulation cases, BNLC comprehensively outperforms BNHC in terms of estimating edge coefficients. Consistent with earlier observations, both competitors tend to be less accurate when node sparsity increases. Figure~\ref{AUC_sim12_BB} records AUC for all competitors in Simulations 1 and 2. 
In almost all cases, AUC for BNHC and BNLC turn out to be higher than other competitors. On the other hand, \cite{relion2017network} appears to have close to random classification of samples with AUC around $0.5$.

\begin{figure}[!ht]
  \begin{center}
  \subfigure[Simulation 1]{\includegraphics[width=7cm,height=7cm]{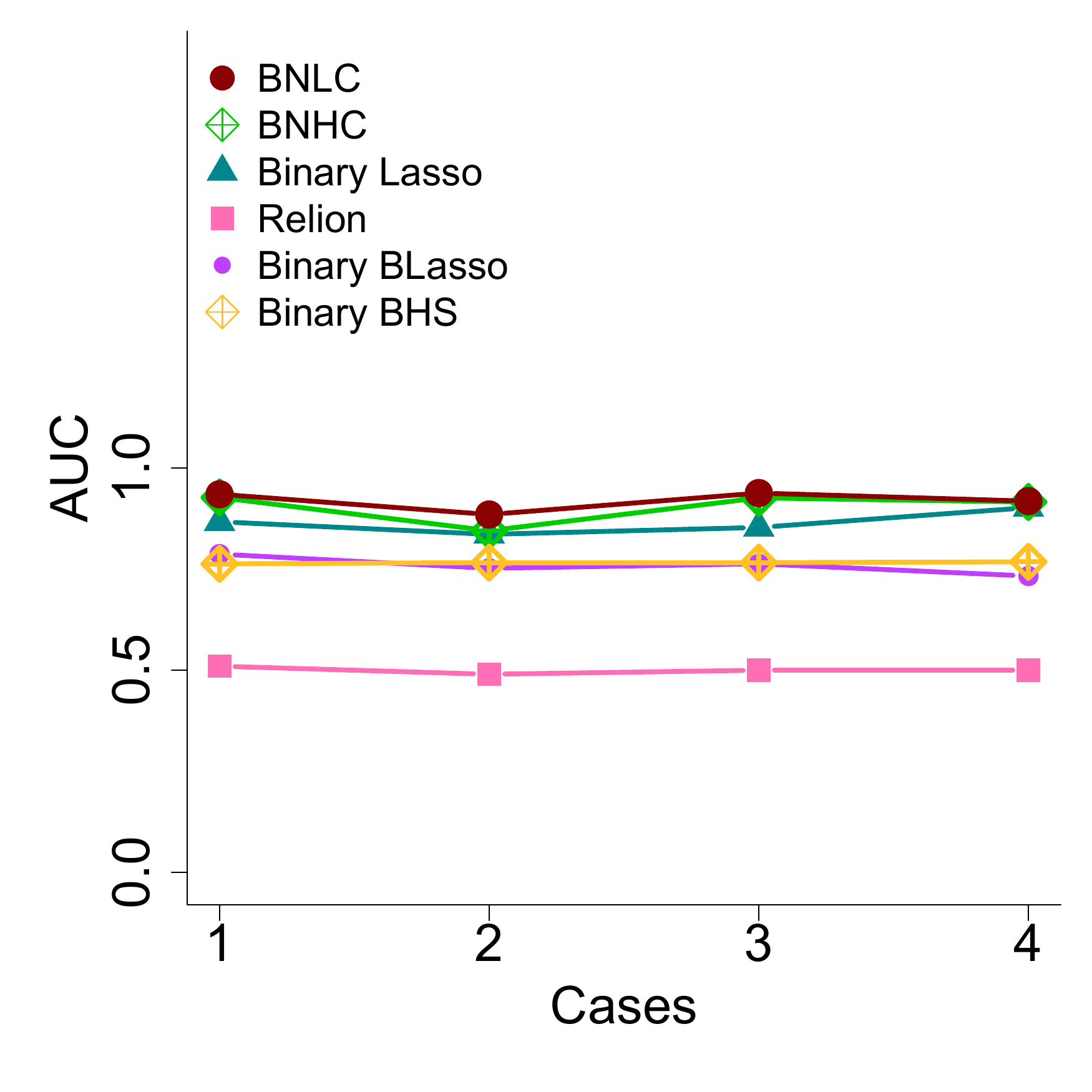}\label{sim1-auc}}
  \subfigure[Simulation 2]{\includegraphics[width=7cm,height=7cm]{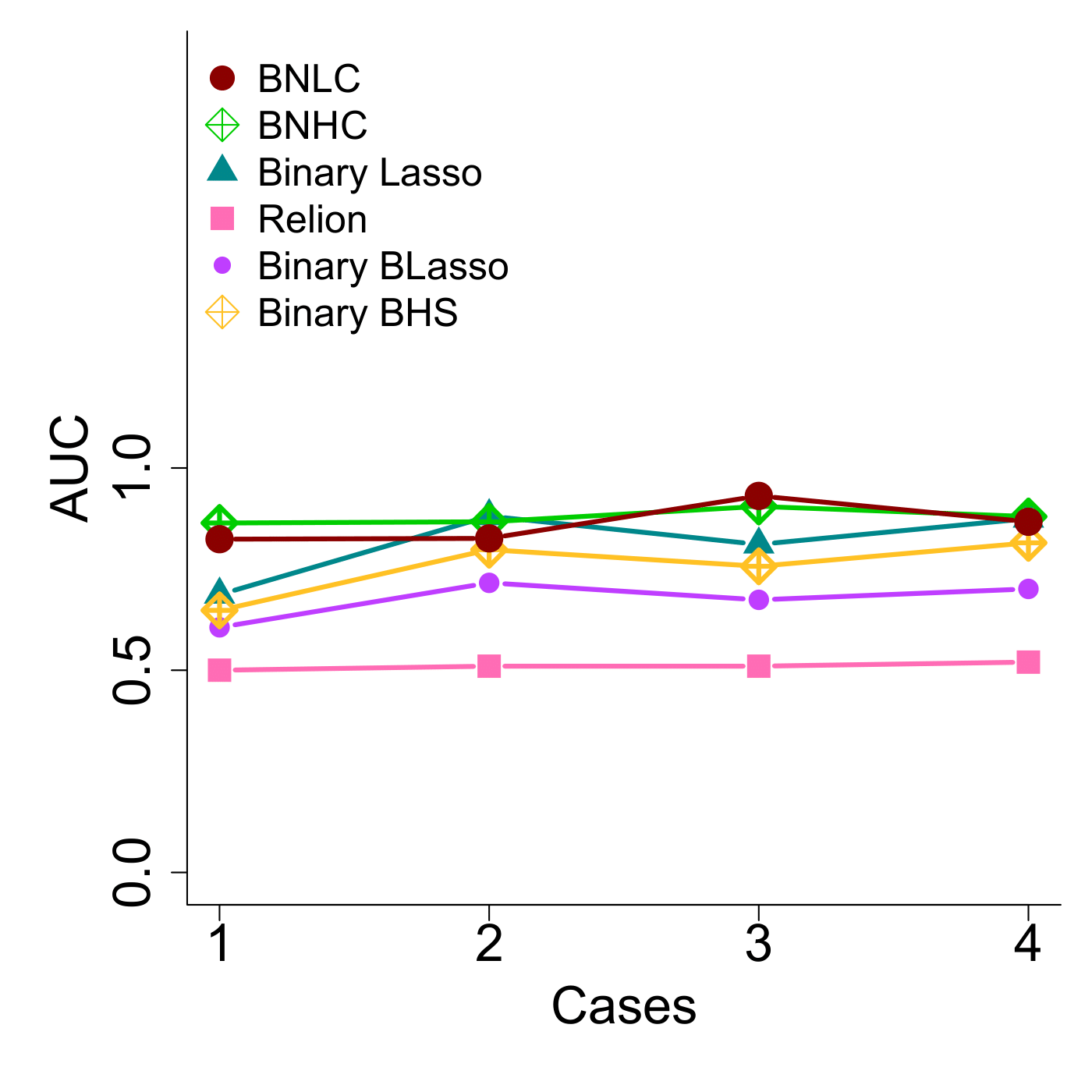}\label{sim2-auc}}
 \end{center}
 \caption{Figure shows classification performance in the form of Area under Curve (AUC) of ROC for all cases in Simulations 1 and 2.} \label{AUC_sim12_BB}
\end{figure}

\begin{figure}
  \begin{center}
    \subfigure[Case 1, BNLC]{\includegraphics[width=4.4 cm,height=4.4cm]{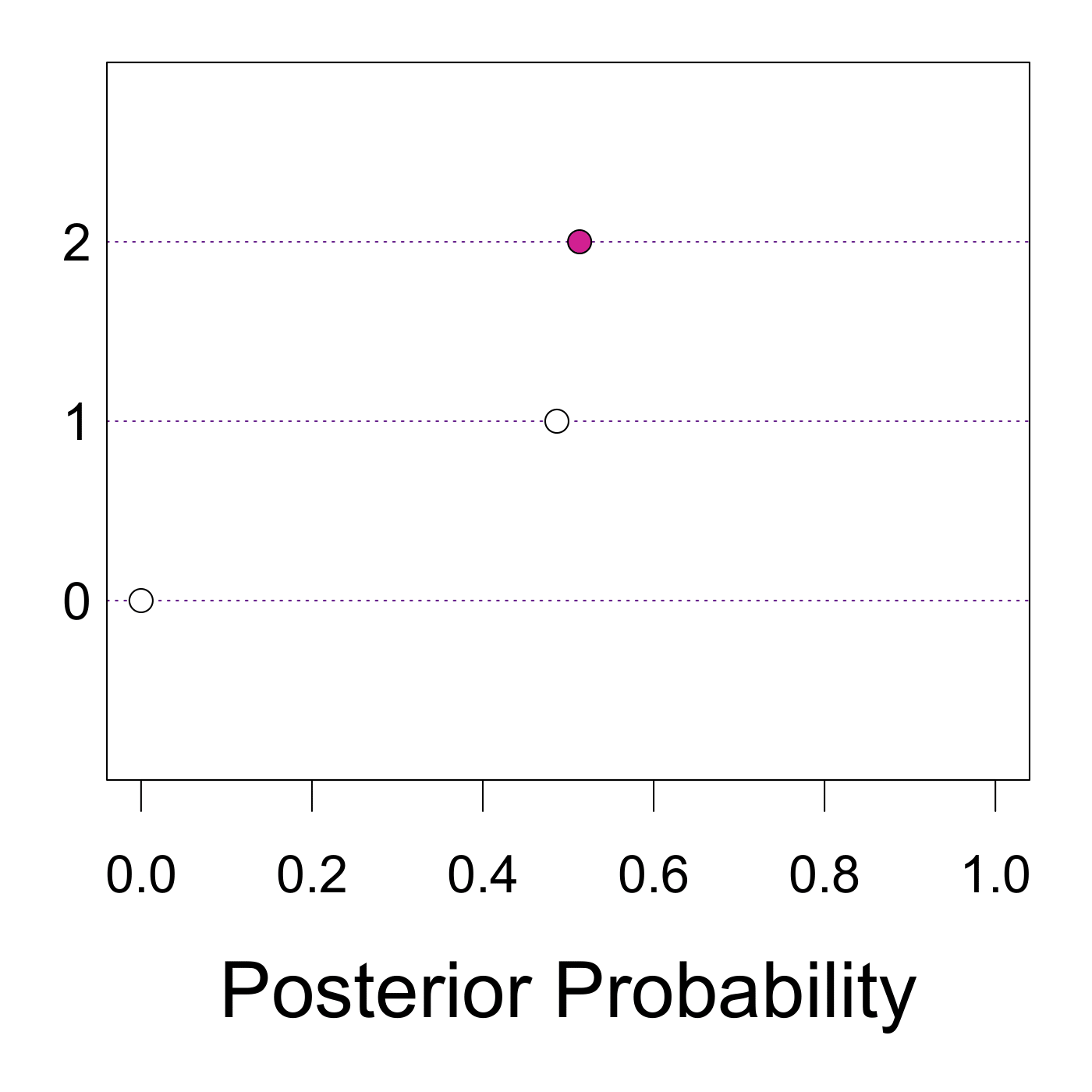}\label{sim1-dc1}}
   \subfigure[Case 2, BNLC]{\includegraphics[width=4.4 cm,height=4.4cm]{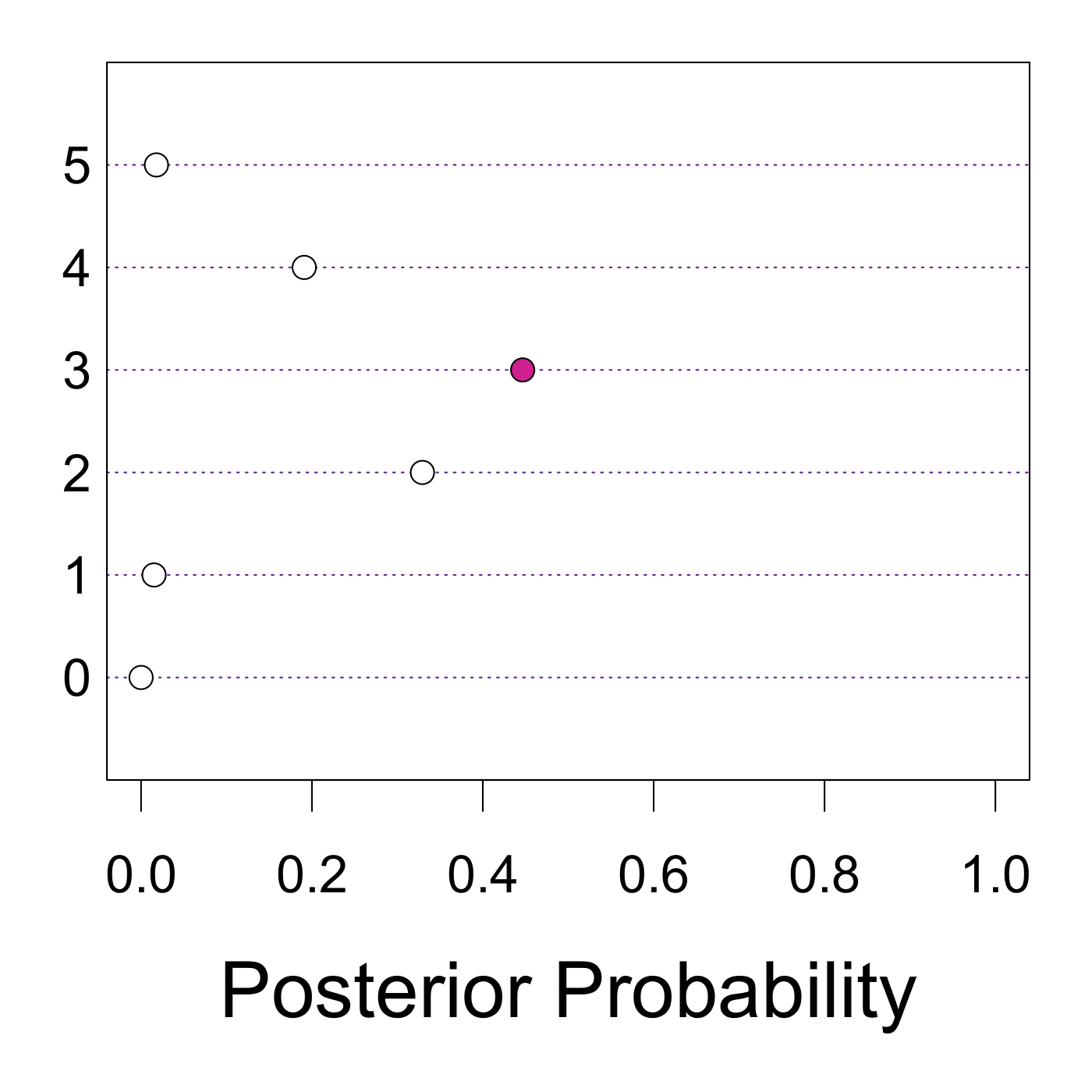}\label{sim1-dc5}}
   \subfigure[Case 3, BNLC]{\includegraphics[width=4.4 cm,height=4.4cm]{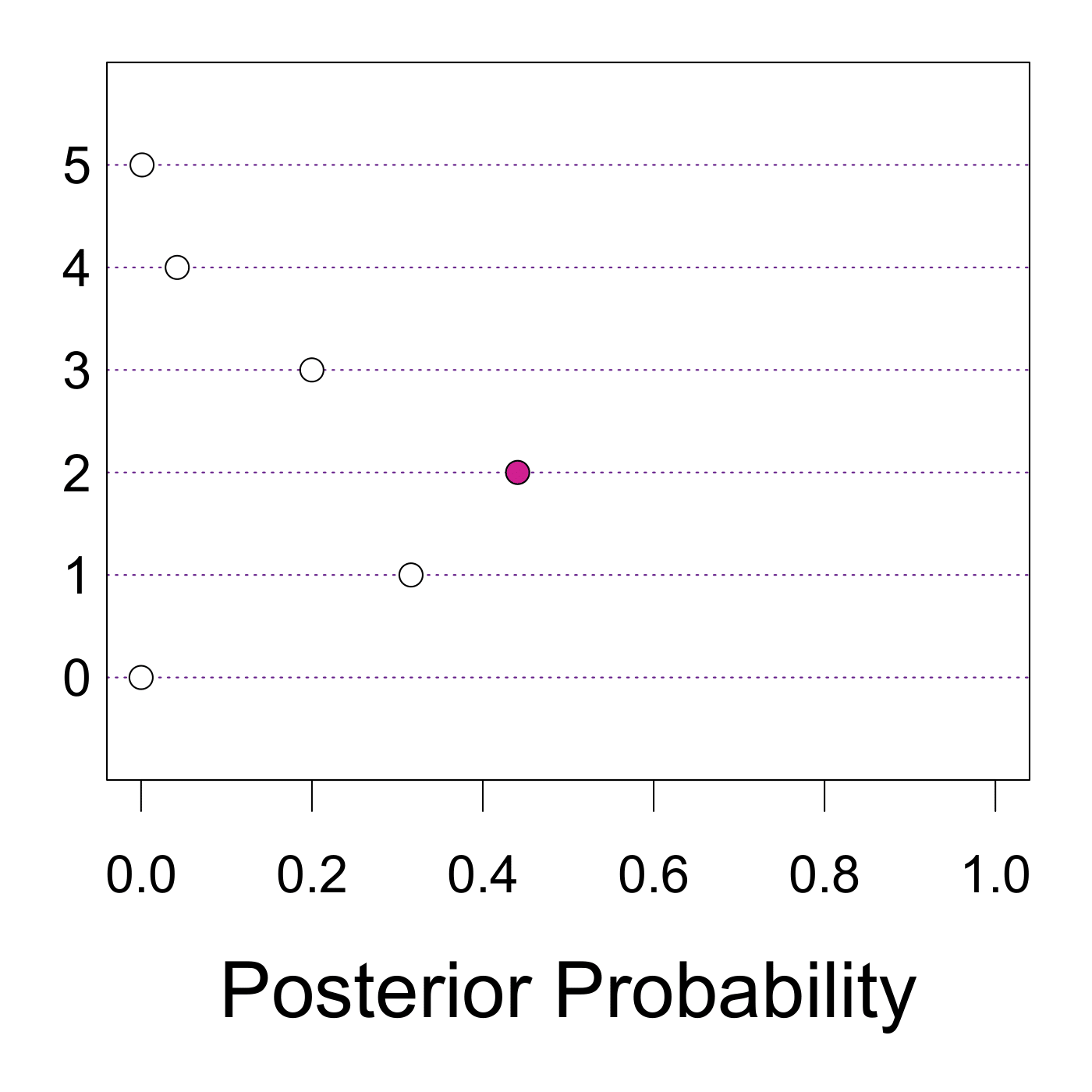}\label{sim1-dc6}}\\
   \subfigure[Case 4, BNLC]{\includegraphics[width=4.4 cm,height=4.4cm]{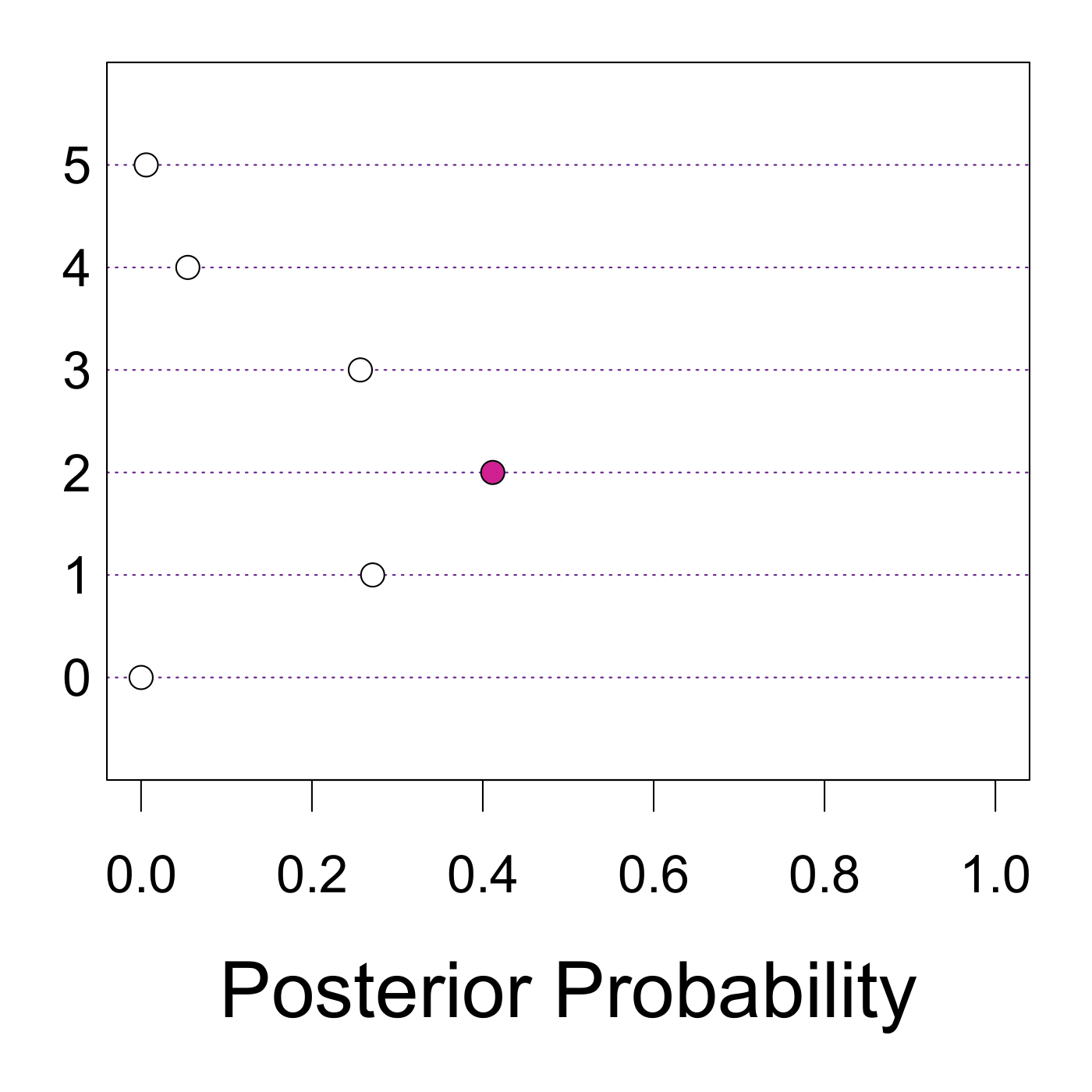}\label{sim1-dc7}}
   \subfigure[Case 1, BNHC]{\includegraphics[width=4.4 cm,height=4.4cm]{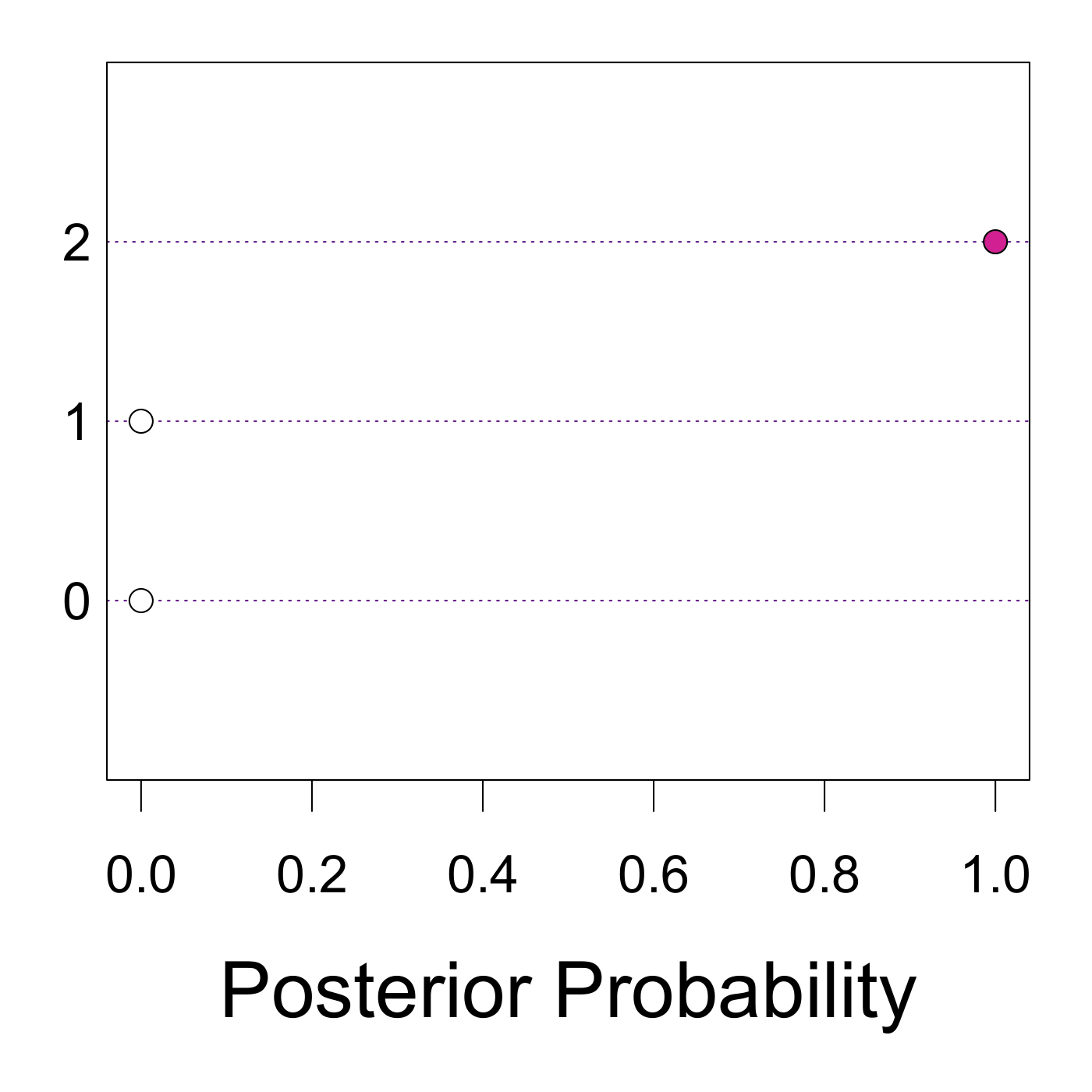}\label{sim1-dc1-hs}}
   \subfigure[Case 2, BNHC]{\includegraphics[width=4.4 cm,height=4.4cm]{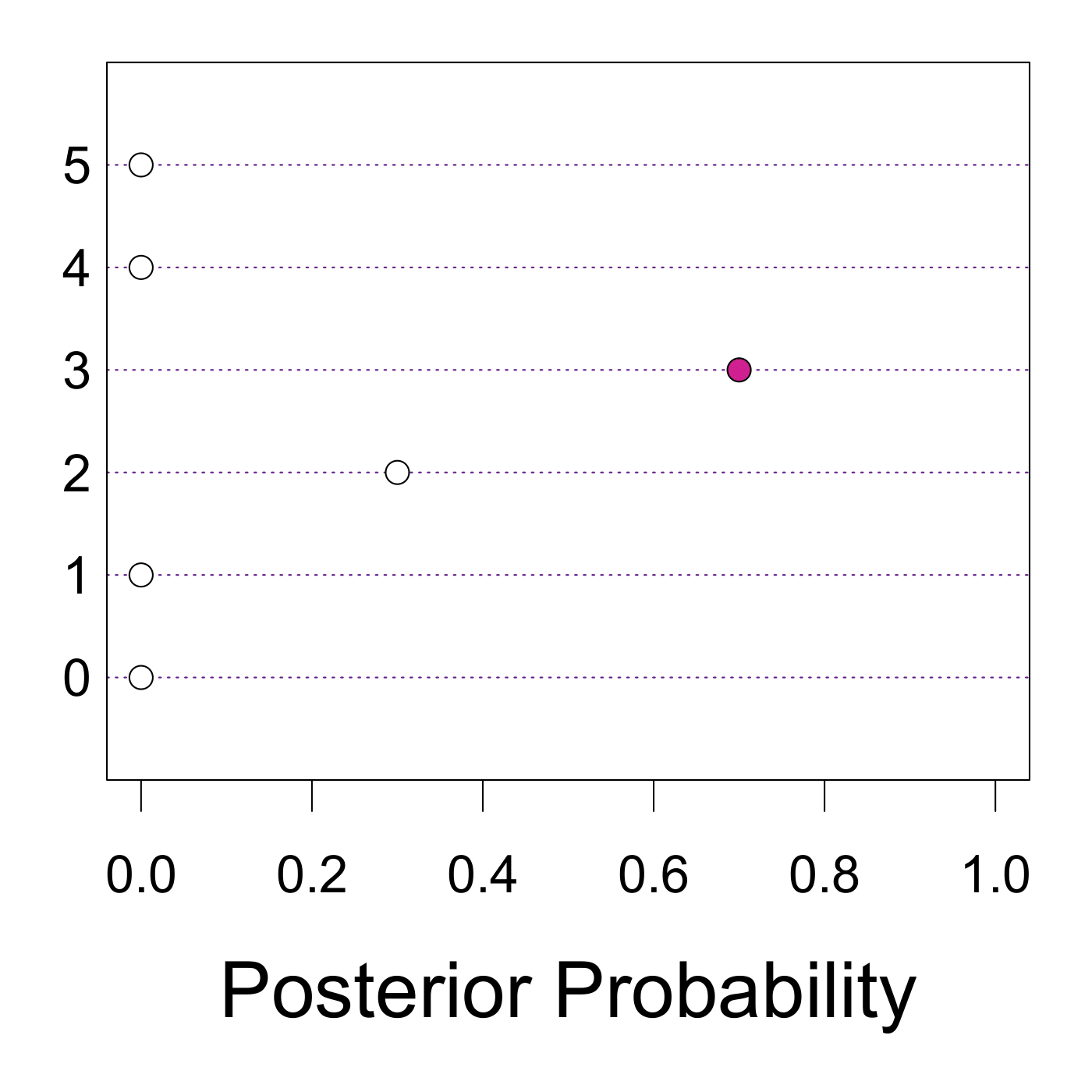}\label{sim1-dc5-hs}}\\
   \subfigure[Case 3, BNHC]{\includegraphics[width=4.4 cm,height=4.4cm]{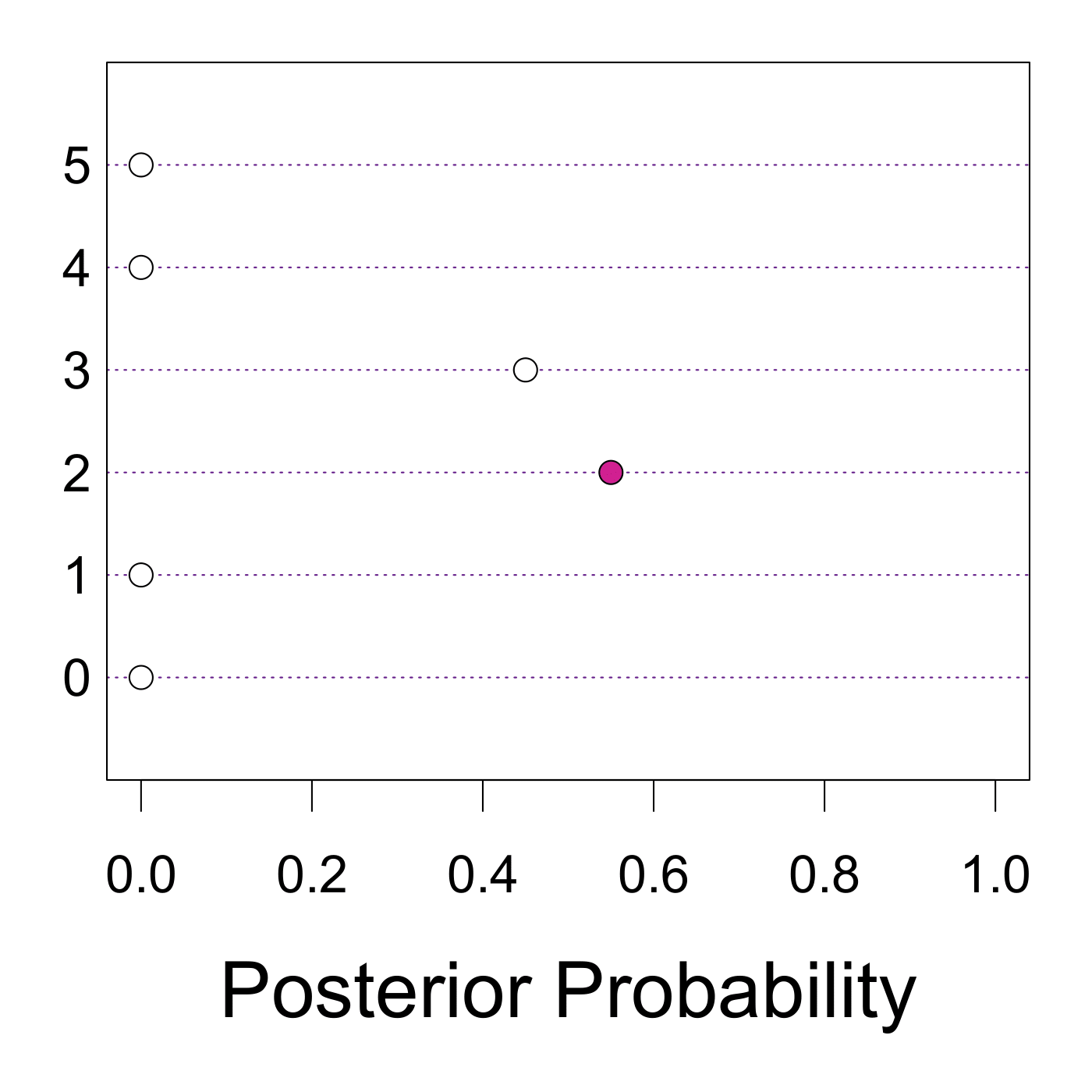}\label{sim1-dc6-hs}}
   \subfigure[Case 4, BNHC]{\includegraphics[width=4.4 cm,height=4.4cm]{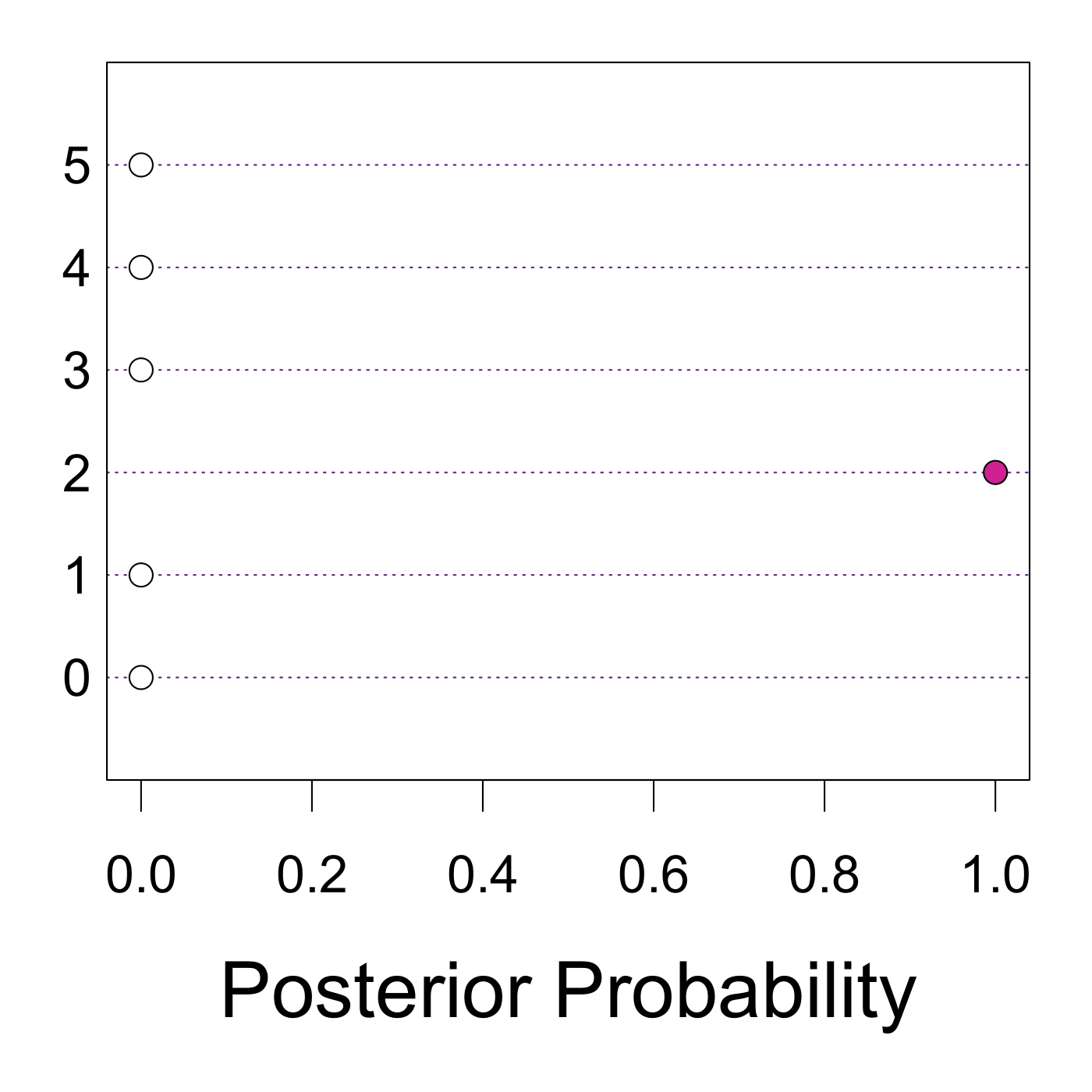}\label{sim1-dc7-hs}}
 \end{center}
 \caption{Plots showing posterior probability distribution of effective dimensionality for BNLC and BNHC models in all $4$ cases in Simulation 1. Filled bullets indicate the true value of effective dimensionality.}\label{effective}
\end{figure}

\begin{figure}
  \begin{center}
    \subfigure[Case 1, BNLC]{\includegraphics[width=4.4 cm,height=4.4cm]{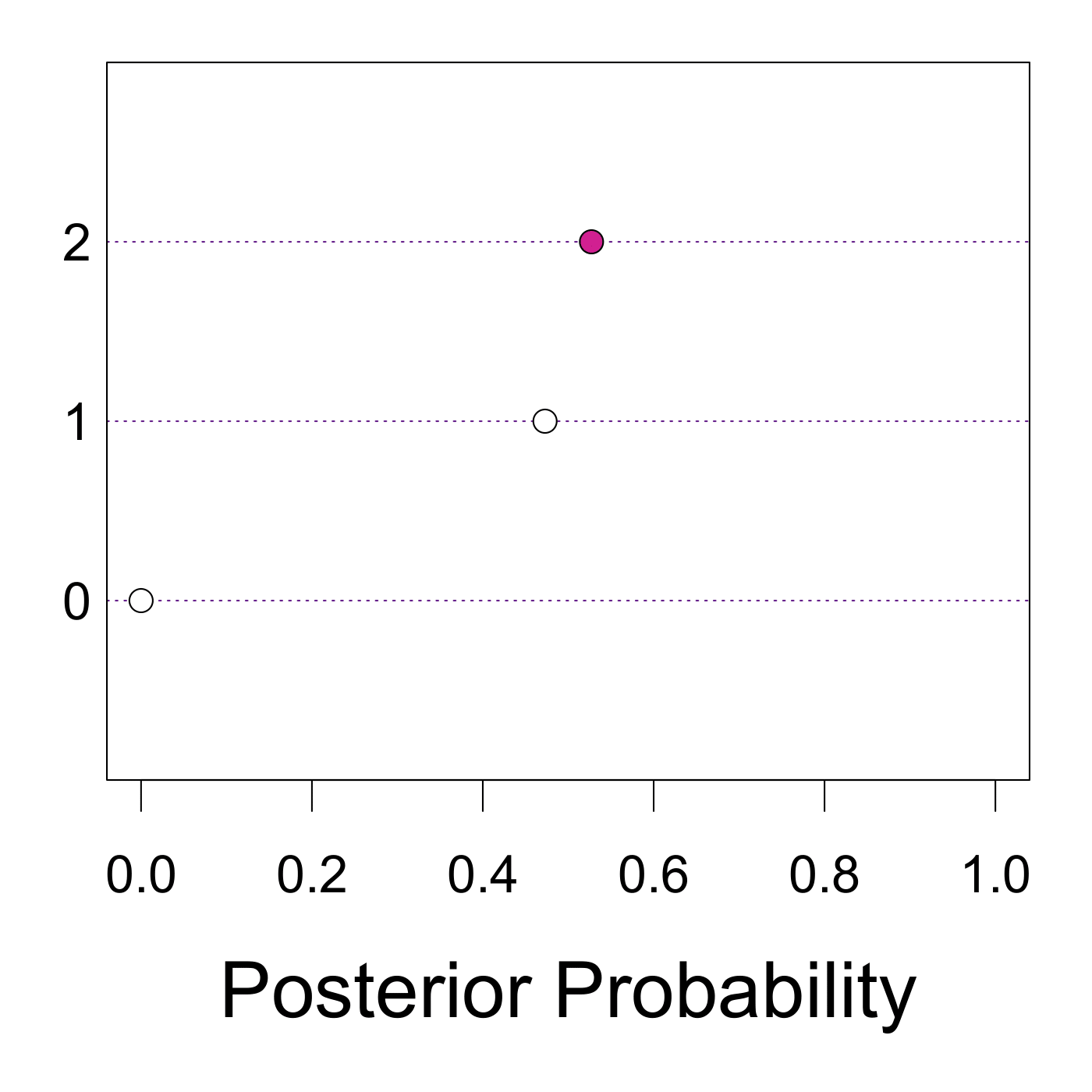}\label{sim2-dc1}}
   \subfigure[Case 2, BNLC]{\includegraphics[width=4.4cm,height=4.4cm]{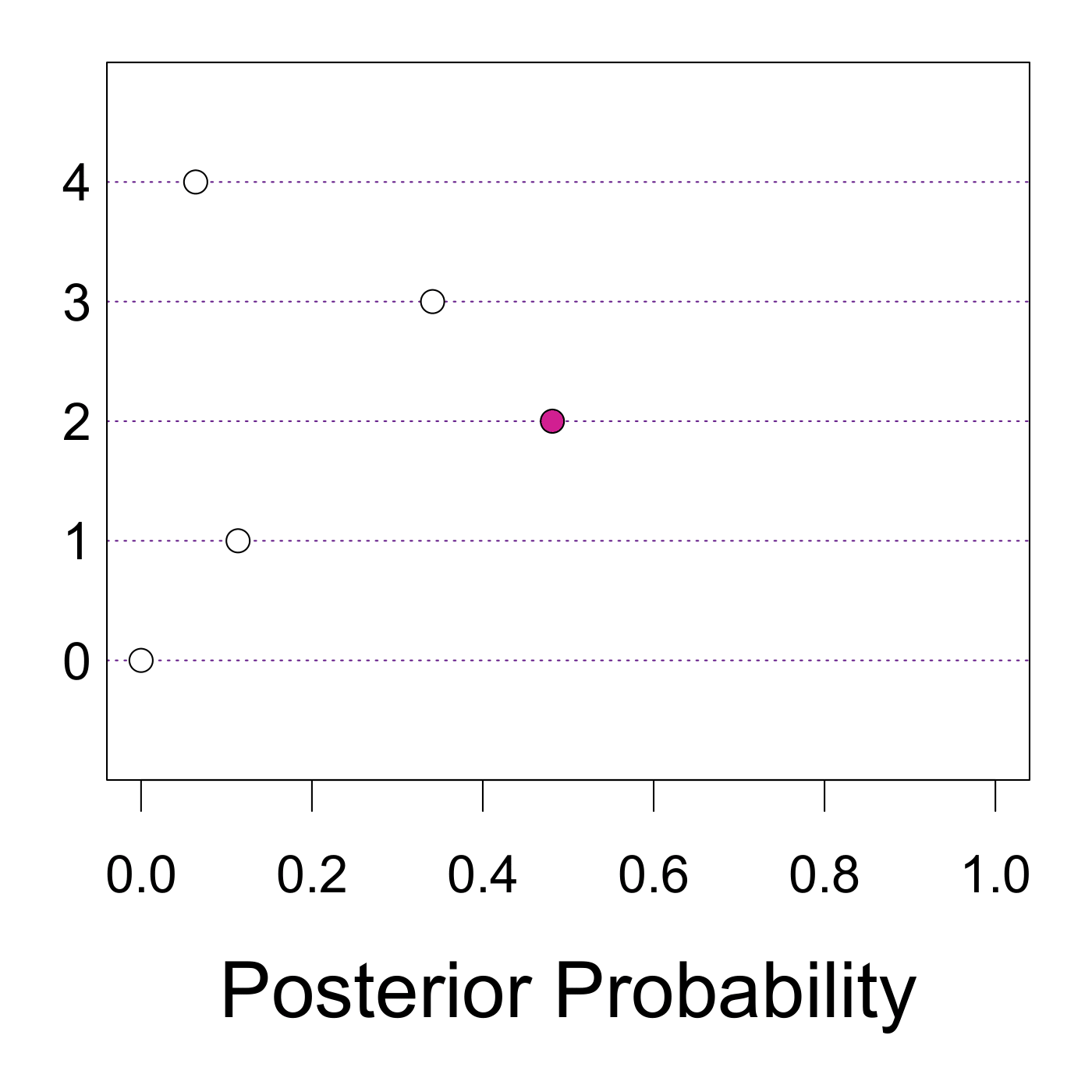}\label{sim2-dc2}}
   \subfigure[Case 3, BNLC]{\includegraphics[width=4.4 cm,height=4.4cm]{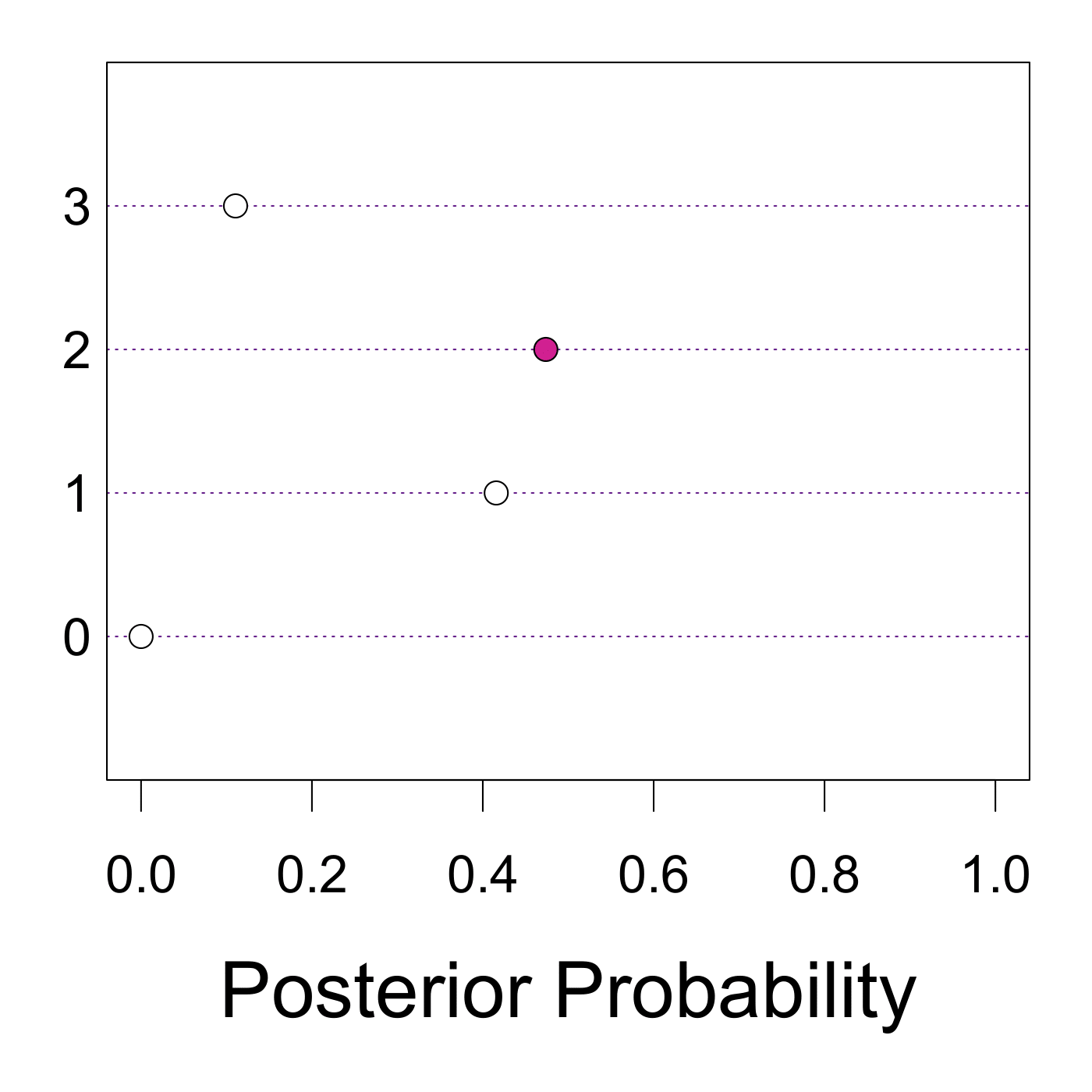}\label{sim2-dc3}}\\
   \subfigure[Case 4, BNLC]{\includegraphics[width=4.4 cm,height=4.4cm]{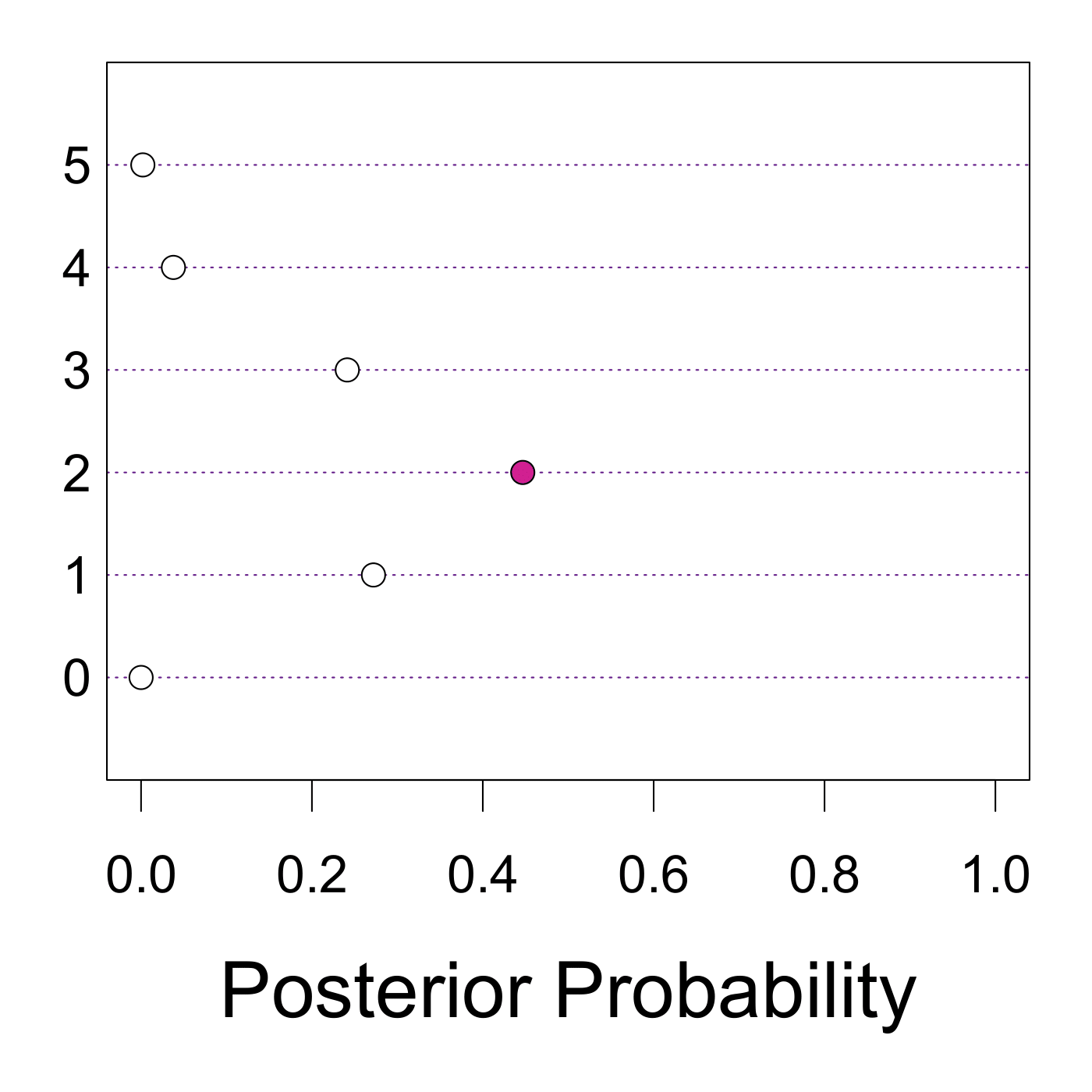}\label{sim2-dc4}}
   \subfigure[Case 1, BNHC]{\includegraphics[width=4.4 cm,height=4.4cm]{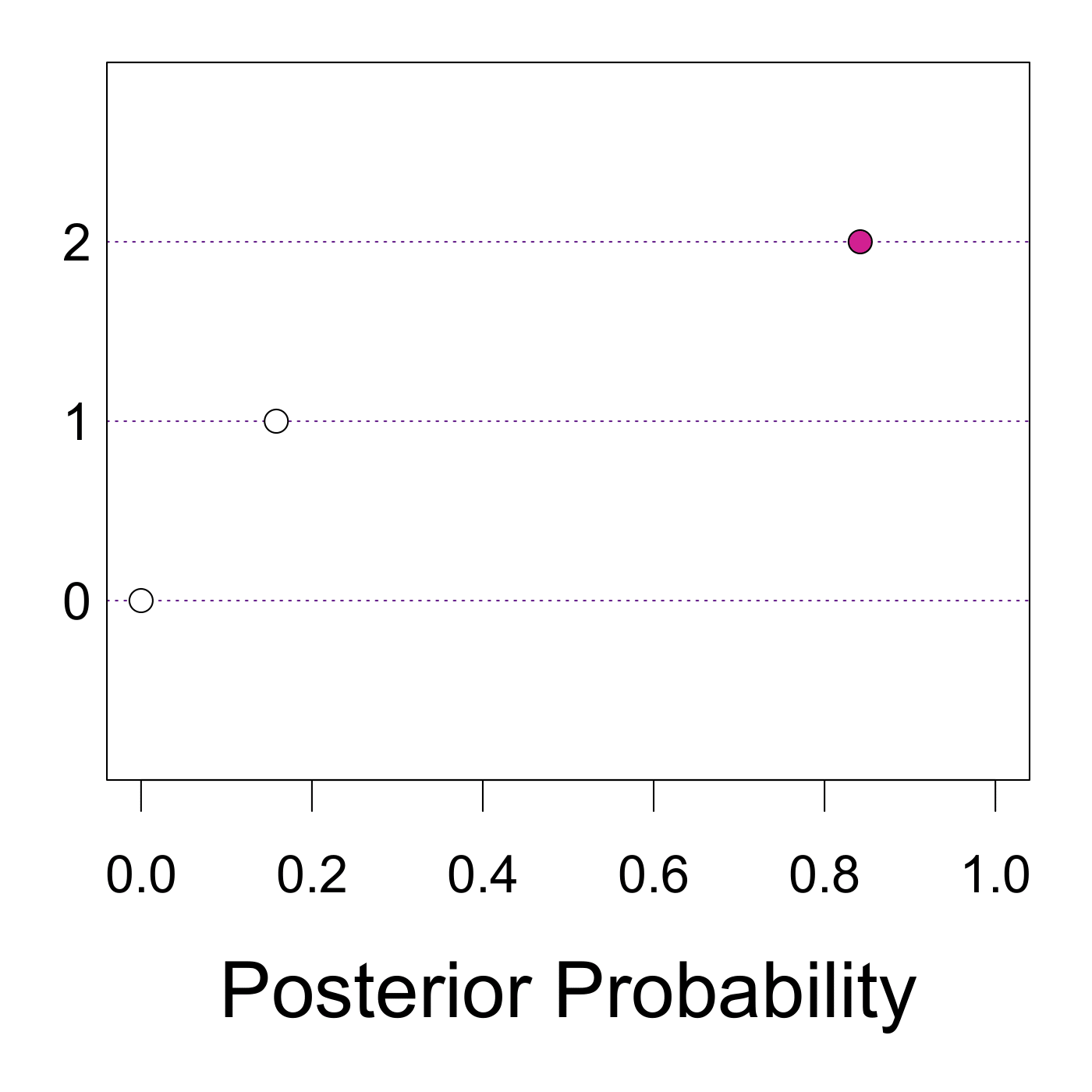}\label{sim2-dc1-hs}}
   \subfigure[Case 2, BNHC]{\includegraphics[width=4.4cm,height=4.4cm]{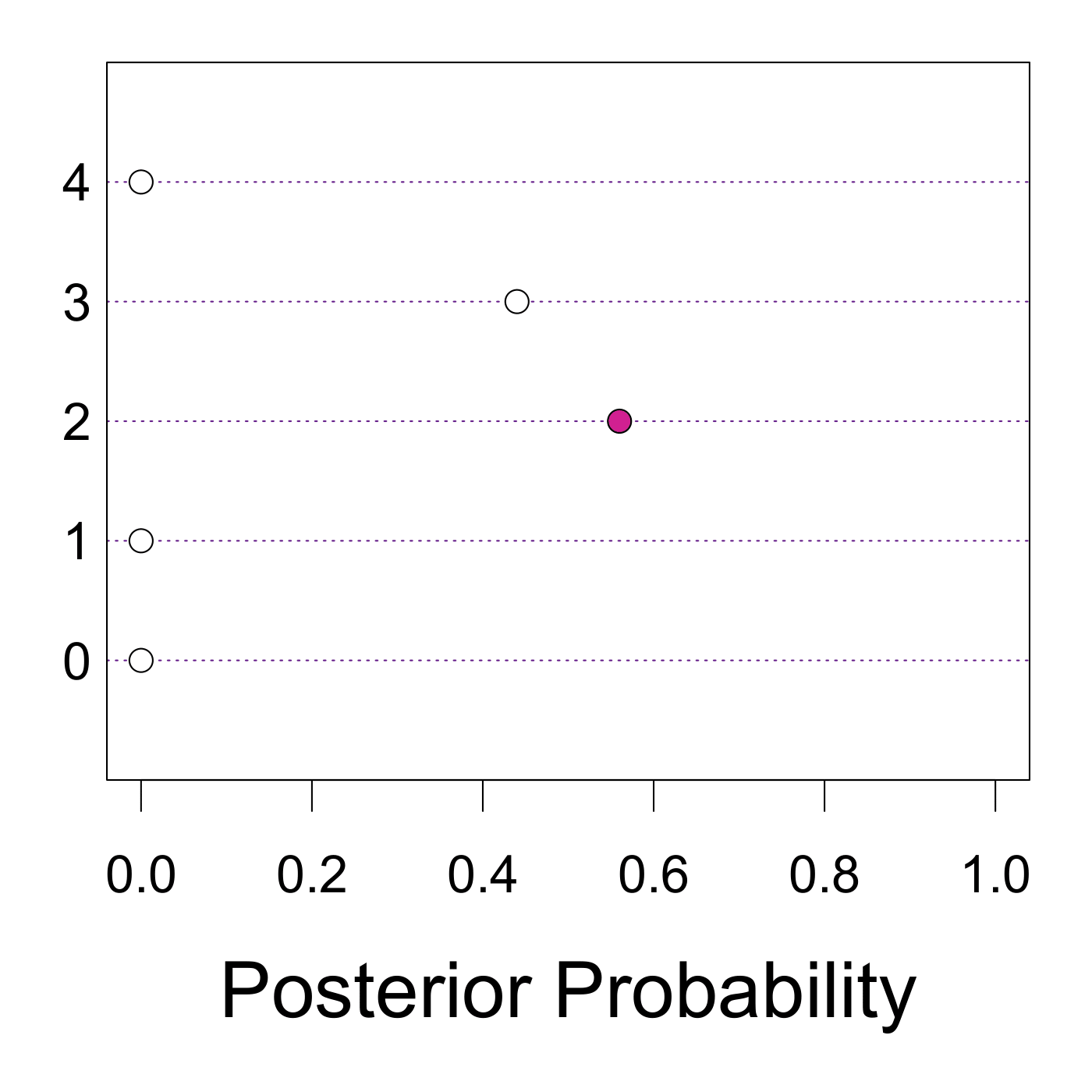}\label{sim2-dc2-hs}}\\
   \subfigure[Case 3, BNHC]{\includegraphics[width=4.4 cm,height=4.4cm]{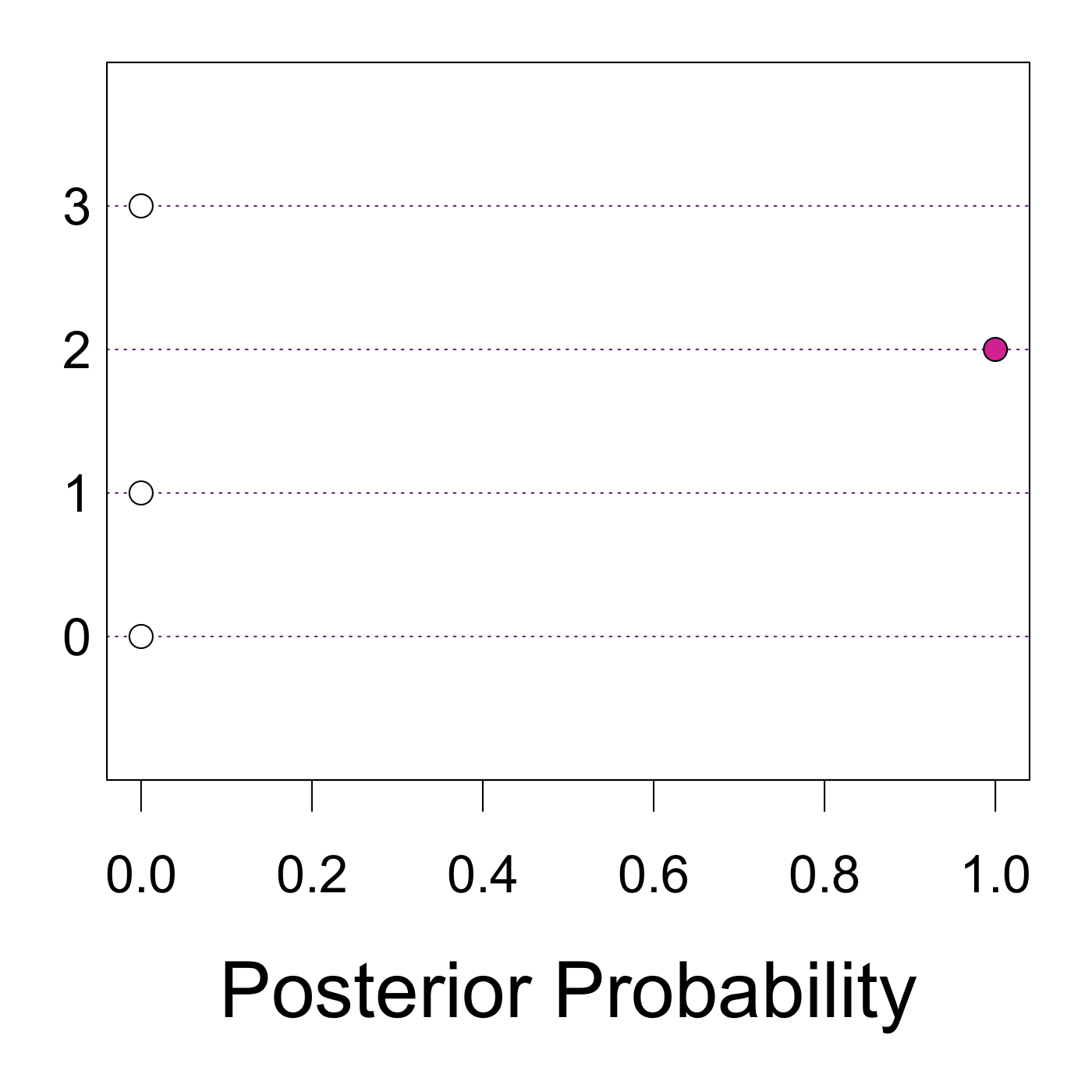}\label{sim2-dc3-hs}}
   \subfigure[Case 4, BNHC]{\includegraphics[width=4.4 cm,height=4.4cm]{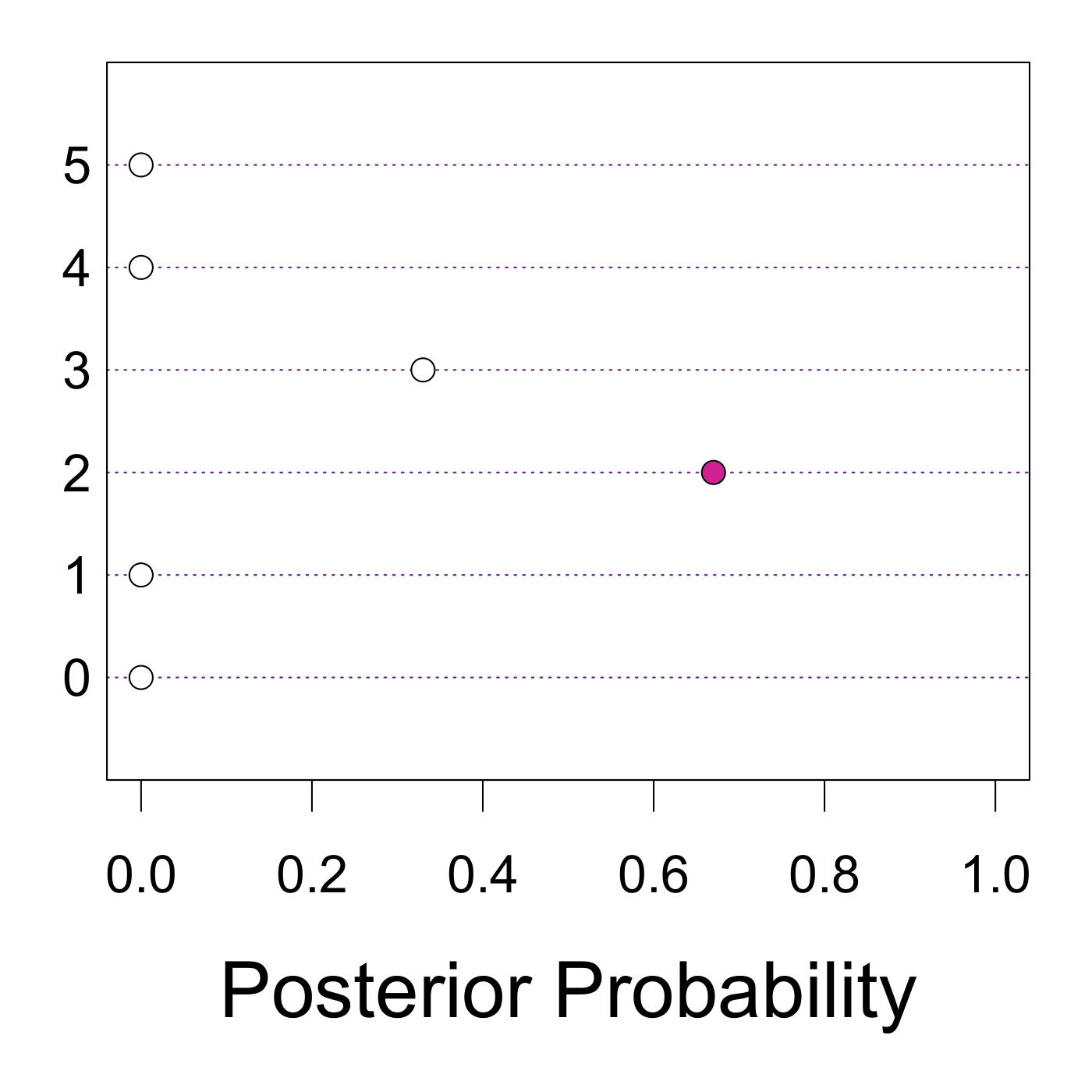}\label{sim2-dc4-hs}}
   \end{center}
 \caption{Plots showing posterior probability distribution of effective dimensionality for BNLC and BNHC models in all $4$ cases in Simulation 2. Filled bullets indicate the true value of effective dimensionality.}\label{effective2}
\end{figure}

\begin{table}[!th]
\begin{center}
\begin{tabular}
[c]{c|cccccc}
\hline
\multicolumn{1}{c}{} & \multicolumn{5}{|c}{MSE}\\
\hline
Cases & BNLC & BNHC & Lasso &  Reli\'{o}n(2017) & Binary  & Binary \\
           &                      &  &   &                       & BL       & Horseshoe \\
\hline
Case - 1 & \textbf{0.164} & 0.683 & 1.197 & 1.387 &  0.980 & 1.160\\
Case - 2 & \textbf{2.349} & 3.568 & 3.943 & 4.368 & 3.502 & 3.993\\
Case - 3 & \textbf{0.106} & 0.467  & 0.906 & 1.056 & 0.695 & 0.856\\
Case - 4 & \textbf{0.166} & 0.200 & 0.485 & 0.617 & 0.329 & 0.415\\
\hline
\end{tabular}
\caption{Performance of BNLC and BNHC vis-a-vis competitors for cases in Simulation 1. Parametric inference in terms of point estimation of edge coefficients has been captured through the Mean Squared Error (MSE). The minimum MSE among competitors for any case is made bold.}\label{Tab_sim1_B}
\end{center}
\end{table}

\subsection{Estimation of Effective Dimensionality}
Figures~\ref{effective} and \ref{effective2} present posterior probabilities of effective dimensionality of the latent positions $\bu_1, \ldots, \bu_V$ for BNLC and BNHC in Simulations 1 and 2, respectively.
Note that the true dimension of the latent space is known and recorded for all simulations in Tables~\ref{Tab1} and \ref{Tab2}. In all $8$ cases, the posterior mode corresponds to the true dimension of the latent space for both BNLC and BNHC. Compared to BNLC, the posterior distribution of $R_{eff}$ in BNHC concentrates more sharply around $R_g$ in all cases.

\begin{table}[!th]
\begin{center}
\begin{tabular}
[c]{c|cccccc}
\hline
\multicolumn{1}{c}{} & \multicolumn{5}{|c}{MSE}\\
\hline
Cases & BNLC & BNHC & Lasso &  Reli\'{o}n(2017) & Binary  & Binary \\
           &                      &   &   &                       & BL       & Horseshoe \\
\hline
Case - 1 & \textbf{0.279} & 0.418 & 0.807 & 0.939 &  0.712 & 0.739\\
Case - 2 & \textbf{0.180} & 0.388 & 0.514 & 0.665 &  0.423 & 0.548\\
Case - 3 & \textbf{0.134} & 0.549 & 0.906 &  1.097 & 0.748 & 0.883\\
Case - 4 & \textbf{0.066} & 0.106 & 0.167 &  0.221 & 0.137 & 0.141\\
\hline
\end{tabular}
\caption{Performance of BNLC and BNHC vis-a-vis competitors for cases in Simulation 2. Parametric inference in terms of point estimation of edge coefficients has been captured through the Mean Squared Error (MSE). The minimum MSE among competitors for any case is made bold.}\label{Tab_sim2_B}
\end{center}
\end{table}

\subsection{Sensitivity to the choice of Hyperparameters}\label{sim_ch3_sens}

To assess how sensitive the inferences from BNLC and BNHC are, we analyze BNLC and BNHC with different combinations of hyperparameters. Specifically for BNLC, we use the five different combinations given by, (i) $a_{\Delta}=1, b_{\Delta}=9$; (ii) $\nu=20,\delta=5$ (iii) $\nu=50,\delta=5$ (iv) $\nu=20,\delta=0.2$ (v) $\nu=50,\delta=0.2$. Combination (i) ensures small prior mean for $\xi_k$'s, while combinations (ii)-(v) allow a range of prior means for $\theta$ and $\bM$.
On the other hand, the three different combinations we employ for BNHC are, (i)' $a=1$, $b=9$ (ii)' $\nu=10$ (iii)' $\nu=50$. With these hyperparameter combinations for BNLC and BNHC, we analyze the data simulated in case 4, Simulation 1 (case chosen randomly), report performances on influential node and edge identification and the MSE values for estimating the network coefficient matrix. All these inferences  with different choices of hyperparameters are compared among themselves and compared with the inferences reported earlier on case 4, Simulation 1.

Table~\ref{Tab_sensitivity_chap3} records the MSE values for estimating the network coefficient under all these combinations. The MSE values for BNLC range between $0.10$ and $0.30$ (please see table~\ref{Tab_sim1_B}). MSE values for BNHC are found to range between $0.19$ and $0.28$ with different choices of hyperparameters, as shown in able~\ref{Tab_sim1_B}. Figure ~\ref{node_s_ch2} shows the posterior probabilities of a node being identified as influential under all these hyperparameter combinations. It shows probabilities being only little affected by the change of hyper-parameters. In fact, under hyper-parameter combinations (i),(ii) and (iv), BNLC  identifies the same set of nodes as influential which have been identified as influential by the original BNLC prior. Under combination (iii), BNLC does not identify node $9$ as influential which has been identified as influential by the original BNLC prior. Under combination (iv) BNLC identifies one additional node (node $21$) as influential over the set of nodes identified by the original prior. Under hyperparameter combination (i)', BNHC identifies the same set of nodes with the original BNHC prior except nodes $4,9,18,25$ which are identified as influential by the original prior, but not by the combination (i)'. Combinations (ii)' and (iii)' also identify the same set of nodes with the original BNHC prior except for nodes $9,18,25$. Finally, Table~\ref{Tab_sensitivity_tpr_chap3} offers TPR and FPR values corresponding to the identification of influential edges for BNLC and BNHC under various combinations of hyper-parameters. The TPR for BNHC under combination (iii)' turns out to be a little higher than the rest, but overall numbers do not show a lot of variation. We emphasize that the results turn out to be better than our competitors under all combinations.

\begin{table}[!th]
\begin{center}
\begin{tabular}
[c]{c|ccccc|ccc}
\hline
\multicolumn{1}{c|}{} &\multicolumn{5}{c|}{BNLC} & \multicolumn{3}{c}{BNHC}\\
\hline
Combinations & (i) & (ii) & (iii) &  (iv) & (v)  & (i)' & (ii)' & (iii)'\\
\hline
MSE & 0.14 & 0.30 & 0.22 & 0.10 &  0.22 & 0.19 & 0.28 & 0.28\\
\hline
\end{tabular}
\caption{Mean Squared Error (MSE) of estimating the network coefficient in BNLC and BNHC for different combinations of hyper-parameters.}\label{Tab_sensitivity_chap3}
\end{center}
\end{table}

\begin{figure}[!ht]
  \begin{center}
  \subfigure[BNLC Sensitivity]{\includegraphics[width=7cm,height=7cm]{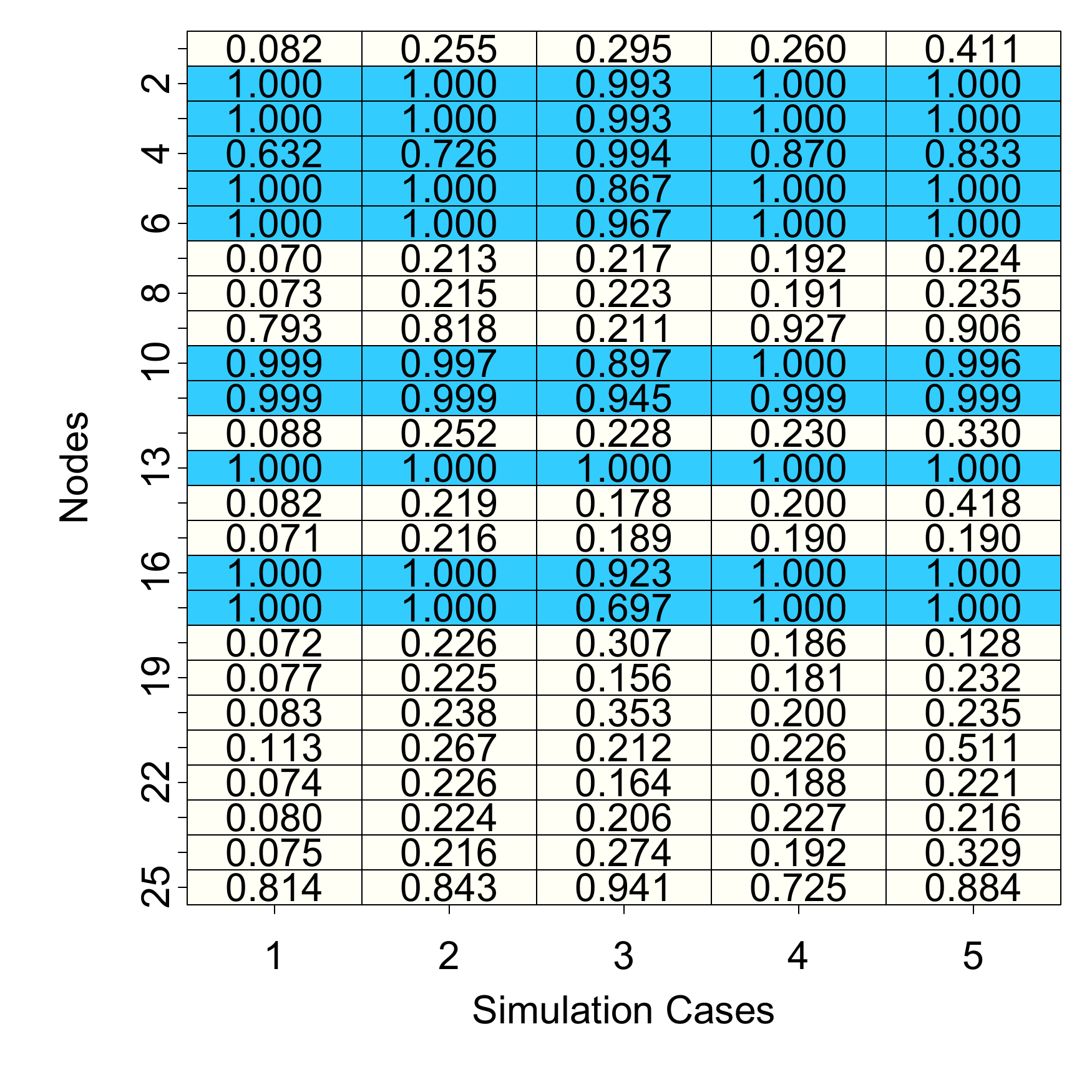}\label{sim1-node_ch2}}
  \subfigure[BNHC Sensitivity]{\includegraphics[width=7cm,height=7cm]{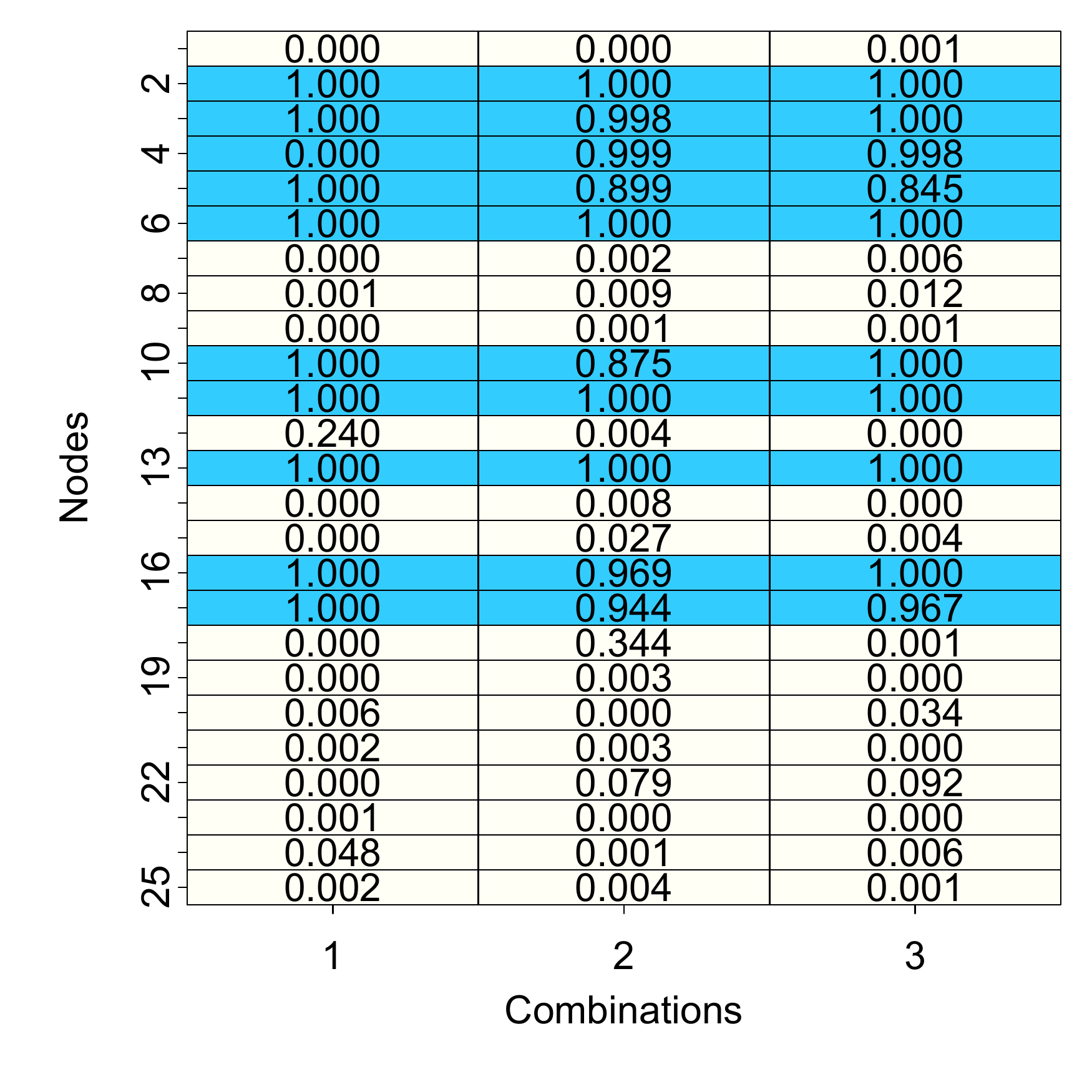}\label{sim2-node_ch2}}
 \end{center}
 \caption{Figure shows $P(\xi_k=1|Data)$ for BNLC and BNHC under different hyper-parameter combinations in the simulated data for case 4 (Simulation 1).}
\label{node_s_ch2}
\end{figure}

\begin{table}[!th]
\begin{center}
\begin{tabular}
[c]{c|ccccc|ccc}
\hline
\multicolumn{1}{c|}{} &\multicolumn{5}{c|}{BNLC} & \multicolumn{3}{c}{BNHC}\\
\hline
Combinations & (i) & (ii) & (iii) &  (iv) & (v)  & (i)' & (ii)' & (iii)'\\
\hline
TPR & 0.80 & 0.76 & 0.82 & 0.83 &  0.78 & 0.64 & 0.88 & 0.82\\
FPR & 0.16 & 0.21 & 0.17 & 0.21 &  0.18 & 0.19 & 0.24 & 0.18\\
\hline
\end{tabular}
\caption{True Positive Rates (TPR) and False Positive Rates (FPR) of identifying influential edges in BNLC and BNHC for different combinations of hyper-parameters.}\label{Tab_sensitivity_tpr_chap3}
\end{center}
\end{table}

\section{Brain Connectome Application}\label{connectome}
In this section, we present the inferential and classification ability of BNLC and BNHC in the context of a weighted diffusion tension imaging (DTI) dataset. Our dataset contains information on the \emph{full scale intelligence quotient} (FSIQ) for multiple individuals. Full scale intelligence quotient (FSIQ) is a  measure of an individual's complete cognitive capacity. It is derived from administration of selected sub-tests from the Wechsler Intelligence Scales (WIS), designed to provide a measure of an individual's overall level of general cognitive and intellectual functioning, and is a summary score derived from an individual's performance on a variety of tasks that measure acquired knowledge, verbal reasoning, attention to verbal materials, fluid reasoning, spatial processing, attentiveness to details, and visual-motor integration \cite{caplan2011encyclopedia}. A substantial body of literature has suggested that there is an IQ threshold (usually described as an IQ of approximately $120$ points) that may be characterized as superior reasoning ability \cite{brown2009executive,carson2003decreased}. Following this literature, we have converted the FSIQ scores into a binary response variable $y$, which takes value 0 if FSIQ is less or equal to 120, and takes value 1 if FSIQ is greater than 120. Thus, we classify the subjects in our study as belonging to the \emph{low IQ} group if $y=0$, and the \emph{high IQ} group if $y=1$.

Along with FSIQ measurements, brain connectome information for $n = 114$ subjects is gathered using weighted diffusion tensor imaging (DTI). DTI  is a brain imaging technique that enables measurement of the restricted diffusion of water in tissue in order to produce neural tract images. The brain imaging data we use has been pre-processed using the NDMG pre-processing pipeline \cite{kiar2016ndmg}; \cite{kiar2017example}; \cite{kiar2017science}.
In the context of DTI, the human brain is divided according to the Desikan atlas \cite{desikan2006automated}, which identifies $34$ cortical regions of interest (ROIs) both in the left and right hemispheres of the human brain, implying $68$ cortical ROIs in all.
Similar to \cite{guha2020bayesian}, this results in a brain network of a $68\times 68$ matrix for each individual. Our scientific goals in this setting include identification of brain regions or network nodes significantly related to FSIQ and classification of a subject into the low IQ or high IQ group based on his/her brain connectome information.


 Identical prior distributions for all the parameters as in the simulation studies have been used. BNLC and BNHC are both fitted with $R=4$, which is found to be sufficient for this study. Further, \cite{guha2020bayesian} show robust inference as long as the chosen $R$ is bigger than the effective dimensionality of the latent variables. Similar to article \cite{guha2020bayesian}, we also do a sensitivity study to check the impact of $R$ on predictive inference.
 The choice of hyperparameters for BNLC and BNHC are made similar to the simulation studies. A brief explanation for such choices of hyper parameters is provided in the simulation section. The MCMC chain is run for $50,000$ iterations, with the first $30,000$ iterations discarded as burn-in. Convergence is assessed by comparing different simulated sequences of representative parameters started at different initial values \cite{gelman2014bayesian}. All inference is based on the remaining $20,000$ post burn-in iterates appropriately thinned. 

\subsection{Findings from the Brain Connectome Application}
As in simulation studies, we put our emphasis on identifying influential brain regions of interest (ROIs) associated with FSIQ. 
The BNLC model estimates posterior probabilities over $0.5$ (hence detecting as \emph{influential}) for $38$ ROIs, out of which $20$ regions are in the left hemisphere and $18$ regions are in the right hemisphere. Among the regions detected in both the hemispheres, a large number belong to the \emph{frontal}, \emph{temporal} and \emph{cingulate} lobes. 
Using the same principle, the BNHC model identifies $48$ nodes to be influential.  Out of the $48$ influential nodes, $26$ are detected in the left hemisphere and the rest in the right hemisphere. The ROIs are mainly detected in the \emph{temporal, frontal, parietal} and \emph{cingulate} lobes in both hemispheres. 
Figure~\ref{nodes} plots the estimated posterior probability of an ROI being detected as influential by the BNLC and BNHC models. Notably, there are $29$ ROIs identified by both BNLC and BNHC, given in Table~\ref{Tab_chbin_nodes}. 

A large number of the $29$ influential nodes detected by both BNLC and BNHC are part of the \emph{frontal} lobes in both the hemispheres. Numerous studies have linked the frontal region to an individual's intelligence and cognitive functions \cite{yoon2017brain,stuss1985subtle,razumnikova2007creativity,miller1985cognitive,kolb1981performance}. Our method also finds a significant association between FSIQ and the left \emph{inferior parietal lobule}, the left \emph{precuneus} and the \emph{supramarginal gyri} in both the hemispheres, in the \emph{parietal} lobe, regions also found to be significantly related to FSIQ by \cite{yoon2017brain}.  

We additionally look into ROIs which are detected by only of the two methods (lets say, BNLC), and report the posterior probabilities of these ROIs being active under the other method (i.e., BNHC). 
Figure ~\ref{cross-ROIs} shows the posterior probabilities of nodes being active under the `other' method as discussed above. It is observed that the nodes selected by BNHC but not by BNLC have probabilities not very far from $0.5$ under BNLC, which says that BNLC is not enough confident to exclude these nodes from the set of influential nodes. However, most of the nodes selected by BNLC but not by BNHC show smaller probabilities of being influential under BNHC. Perhaps, BNLC is more conservative in including nodes in the set of influential nodes, which is responsible for the discrepancy between the number of identified nodes by BNHC and BNLC.

\begin{figure}[!ht]
  \begin{center}
  \includegraphics[width=12cm,height=12cm]{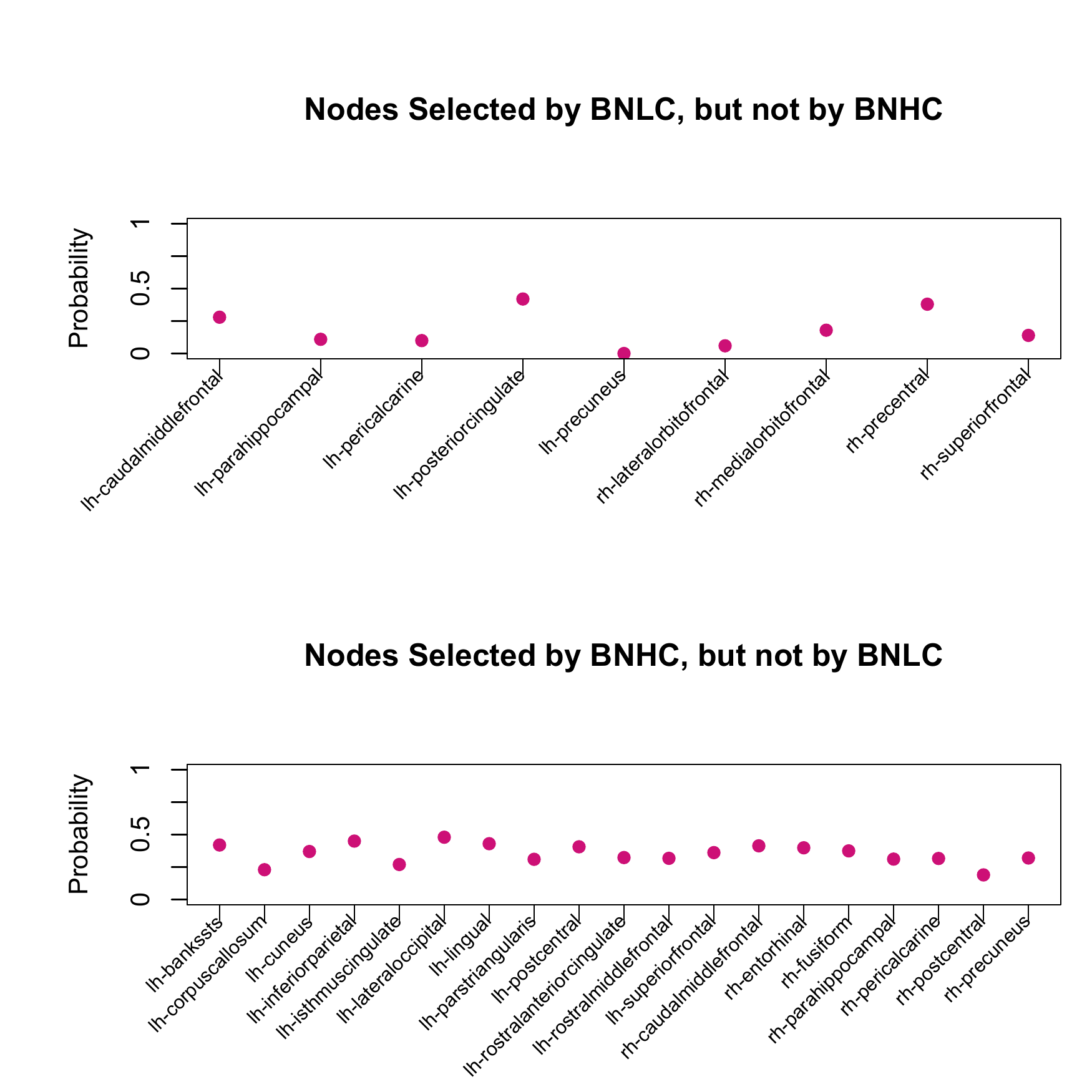}
 \end{center}
 \caption{Figure shows the posterior probabilities of nodes selected as \emph{influential} by one method, but not by another, of being active.}
\label{cross-ROIs}
\end{figure}

As described earlier, we identify influential edges connecting pairs of influential nodes using the algorithm described in Appendix C. Figure~\ref{edges} presents the influential edges (among all edges connecting pairs of influential nodes) identified by the BNLC and BNHC models. Note that BNLC and BNHC identify $142$ and $291$ edges as being influential out of $\binom{38}{2}$ and $\binom{48}{2}$ possibilities, respectively. Since a different number of nodes are detected as influential by BNHC and BNLC, to make a fair comparison, we consider the $29$ nodes detected as influential by both  methods, and use our algorithm to find the number of influential edges among these $\binom{29}{2}$ possibilities for both BNLC and BNHC. The numbers turn out to be $96$ and $184$, respectively. We note that there are a few nodes which are identified as influential by either BNHC or BNLC, but none of the edges connecting these nodes are found to be influential. As an example, although the \emph{frontal pole} and the \emph{temporal pole} in the left hemisphere are identified as influential nodes by BNLC, none of the edges connecting these two nodes turn out to be influential. This phenomenon may be due to the use of the FDR in the edge selection procedure, which finds edges that are most likely to be active while controlling for false discoveries.  Hence, not identifying an edge does not necessarily mean that the edge is not active, it just means that there are others that satisfy the criteria better.


Similar to simulation studies, we dig deeper to analyze the discrepancy in the number of influential edges identified by BNLC and BNHC. Specifically, we rank the $\binom{29}{2}=406$ edges connecting the nodes found to be influential by both BNLC and BNHC, according to the absolute values of their posterior means.  Table~\ref{Tab_intersect_edge} shows between 23-74\% intersections. 

To examine the predictive ability of the Bayesian network classification model, we report the area under curve (AUC) of the ROC curve for BNLC and BNHC, along with all competing methods. The AUCs are computed using a $10$-fold cross validation approach. The AUC estimates presented in Table~\ref{Tab5} indicate better performance of both BNLC and BNHC, with BNLC slightly outperforming. Frequentist Binary Lasso turns out to be the next best performer, while BLasso and BHS perform very similar to a random classifier. Finally, the effective dimensionality of the model is investigated for both BNLC and BNHC, and they turn out to be $2.17$ and $2$, respectively. 

\begin{figure}[!ht]
   \begin{center}
   \subfigure[BNLC]{\includegraphics[width=10 cm, height = 7 cm]{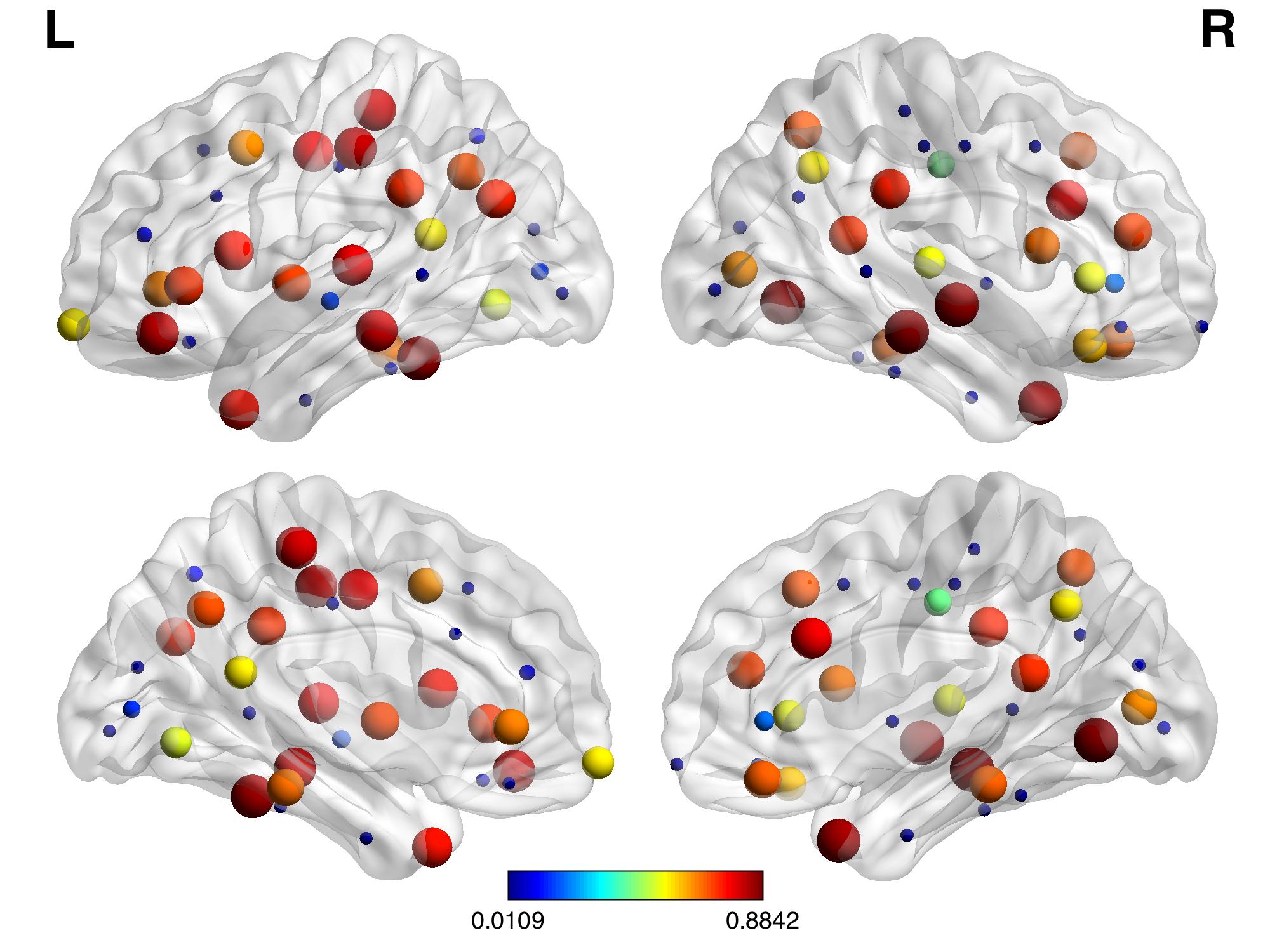}}\label{BNLC_node}\\
   \subfigure[BNHC]{\includegraphics[width=10 cm, height = 7 cm]{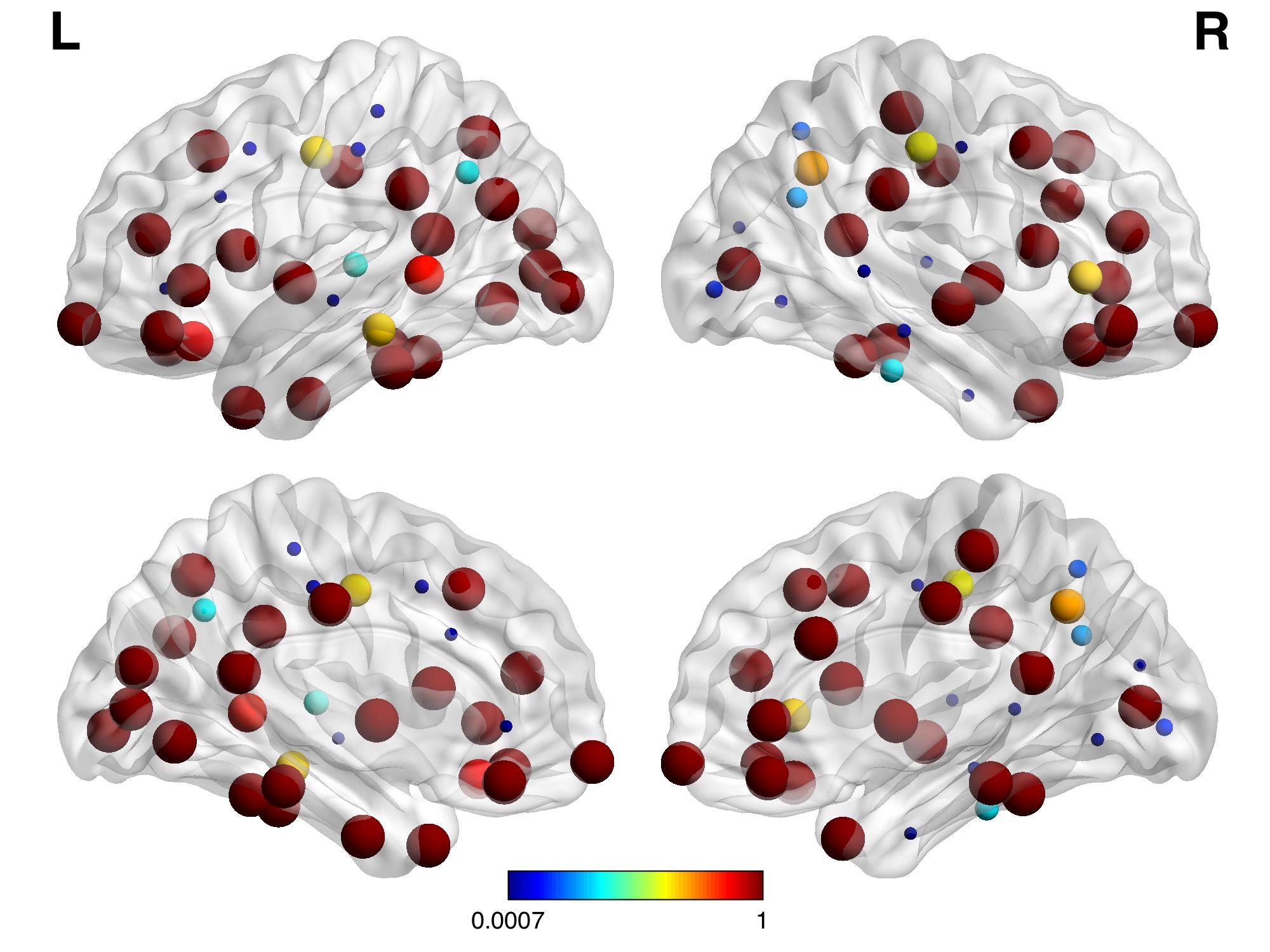}}\label{BNHC_node}
    \end{center}
\caption {Lateral and medial views of the brain (left and right hemispheres) showing all $68$ regions of interest (ROIs). The size and color of the ROIs vary according to the value of the posterior probabilities of them being actively related to the binary response for both BNLC and BNHC models.}\label{nodes}
\end{figure}

\begin{table}
\footnotesize
\begin{tabular}{ lll }
\hline
\hline
Hemisphere & Lobe & Node \\ \hline
\multirow{6}{*}{Left} & Temporal & fusiform, middle temporal gyrus, parahippocampal, temporal pole, transverse temporal\\
  & Cingulate & isthmus cingulate cortex\\
  & Frontal & pars opercularis, pars orbitalis, pars triangularis, frontal pole\\
  & Occipital & lingual\\
  & Parietal &  inferior parietal lobule, precuneus, supramarginal gyrus\\
     & Insula &  insula\\ \hline
\multirow{7}{*}{Right} & Temporal & parahippocampal, superior temporal gyrus, temporal pole  \\
 & Cingulate & caudal anterior cingulate, isthmus cingulate cortex\\
 & Frontal &  lateral orbitofrontal, medial orbitofrontal, pars opercularis, pars orbitalis,\\
 &         & rostral middle frontal gyrus, superior frontal gyrus\\
   & Occipital & pericalcarine\\
   & Parietal & supramarginal gyrus\\
    & Insula &  insula\\
 \hline
\hline
\end{tabular}
\caption{Nodes identified as influential by both BNLC and BNHC.}\label{Tab_chbin_nodes}
\end{table}

\begin{table}[!th]
\begin{center}
\begin{tabular}
[c]{cccc}
\hline
Top 100 & Top 200 & Top 300 \\
\hline
23 & 99 & 222 \\
\hline
\end{tabular}
\caption{Top 100 represents the number of edges common among the top $100$ edges identified by BNLC and BNHC. Top 200 and Top 300 are defined analogously.}\label{Tab_intersect_edge}
\end{center}
\end{table}

\begin{table}[!th]
\begin{center}
\begin{tabular}
[c]{c|cccccc}
\hline
\textbf{Method} & BNLC & BNHC & Lasso &  Reli\'{o}n(2017) & Binary  & Binary \\[-0.15in]
&                &   &  &                       & BL       & BHS \\
\hline
\textbf{AUC} & 0.617 & 0.598 & 0.532 & 0.466 &  0.461 & 0.484\\
\hline
\end{tabular}
\caption{Predictive performance of Bayesian Network Classification (BNC) vis-a-vis competitors in terms of Area Under Curve (AUC) of the ROC. AUC has been calculated in each case using 10-fold cross validation.}\label{Tab5}
\end{center}
\end{table}

\begin{figure}[!ht]
   \begin{center}
   \subfigure[BNLC]{\includegraphics[width=8 cm, height = 7.5 cm]{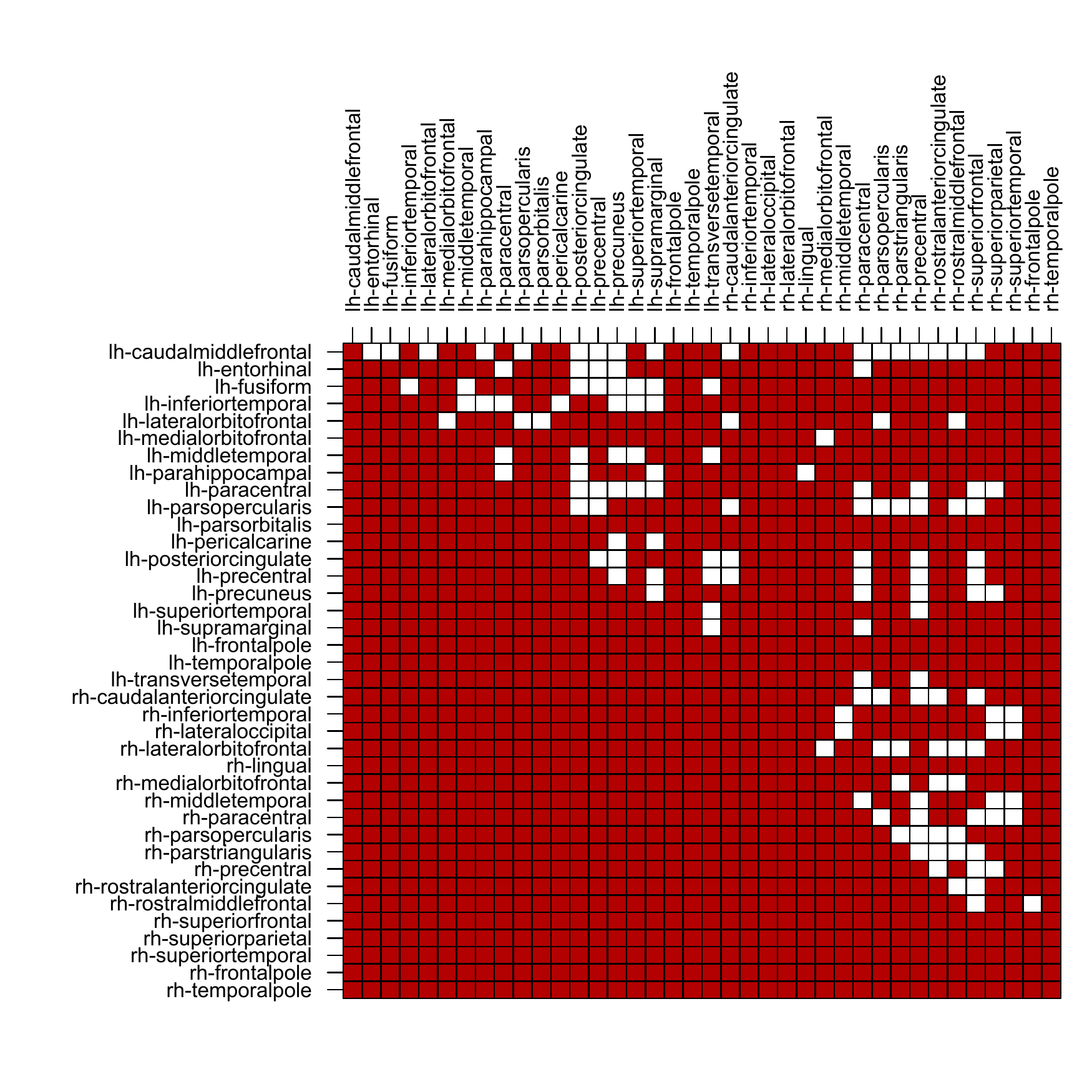}}\label{BNLC_edge}\\
   \subfigure[BNHC]{\includegraphics[width=8 cm, height = 7.5 cm]{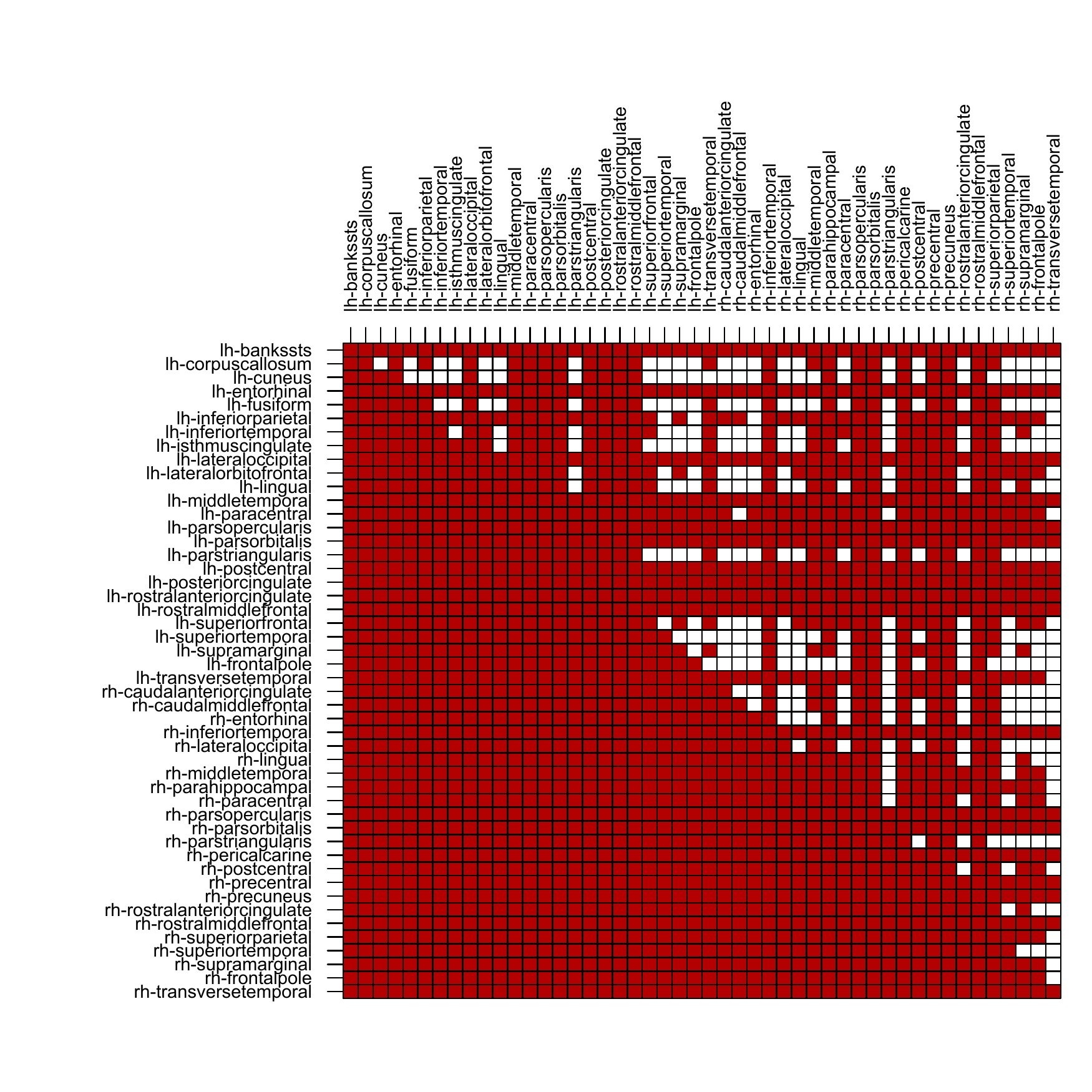}}\label{BNHC_edge}
    \end{center}
\caption {Plot showing whether an edge connecting two influential nodes is influential or not. Note that the map is a $M \times M$ symmetric matrix, where $M$ denotes the number of influential nodes, and each cell denotes an edge connecting the corresponding pair of nodes. The axis labels are the abbreviated names of the influential ROIs in the left (starting with `lh -') and the right (starting with `rh -') hemispheres of the brain. Full names of the ROIs can be obtained from the widely available Desikan brain atlas. A white cell represents an influential edge, while red cell represents a non-influential edge.}\label{edges}
\end{figure}

\subsection{Sensitivity to the choice of hyperparameters}\label{sense_real_362}

We have already discussed how the hyperparameters are chosen for the simulation studies and data analysis. To assess how sensitive the inferences from BNLC and BNHC are, we analyze BNLC and BNHC with different combinations of hyperparameters. Specifically for BNLC, we use the five different combinations (i)-(v) given in Section~\ref{sim_ch3_sens}, and three different combinations (i)'-(iii)' for BNHC also mentioned in Section~\ref{sim_ch3_sens}. We report performances on the number of influential nodes identified. We also find the number of influential edges connecting influential nodes.

\begin{table}[!th]
\begin{center}
\begin{tabular}
[c]{c|ccccc|ccc}
\hline
\multicolumn{1}{c|}{} &\multicolumn{5}{c|}{BNLC} & \multicolumn{3}{c}{BNHC}\\
\hline
Combinations & (i) & (ii) & (iii) &  (iv) & (v)  & (i)' & (ii)' & (iii)'\\
\hline
\# Nodes detected & 35 & 39 & 34 & 40 &  37 & 45 & 49 & 44\\
\# Intersections with original analysis & 34 & 36 & 34 & 37 & 37 & 42 & 45 & 43\\
\hline
\end{tabular}
\caption{Number of nodes identified as influential for all combinations are presented. The table also presents the number of intersections of influential nodes between different combinations and the original analysis.} \label{Tab_sensitivity_chap3_realdata}
\end{center}
\end{table}

\begin{table}[!th]
\begin{center}
\begin{tabular}
[c]{c|ccccc|ccc}
\hline
\multicolumn{1}{c|}{} &\multicolumn{5}{c|}{BNLC} & \multicolumn{3}{c}{BNHC}\\
\hline
Combinations & (i) & (ii) & (iii) &  (iv) & (v)  & (i)' & (ii)' & (iii)'\\
\hline
\# Edges detected & 122 & 113 & 125 & 118 &  107 & 272 & 265 & 262\\
\# Intersections with original analysis & 117 & 112 & 119 & 111 & 101 & 263 & 264 & 257\\
\hline
\end{tabular}
\caption{Number of edges identified as influential for all combinations are presented. The table also presents the number of intersections of influential nodes between different combinations and the original analysis.} \label{Tab_sensitivity_chap3_realdata_edges}
\end{center}
\end{table}

Table~\ref{Tab_sensitivity_chap3_realdata} records the number of nodes identified as influential and the number of intersections of influential nodes between different combinations and the original analysis.
Recall that the original analysis of BNLC identifies $38$ influential nodes. Since this is a high dimensional regression paradigm with number of parameters far exceeding the sample size, one expects the prior hyper-parameters to have some effect on the inference. Indeed, there is some variation in the number of identified nodes, though they largely agree with each other under different hyperparameter settings. In fact, we find a large number of intersections among the identified nodes in the original analysis with the nodes identified under different hyperparameter combinations. A similar story emerges from BNHC. We also find $31$ nodes identified by all hyperparameter combinations in BNLC. Similarly, $40$ nodes are identified by all hyperparameter combinations of BNHC. We calculate the number of influential edges among these $\binom{31}{2}$ edges and $\binom{40}{2}$ edges in BNLC and BNHC respectively, for all hyperparameter combinations. Table~\ref{Tab_sensitivity_chap3_realdata_edges} presents the number of edges detected as influential, as well as the number of intersecting edges with the original analysis. Again, due to the high dimensionality of the problem, the variation in the number of identified edges with different choices of hyperparameters is expected, though the variation turns out not to be very significant. 

Finally, to check sensitivity to the choice of $R$ on the performance of BNLC and BNHC, we run the data analysis for BNHC and BNLC with $R=8$ and $R=10$, and report the posterior mean of the effective dimensionality, along with AUC. Table~\ref{Tab_sensitivity_tpr_chap3_R_realdata} reports the posterior mean of effective dimensionality, which shows very moderate increase with increasing $R$. However, increasing $R$ seems to have almost no effect on AUC.

\begin{table}[!th]
\begin{center}
\begin{tabular}
[c]{c|ccc|ccc}
\hline
\multicolumn{1}{c|}{} &\multicolumn{3}{c|}{BNLC} & \multicolumn{3}{c}{BNHC}\\
\hline
 & $R=4$ & $R=8$ & $R=10$ & $R=4$ & $R=8$ & $R=10$\\
\hline
Posterior mean Eff. Dim. & 2.17 & 2.78 & 2.96 & 2.00 & 2.74 & 3.04 \\
AUC & 0.61 & 0.63 & 0.59 & 0.59 & 0.60 & 0.59\\
\hline
\end{tabular}
\caption{AUC and posterior mean of effective dimensionality for BNLC and BNHC under different choices of $R$.}\label{Tab_sensitivity_tpr_chap3_R_realdata}
\end{center}
\end{table}

\section{Conclusion}\label{conclusion}
We develop a binary Bayesian network regression model that enables classifying multiple networks with ``labeled nodes" into two groups, identifies influential network nodes and predicts the class in which a newly observed network belongs. Our contribution lies in carefully constructing a class of network global-local shrinkage priors on the network predictor coefficient while recognizing the latent network structure in the predictor variable. In particular, we investigate two specific network shrinkage priors from this general class, leading to two network classifiers BNLC and BNHC. Our extensive simulation study shows competitive performance between BNLC and BNHC in terms of inference and classification with no clear winner, and both of them are found to outperform other competitors. Another major contribution of the proposed framework remains theoretically understanding the Bayesian network classifier model with the Network Lasso shrinkage prior. Specifically, we develop theory guaranteeing accurate classification as the sample size tends to infinity. The theoretical developments allow the number of possible interconnections in the network predictor to grow at a faster rate than the sample size. We analyze a brain connectome dataset with brain connectivity networks between different regions of interest for multiple individuals, and information on whether an individual is in a \emph{low} or a \emph{high} IQ category. BNC shows satisfactory out of sample classification and identifies important brain regions actively influencing the FSIQ of an individual.
\bibliographystyle{natbib}
\bibliography{biblio}

\section{Appendix}
\subsection{Appendix A}\label{appA}

This section provides full conditionals for all the parameters in the Bayesian binary network regression with network lasso shrinkage prior on $\bgamma$.
Assume $\bW=(\bu_1'\bLambda\bu_2,...,\bu_1'\bLambda\bu_V,....,\bu_{V-1}'\bLambda\bu_V)'$, $\bD=diag(s_{1,2}^2,...,s_{V-1,V}^2)$ and $\bgamma=(\gamma_{1,2},...,\gamma_{V-1,V})'$.
Thus, with $n$ data points, the hierarchical model with the \emph{network lasso prior} in the binary setting can be written as
\begin{align*}
&\qquad\qquad\qquad\qquad\qquad \bt \sim \mathrm{N}(\mu+\bX\bgamma,\bOmega^{-1})\\
&\bgamma \sim \mathrm{N}(\bW,\bD),\:\:
\bu_k|\xi_k=1 \sim N(\bu_k \given \bzero, \bQ),\:\:\bu_k|\xi_k=0\sim \delta_{\bzero},\:\:\xi_k\sim Ber(\Delta),\:\:\mu \sim flat() \\
& s_{k,l}^2 \sim Exp(\theta^2/2),\:\:\:\:\theta^2 \sim Gamma(\zeta,\iota),\:\:\:\:
\bQ \sim IW(\nu,\bI),\:\:\:\:
\Delta \sim Beta(a_{\Delta},b_{\Delta})\\
&\qquad\qquad\qquad p(\omega_i) \sim PG(1,0),\:\: \lambda_r \sim Ber (\pi_{r}), \:\:
 \pi_{r} \sim Beta(1, r^{\eta}),\:\eta>1.
\end{align*}

The full conditional distributions of the model parameters are given below.
\begin{itemize}
\item $\mu \given - \sim N\left(\frac{{\boldsymbol 1}'\bOmega(\bt-\bX\bgamma)}{{\boldsymbol 1}'\bOmega {\boldsymbol 1}},\frac{1}{{\boldsymbol 1}'\bOmega {\boldsymbol 1}}\right)$
\item $\bgamma \given - \sim N(\bmu_{\bgamma \given \cdot}, \bSigma_{\bgamma \given \cdot})$,
where $\bmu_{\bgamma \given \cdot} = {(\bX' \bOmega \bX + \bD^{-1})}^{-1}(\bX' \bOmega (\bt-\mu{\boldsymbol 1}) + \bD^{-1}\bW)$ and $\bSigma_{\bgamma \given \cdot} = {(\bX' \bOmega \bX + \bD^{-1})}^{-1}$
\item $s_{k,l}^2 \given - \sim GIG\left[\frac{1}{2}, (\gamma_{k,l} - \bu_k '\bLambda \bu_l)^2, \theta^2 \right]$, where GIG denotes the generalized inverse Gaussian distribution.
\item $\theta^2 \given - \sim Gamma\left[ \left({\zeta} + \frac{V(V-1)}{2}\right), \left(\iota + \sum_{k < l} \frac{s_{k,l}^2}{2} \right) \right]$
\item $\bu_k \given - \sim   w_{\bu_k} \: \delta_0 (\bu_k)  +  (1 - w_{\bu_k}) \: N(\bu_k \given \bm_{\bu_k}, \bSigma_{\bu_k})$, where
$\bU^{\ast}_k=(\bu_1:\cdots:\bu_{k-1}:\bu_{k+1}:\cdots:\bu_{V})' \bLambda,\:\:\bH_k=diag(s_{1,k}^2,...,s_{k-1,k}^2,s_{k,k+1}^2,...,s_{k,V}^2),\:\:\bgamma_k=(\gamma_{1,k},...,\gamma_{k-1,k},\gamma_{k,k+1},...,\gamma_{k,V})$, and
\begin{align*}
&\bSigma_{\bu_k} = \left(\bU^{\ast'}_h\bH_k^{-1}\bU^{\ast}_k+\bQ^{-1}\right)^{-1},\:\:\bm_{\bu_k}=\bSigma_{\bu_k}\bU^{\ast'}_k\bH_k^{-1}\bgamma_k\\
& w_{\bu_k} = \frac{(1-\Delta)N(\bgamma_k\given \bzero,\bH_k)}{(1-\Delta)N(\bgamma_k\given \bzero,\bH_k)+\Delta N(\bgamma_k\given \bzero,\bH_k+\bU^{\ast}_k\bQ\bU^{\ast'}_k)}
\end{align*}
\item $\xi_k|-\sim Ber(1 - w_{\bu_k})$
\item $\Delta \given - \sim Beta\left[(a_{\Delta} + \sum_{k = 1}^{V} \xi_k),   (b_{\Delta} + \sum_{k = 1}^{V} (1 - \xi_k))\right]$.
\item $\bQ \given - \sim IW [(\nu + \{\#k : \bu_k \neq \bzero \}),(\bI + \sum_{k : \bu_k \neq \bzero} \bu_k\bLambda \bu_k') ]$.
\item $\lambda_r \given - \sim Ber(p_{\lambda_r})$, where
       $p_{\lambda_r}=\frac{\pi_{r}N(\bgamma\given \bW_1,\bD)}{\pi_{r}N(\bgamma\given \bW_1,\bD)+
       (1-\pi_{r})N(\bgamma\given \bW_0,\bD)}$. Here \\ $\bW_1=(\bu_1'\bLambda_1\bu_2,...,\bu_1'\bLambda_1\bu_V,....,\bu_{V-1}'\bLambda_1\bu_V)'$, $\bW_0=(\bu_1'\bLambda_0\bu_2,...,\bu_1'\bLambda_0\bu_V,....,\bu_{V-1}'\bLambda_0\bu_V)'$,
       $\bLambda_1=diag(\lambda_1,..,\lambda_{r-1},1,\lambda_{r+1},..,\lambda_R)$, $\bLambda_0=diag(\lambda_1,..,\lambda_{r-1},0,\lambda_{r+1},..,\lambda_R)$, for $r=1,..,R$.
\item $\pi_{r}\given - \sim Beta(\lambda_r+1,1-\lambda_r+r^{\eta})$, for $r=1,..,R$.

Using the relationship, $PG(x \given b,c) \propto \exp(- \frac{c^2 x}{2}) PG (x \given 1,0)$ \cite{polson2013bayesian}, we obtain
\item $\omega_i \given - \sim PG(1,\mu + \bx'_i \bgamma)$, for $i=1,..,n$. \\
\end{itemize}

\subsection{Appendix B}\label{appB}

This section provides full conditionals for all the parameters in the Bayesian network classifier model introduced in this article with Bayesian network horseshoe prior.
Assume $\bW=(\bu_1'\bLambda\bu_2,...,\bu_1'\bLambda\bu_V,....,\bu_{V-1}'\bLambda\bu_V)'$, $\bD=diag(\sigma^2s^2_{1,2},...,\sigma^2s^2_{V-1,V})$ and $\bgamma=(\gamma_{1,2},...,\gamma_{V-1,V})'$.
Thus, with $n$ data points, the hierarchical model with the network horseshoe prior in the binary setting can be written as
\begin{align*}
&\qquad\qquad\qquad\qquad\qquad \bt \sim \mathrm{N}(\mu+\bX\bgamma,\bOmega^{-1})\\
&\bgamma \sim \mathrm{N}(\bW,\bD),\:\:
\bu_k|\xi_k=1 \sim N(\bu_k \given \bzero, \bQ),\:\:\bu_k|\xi_k=0\sim \delta_{\bzero},\:\:\xi_k\sim Ber(\Delta),\:\:\mu \sim flat() \\
& s_{k,l} \sim C^+(0,1),\:\:\:\:\sigma \sim C^+(0,1),\:\:\:\:
\bQ \sim IW(\nu,\bI),\:\:\:\:
\Delta \sim Beta(a_{\Delta},b_{\Delta})\\
&\qquad\qquad\qquad p(\omega_i) \sim PG(1,0),\:\: \lambda_r \sim Ber (\pi_{r}), \:\:
 \pi_{r} \sim Beta(1, r^{\eta}),\:\eta>1.
\end{align*}

Note that, following \cite{makalic2015simple},
\begin{align*}
& s_{k,l} \sim C^+(0,1),\:\:\:\:\sigma \sim C^+(0,1)\:\:\:\:
\end{align*}

can be written in an augmented form as
\begin{align*}
& s^2_{k,l} \given \nu_{k,l} \sim IG\left(\frac{1}{2},\frac{1}{\nu_{k,l}}\right),\:\:\:\:\nu_{k,l} \sim IG\left(\frac{1}{2},1\right),\:\:\:\:\sigma^2 \given \sigma_2 \sim IG\left(\frac{1}{2},\frac{1}{\sigma_2}\right),\:\:\:\:\sigma_2 \sim IG\left(\frac{1}{2},1\right).
\end{align*}

With the model formulation described above, the full conditional distributions of the model parameters are given by the following distributions:
\begin{itemize}
\item $\mu \given - \sim N\left(\frac{{\boldsymbol 1}'\bOmega(\bt-\bX\bgamma)}{{\boldsymbol 1}'\bOmega {\boldsymbol 1}},\frac{1}{{\boldsymbol 1}'\bOmega {\boldsymbol 1}}\right)$
\item $\bgamma \given - \sim N(\bmu_{\bgamma \given \cdot}, \bSigma_{\bgamma \given \cdot})$,
where $\bmu_{\bgamma \given \cdot} = {(\bX' \bOmega \bX + \bD^{-1})}^{-1}(\bX' \bOmega (\bt-\mu{\boldsymbol 1}) + \bD^{-1}\bW)$ and $\bSigma_{\bgamma \given \cdot} = {(\bX' \bOmega \bX + \bD^{-1})}^{-1}$
\item $s^2_{k,l} \given - \sim IG\left[1, (\frac{1}{\nu_{k,l}} + \frac{(\gamma_{k,l} - \bu_k '\bLambda \bu_l)^2}{2\sigma^2}) \right]$
\item $\sigma^2 \given - \sim IG \left[ \left(\frac{1}{2} + \frac{V(V-1)}{4}\right), \left(\frac{1}{\sigma_2} +  \sum_{k < l} \frac{(\gamma_{k,l} - \bu_k '\bLambda \bu_l)^2}{ 2s^2_{k,l} }\right) \right]$
\item $\nu_{k,l} \given - \sim IG\left[1, (1 + \frac{1}{s^2_{k,l}}) \right]$
\item $\sigma_2 \given - \sim IG\left[1, (1 + \frac{1}{\sigma^2}) \right]$

\item $\bu_k \given - \sim   w_{\bu_k} \: \delta_0 (\bu_k)  +  (1 - w_{\bu_k}) \: N(\bu_k \given \bm_{\bu_k}, \bSigma_{\bu_k})$, where
$\bU^{\ast}_k=(\bu_1:\cdots:\bu_{k-1}:\bu_{k+1}:\cdots:\bu_{V})' \bLambda,\:\:\bH_k=diag(s_{1,k}^2,...,s_{k-1,k}^2,s_{k,k+1}^2,...,s_{k,V}^2),\:\:\bgamma_k=(\gamma_{1,k},...,\gamma_{k-1,k},\gamma_{k,k+1},...,\gamma_{k,V})$, and
\begin{align*}
&\bSigma_{\bu_k} = \left(\bU^{\ast'}_h\bH_k^{-1}\bU^{\ast}_k/\sigma^2+\bQ^{-1}\right)^{-1},\:\:\bm_{\bu_k}=\bSigma_{\bu_k}\bU^{\ast'}_k\bH_k^{-1}\bgamma_k/\sigma^2\\
& w_{\bu_k} = \frac{(1-\Delta)N(\bgamma_k\given \bzero,\sigma^2\bH_k)}{(1-\Delta)N(\bgamma_k\given \bzero,\sigma^2\bH_k)+\Delta N(\bgamma_k\given \bzero,\sigma^2\bH_k+\bU^{\ast}_k\bQ\bU^{\ast'}_k)}
\end{align*}
\item $\xi_k|-\sim Ber(1 - w_{\bu_k})$
\item $\Delta \given - \sim Beta\left[(a_{\Delta} + \sum_{k = 1}^{V} \xi_k),   (b_{\Delta} + \sum_{k = 1}^{V} (1 - \xi_k))\right]$.
\item $\bQ \given - \sim IW [(\nu + \{\#k : \bu_k \neq \bzero \}),(\bI + \sum_{k : \bu_k \neq \bzero} \bu_k\bLambda \bu_k') ]$.
\item $\lambda_r \given - \sim Ber(p_{\lambda_r})$, where
       $p_{\lambda_r}=\frac{\pi_{r}N(\bgamma\given \bW_1,\sigma_2^2\bD)}{\pi_{r}N(\bgamma\given \bW_1,\sigma_2^2\bD)+
       (1-\pi_{r})N(\bgamma\given \bW_0,\sigma_2^2\bD)}$. Here \\ $\bW_1=(\bu_1'\bLambda_1\bu_2,...,\bu_1'\bLambda_1\bu_V,....,\bu_{V-1}'\bLambda_1\bu_V)'$, $\bW_0=(\bu_1'\bLambda_0\bu_2,...,\bu_1'\bLambda_0\bu_V,....,\bu_{V-1}'\bLambda_0\bu_V)'$,
       $\bLambda_1=diag(\lambda_1,..,\lambda_{r-1},1,\lambda_{r+1},..,\lambda_R)$, $\bLambda_0=diag(\lambda_1,..,\lambda_{r-1},0,\lambda_{r+1},..,\lambda_R)$, for $r=1,..,R$.
\item $\pi_{r}\given - \sim Beta(\lambda_r+1,1-\lambda_r+r^{\eta})$, for $r=1,..,R$.

Using the relationship, $PG(x \given b,c) \propto \exp(- \frac{c^2 x}{2}) PG (x \given b,0)$ \cite{polson2013bayesian}, we obtain
\item $\omega_i \given - \sim PG(1,\mu + \bx'_i \bgamma)$, for $i=1,..,n$. \\
\end{itemize}


\subsection{Appendix C}\label{appCW}
Similar to the assumptions made by \cite{wei2017contraction} in their proof of posterior consistency for binary logistic regression, we prove our results assuming that the centering parameter $\mu=0$ in both the true and the data generating models. We note that the main structure of the proof will remain unchanged with this assumption and the result proved in this article can be trivially extended to the setting with nonzero $\mu$.

We begin by defining some notations. In the proof, $\Pi(\cdot)$ will be used to denote the generic probability notation.
We define the notation of the log-likelihood function by
\begin{align}\label{abar}
w_{\bgamma,n}(\by_n)=\sum\limits_{i=1}^{n}[(\bx_i'\bgamma)y_i-z(\bx_i'\bgamma)],\:\:z(\bx_i'\bgamma)=\log(1+\exp(\bx_i'\bgamma)).
\end{align}
We also introduce the function $C_{\by_n,n}(\cdot)$ to quantify the curvature of $w_{\bgamma,n}(\by_n)$ around $\bgamma^{(0)}$,
\begin{align}\label{curvature}
C_{\by_n,n}(\bgamma)=w_{\bgamma,n}(\by_n)-w_{\bgamma^{(0)},n}(\by_n)-\nabla w_{\bgamma^{(0)},n}(\by_n)'(\bgamma-\bgamma^{(0)}),
\end{align}
where $\nabla w_{\bgamma^{(0)},n}(\by_n)$ is the derivative of $w_{\bgamma^{(0)},n}(\by_n)$ w.r.t. $\bgamma$, evaluated at $\bgamma^{(0)}$. Also the likelihood $p_{\bgamma}(\by_n)$ can be written using the above notations as
$p_{\bgamma}(\by_n)=\prod_{i=1}^n \exp(w_{\bgamma,n}(y_i))$. The notations $E_{\bgamma}(\cdot)$ and $E_{\bgamma^{(0)}}(\cdot)$ have been reserved to denote expectation w.r.t the distribution of $\by_n|\bgamma$ and $\by_n|\bgamma^{(0)}$ respectively.

The proof of Theorem \ref{theorem:main} relies in part on the existence of exponentially consistent sequence of tests.
\paragraph{\bf Definition}
An exponentially consistent sequence of test functions $\Phi_n$ for testing $H_0:\bgamma=\bgamma^{0}$ vs. $H_1:\bgamma\in
\mathcal{A}_n^c$ satisfies
\begin{equation*}
E_{\bgamma^{0}}(\Phi_n)\leq d_1\exp(-h_1n), \qquad \sup\limits_{\bgamma\in\mathcal{A}_n^c}E_{\bgamma}(1-\Phi_n)\leq d_2\exp(-h_2n)
\end{equation*}
for some $d_1, d_2, h_1, h_2>0$.
\begin{lemma}\label{lemma:numconc}
For some $h>0$, there exists a sequence of test functions for testing $H_0:\bgamma=\bgamma^{0}$ vs. $H_1:\bgamma\in \mathcal{A}_n^c$, which satisfy
\begin{equation}\label{exp_cons}
E_{\bgamma^{0}}(\Phi_n)\leq \exp(-h n), \qquad \sup\limits_{\bgamma\in\mathcal{A}_n^c}E_{\bgamma}(1-\Phi_n)\leq \exp(-h n).
\end{equation}
\end{lemma}
\begin{proof}
The construction of the test is provided in the proof of Theorem 2 and Lemma 4 in \cite{ghosal2006posterior}.
\end{proof}
We also state another result which will be subsequently used in the proof.
\begin{lemma}\label{lemma:priorc}
Let $\bu_k^{(0)}=(u_{k,1}^{(0)},...,u_{k,R}^{(0)})'$ for $k=1,..,V_n$, and $\upsilon_{k,l}$ be the only positive root of the equation
\begin{align}\label{rooteq}
x^2+x(||\bu_k^{(0)}||_2+||\bu_l^{(0)}||_2)-\eta_1=0,\:\:k<l.
\end{align}
Assume $\upsilon={\rm min}_{k,l}\:\upsilon_{k,l}$. Then, for $\bW=(\bu_1'\bu_2,...,\bu_{V_n-1}'\bu_{V_n})'$ and
$\bW^{(0)}=(\bu_1^{(0)'}\bu_2^{(0)},...,\bu_{V_n-1}^{(0)'}\bu_{V_n}^{(0)})'$
\begin{align}\label{Multifactor}
\Pi(||\bW-\bW^{(0)}||_{\infty}<\eta_1)\geq \Pi(||\bu_k-\bu_k^{(0)}||_2\leq\upsilon,\:\forall\:k=1,..,V_n).
\end{align}
\end{lemma}
\begin{proof}
for $k<l$,
\begin{align*}
|\bu_k'\bu_l-\bu_k^{(0)'}\bu_l^{(0)}| &=|\sum\limits_{r=1}^{R}u_{k,r}u_{l,r}-\sum\limits_{r=1}^{R}u_{k,r}^{(0)}u_{l,r}^{(0)}|\\
& \leq |\sum\limits_{r=1}^{R}(u_{k,r}-u_{k,r}^{(0)})u_{lr}|+|\sum\limits_{r=1}^{R}(u_{l,r}-u_{l,r}^{(0)})u_{k,r}^{(0)}|\\
&\leq ||\bu_k-\bu_k^{(0)}||_2||\bu_l||_2+||\bu_l-\bu_l^{(0)}||_2||\bu_k^{(0)}||_2\\
&\leq ||\bu_k-\bu_k^{(0)}||_2\left[||\bu_l-\bu_l^{(0)}||_2+||\bu_l^{(0)}||_2\right]+||\bu_l-\bu_l^{(0)}||_2||\bu_k^{(0)}||_2.
\end{align*}
If $||\bu_k-\bu_k^{(0)}||_2\leq\upsilon,\:\forall\:k=1,..,V_n$, the above inequality implies
\begin{align*}
|\bu_k'\bu_l-\bu_k^{(0)'}\bu_l^{(0)}|\leq \upsilon(\upsilon+||\bu_l^{(0)}||_2)+\upsilon||\bu_k^{(0)}||_2\leq \eta_1,\:\forall\:k<l.
\end{align*}
Hence $\Pi(||\bW-\bW^{(0)}||_{\infty}<\eta_1)\geq \Pi(||\bu_k-\bu_k^{(0)}||_2\leq\upsilon,\:\forall\:k=1,..,V_n)$.
\end{proof}

\noindent \underline{\textbf{Proof of Theorem~\ref{theorem:main}}}\\
Suppose $\mathcal{E}_n=\left\{\by:||\nabla w_{\bgamma^{(0)},n}(\by)||_{\infty}\leq 2\sqrt{nq_n}\right\}$. Then the probability of the vector $\by_n$ belonging to the set $\mathcal{E}_n$ is given by,
\begin{align*}
P_{\bgamma^{(0)}}(\by_n\in\mathcal{E}_n)\geq 1-P_{\bgamma^{(0)}}(\max\limits_{1\leq j\leq q_n}|\sum\limits_{i=1}^{n}(y_i-\nabla z(\bx_i'(\bgamma-\bgamma^{(0)})))x_{ij}|>2\sqrt{nq_n})\geq 1-\frac{2}{q_n},
\end{align*}
where the last step follows from the Hoeffding inequality. Note that as $n\rightarrow \infty$, $q_n\rightarrow\infty$, hence
$P_{\bgamma^{(0)}}(\by_n\in\mathcal{E}_n)\rightarrow 1$. Hence, in the subsequent proof we can assume without loss of generality that $\by_n\in\mathcal{E}_n$.
It can be observed that
\begin{align}
\Pi_n(\mathcal{A}_n^c) =\frac{\int_{\mathcal{A}_n^c}p_{\bgamma}(\by_n)\pi_n(\bgamma)}{\int p_{\bgamma}(\by_n)\pi_n(\bgamma)}
=\frac{\int_{\mathcal{A}_n^c}\frac{p_{\bgamma}(\by_n)}{p_{\bgamma^{(0)}}(\by_n)}\pi_n(\bgamma)}{\int \frac{p_{\bgamma}(\by_n)}{p_{\bgamma^{(0)}}(\by_n)}\pi_n(\bgamma)}
= \frac{\mathcal{N}_n}{\mathcal{D}_n}\leq \Phi_n+(1-\Phi_n)\frac{\mathcal{N}_n}{\mathcal{D}_n}\label{consistency},
\end{align}
where $\Phi_n$ is the exponentially consistent sequence of tests given in Lemma~\ref{lemma:numconc}. The above equation is true as $\mathcal{N}_n/\mathcal{D}_n\leq 1$. This is in turn true as both are integrals of the same nonnegative functions,  $\mathcal{D}_n$ is the integral of that function over the entire set of possible $\bgamma$'s, while $\mathcal{N}_n$ is the integral over a subset $\mathcal{A}_n^c$.
In proving Theorem~\ref{theorem:main}, we will proceed in three steps as following.
\begin{enumerate}[(a)]
\item Step 1 shows that $\Phi_n\rightarrow 0$, as $n\rightarrow \infty$, almost surely.
\item Step 2 shows that $\exp(hn/2)(1-\Phi_n)\mathcal{N}_n\rightarrow 0$, as $n\rightarrow\infty$, almost surely.
\item Finally, step 3 shows that $\exp(hn/2)\mathcal{D}_n\rightarrow\infty$, as $n\rightarrow\infty$.
\end{enumerate}
Here $h$ is the one as defined in Lemma~\ref{lemma:numconc}. By (\ref{consistency}), (a)-(c) implies $\Pi_n(\mathcal{A}_n^c)\rightarrow 0$. We will now proceed proving (a)-(c).

\noindent \underline{(a) Step 1}\\
An application of the Markov inequality and (\ref{exp_cons}) in Lemma~\ref{lemma:numconc} yield,
\begin{align*}
P_{\bgamma^{(0)}}\left(\Phi_n>\exp(-nh/2)\right)\leq E_{\bgamma^{(0)}}\left(\Phi_n\right)\exp(nh/2)\leq \exp(-nh/2).
\end{align*}
Therefore $\sum_{n=1}^{\infty}P_{\bgamma^{(0)}}\left(\Phi_n>\exp(-nh/2)\right)<\infty$.\\

Applying Borel-Cantelli lemma, 
Thus, $P_{\bgamma^{(0)}}(\Phi_n>\exp(-nh/2)\:\mbox{happens infinitely often})=0$. 
This means that $\exists\:n_0$ and a set $\Omega$ with $P_{\bgamma^{(0)}}(
\Omega)=0$,  s.t. for all $n>n_0$, $\Phi_n(\omega)<\exp(-nh/2)$, for all $\omega\in\Omega^c$. Since $\exp(-nh/2)\rightarrow 0$, this means that $\Phi_n\rightarrow 0$ almost surely.

Thus,
\begin{align}\label{eq:show1}
\Phi_n \rightarrow 0 \quad a.s.
\end{align}
\noindent \underline{(b) Step 2}\\
We have
\begin{align*}
E_{\bgamma^{(0)}}((1-\Phi_n)\mathcal{N}_n) &=\int (1-\Phi_n)\int_{\mathcal{A}_n^c}\frac{p_{\bgamma}(\by_n)}{p_{\bgamma^{(0)}}(\by_n)}\pi_n(\bgamma)p_{\bgamma^{(0)}}(\by_n)\\
&=\int_{\mathcal{A}_n^c}\int (1-\Phi_n)p_{\bgamma}(\by_n)\pi_n(\bgamma)\\
&=\int_{\mathcal{A}_n^c} E_{\bgamma}(1-\Phi_n)\pi_n(\bgamma)\\
&\leq \sup\limits_{\bgamma\in\mathcal{A}_n^c}E_{\bgamma}(1-\Phi_n \Pi(\mathcal{A}_n^c)\\
&\leq \sup\limits_{\bgamma\in\mathcal{A}_n^c}E_{\bgamma}(1-\Phi_n)\leq \exp(-nh)\leq \exp(-nh/2).
\end{align*}
Consider the set $\mathcal{G}_{n,h,2}=\{(1-\Phi_n)\mathcal{N}_n\exp(nh/2)>\exp(-nh/4)\}$. The above inequality implies that $\sum_{n=1}^{\infty}P_{\bgamma^{(0)}}(\mathcal{G}_{n,h,2})<\infty$. Again since $h$ is fixed, applying Borel-Cantelli lemma $P_{\bgamma^{(0)}}(limsup_{n\rightarrow\infty}\mathcal{G}_{n,h,2})=0$. Using the definition of limsup of the sets $\mathcal{G}_{n,h,2}$ \cite{klenke2013probability}, $P_{\bgamma^{(0)}}(\mathcal{G}_{n,h,2}\:\mbox{happens infinitely often})=0$. Thus, $P_{\bgamma^{(0)}}((1-\Phi_n)\mathcal{N}_n\exp(nh/2)>\exp(-nh/4)\:\mbox{happens infinitely often})=0$. Let $\Omega_2$ be the set s.t. $P_{\bgamma^{(0)}}(
\Omega)=0$ and $(1-\Phi_n(\omega))\mathcal{N}_n\exp(nh/2)>\exp(-nh/4)\:\mbox{happens infinitely often}$ for all $\omega\in\Omega_2$.
This means that $\exists\:n_{0,2}$ s.t. for all $n>n_{0,2}$, $(1-\Phi_n(\omega))\mathcal{N}_n\exp(nh/2)<\exp(-nh/4)$, for all $\omega\in\Omega_2^c$. Since $\exp(-nh/4)\rightarrow 0$, this means that $\exp(nh/2)(1-\Phi_n)\mathcal{N}_n \rightarrow 0$ almost surely.


\begin{align}
\exp(nh/2)(1-\Phi_n)\mathcal{N}_n \rightarrow 0 \quad a.s.\label{eq:show2}.
\end{align}
\noindent \underline{(c) Step 3}
\begin{align*}
\int\frac{p_{\bgamma}(\by_n)}{p_{\bgamma^{(0)}}(\by_n)}\pi(\bgamma)
&=\int\exp\left(\nabla w_{\bgamma^{(0)},n}(\by_n)'(\bgamma-\bgamma^{(0)})+C_{\by_n,n}(\bgamma)\right)\pi(\bgamma)\\
&\geq\int\exp\left(-||\nabla w_{\bgamma^{(0)},n}(\by_n)||_{\infty}||\bgamma-\bgamma^{(0)}||_2-\frac{n}{8}||\bgamma-\bgamma^{(0)}||_2^2\right)\pi(\bgamma)\\
&\geq\int\exp\left(-2\sqrt{nq_n}||\bgamma-\bgamma^{(0)}||_2-\frac{n}{8}||\bgamma-\bgamma^{(0)}||_2^2\right)\pi(\bgamma)\\
&\geq \exp\left(-2\sqrt{nq_n}\frac{\eta_1}{n^{\rho/2}}-\frac{n\eta_1^2}{8n^{\rho}}\right)\Pi\left(||\bgamma-\bgamma^{(0)}||_2<\frac{\eta_1}{n^{\rho/2}}\right),
\end{align*}
where $\rho$ is the one defined in the statement of the theorem and the inequality in the second line follows from the Taylor series expansion after taking into account that $\nabla^2z(\cdot)\leq 1/4$ ($z(\cdot)$ defined in (\ref{abar})), which is true as $\frac{d^2}{df^2}\log\left(1+e^f\right)=\frac{e^f}{(1+e^f)^2}\leq 1/4.$
The inequality in the third line follows from the fact that $\by_n\in\mathcal{E}_n$.

First, observe that, given all the hierarchical parameters, the Bayesian network lasso prior distribution on $\bgamma$ can be written as $\bgamma=\bW+\bgamma_2$, where $\bgamma_2$ follows the ordinary Bayesian lasso shrinkage prior. With this observation, one can see
\begin{align*}
\Pi\left(||\bgamma-\bgamma^{(0)}||_2<\frac{\eta_1}{n^{\rho/2}}\right)
\geq\Pi\left(||\bgamma_2-\bgamma_2^{(0)}||_2<\frac{\eta_1}{2n^{\rho/2}}\right)\Pi\left(||\bW-\bW^{(0)}||_2<\frac{\eta_1}{2n^{\rho/2}}\right),
\end{align*}
where $\bW$ and $\bW^{(0)}$ are as defined in Lemma~\ref{lemma:priorc}. We will show sequentially\\
(i) $-\log\Pi\left(||\bW-\bW^{(0)}||_2<\frac{\eta_1}{2n^{\rho/2}}\right)=o(n)$ and\\
(ii) $-\log\left\{\Pi\left(||\bgamma_2-\bgamma_2^{(0)}||_2<\frac{\eta_1}{2n^{\rho/2}}\right)\right\}=o(n)$.

\noindent (i) Note that, with $R$ (dimensions of the latent variables) and $\Delta$ (probability of a node being influential) as defined before we obtain,
\begin{align}\label{eq_int}
\Pi(||\bW-\bW^{(0)}||_2<\frac{\eta_1}{2n^{\rho/2}})
&\geq\Pi(||\bu_k-\bu_k^{(0)}||_2\leq\upsilon_n,\:\forall\:k=1,..,V_n)\nonumber\\
&\geq\mbox{E}\left[\Pi(||\bu_k-\bu_k^{(0)}||_2\leq\upsilon_n,\:\forall\:k=1,..,V_n|\Delta)\right]\nonumber\\
&\geq\mbox{E}\left[\prod\limits_{k=1}^{V_n}\left\{\exp\left(-\frac{1}{2}\bu_k^{(0)'}\bu_k^{(0)}\right)
\Pi(||\bu_k||_2\leq\upsilon_n|\Delta)\right\}\right],
\end{align}
where the first inequality follows from Lemma~\ref{lemma:priorc} by replacing $\eta_1$ with $\frac{\eta_1}{2n^{\rho/2}}$ with a slight abuse of notation, and $\upsilon_n$ is
defined accordingly. The last inequality follows from the Anderson's Lemma. We will now make use of the fact that
$\int_{-a}^{a}\exp(-x^2/2)dx\geq \exp(-a^2)2a$ to conclude
\begin{align*}
&\Pi(||\bu_k||_2\leq\upsilon_n|\Delta)\geq \prod\limits_{r=1}^{R}\Pi\left(|u_{k,r}|\leq\frac{\upsilon_n}{R}|\Delta\right)
=\prod\limits_{r=1}^{R}\left((1-\Delta)+\frac{\Delta}{\sqrt{2\pi}}\int_{-\upsilon_n/R}^{\upsilon_n/R}\exp(-x^2/2)\right)\\
&\geq\prod\limits_{r=1}^{R}\left((1-\Delta)+\frac{\Delta}{\sqrt{2\pi}}\exp(-\upsilon_n^2/R^2)\frac{2\upsilon_n}{R}\right)
\geq\left[(1-\Delta)+\frac{\Delta}{\sqrt{2\pi}}\exp(-\upsilon_n^2/R^2)\frac{2\upsilon_n}{R}\right]^{R}.
\end{align*}
\begin{align*}
\prod\limits_{k=1}^{V_n}\Pi(||\bu_k||_2\leq\upsilon_n)&\geq
\mbox{E}\left[(1-\Delta)+\frac{\Delta}{\sqrt{2\pi}}\exp(-\upsilon_n^2/R^2)\frac{2\upsilon_n}{R}\right]^{RV_n}\\
&=\mbox{E}\left[\sum\limits_{h_1=1}^{RV_n}{\binom{RV_n}{h_1}}(1-\Delta)^{h_1}\Delta^{RV_n-h_1}\left(\frac{2\upsilon_n}{R}\right)^{RV_n-h_1}\exp\left(-(RV_n-h_1)\upsilon_n^2/R^2\right)\right]\\
&\geq\sum\limits_{h_1=1}^{RV_n}{\binom{RV_n}{h_1}}Beta(RV_n-h_1+1,h_1+1)\\
&\left(\frac{2\upsilon_n}{R}\right)^{RV_n-h_1}\exp\left(-(RV_n-h_1)\upsilon_n^2/R^2\right)\\
&\geq\sum\limits_{h_1=1}^{RV_n}\frac{(RV_n)!}{h_1!(RV_n-h_1)!}\frac{h_1!(RV_n-h_1)!}{(RV_n+1)!}\\
&\left(\frac{2\upsilon_n}{R}\right)^{RV_n-h_1}\exp\left(-(RV_n-h_1)\upsilon_n^2/R^2\right)\\
&\geq\frac{RV_n}{RV_n+1}\left(\frac{2\upsilon_n}{R}\right)^{RV_n}\exp(-V_n\upsilon_n^2/R).
\end{align*}
Where the last inequality follows from Lemma~\ref{lemma:priorc} by considering the fact that,\\
 $\upsilon_n=\min\limits_{k,l}\frac{-[||\bu_k^{(0)}||+||\bu_l^{(0)}||]+\sqrt{[||\bu_k^{(0)}||+||\bu_l^{(0)}||]^2+2\eta_1/n^{\rho/2}}}{2}
\leq \frac{\sqrt{\eta_1}}{\sqrt{2}n^{\rho/4}}$. Hence, $0<\frac{2\upsilon_n}{R}<1$ for large $n$. It now follows from (\ref{eq_int}) that
\begin{align*}
-\log\Pi\left(||\bW-\bW^{(0)}||_2 <\frac{\eta_1}{2n^{\rho/2}}\right)&\leq\sum\limits_{k=1}^{V_n}\frac{\bu_k^{(0)'}\bu_k^{(0)}}{2}
+\frac{V_n\eta_1}{2Rn^{\rho/2}}-(RV_n)\log\left(\frac{2\sqrt{\eta_1}}{\sqrt{2}Rn^{\rho/4}}\right)+\log(RV_n+1)\\
&-\log(RV_n)=o(n),
\end{align*}
by the assumptions (A) and (B). This proves (i).

We will now prove (ii). Let $\mathcal{S}^0=\{j:\gamma_{2,j}^{(0)}\neq 0\}$. Define $\bs$ as the vector of upper triangular part of the matrix with $(k,l)$th entry $s_{k,l}$. It follows that
\begin{align}\label{eq:new}
\Pi\left(||\bgamma_2-\bgamma_2^{(0)}||_2<\frac{\eta_1}{2n^{\rho/2}}\right)
\geq \Pi\left(|\gamma_{2,j}-\gamma_{2,j}^{(0)}|<\frac{\eta_1}{2\sqrt{q_n}n^{\rho/2}},j\in\mathcal{S}^0\right)
\Pi\left(\sum_{j\not\in\mathcal{S}^0}|\bgamma_{2,j}|^2<\frac{(q_n-s_{2,n}^0)\eta_1^2}{4q_n n^{\rho}}\right).
\end{align}
We will lower bound two components of the product in (\ref{eq:new}) individually. By Chebyshev's inequality
\begin{align}\label{first}
\Pi\left(\sum_{j\not\in\mathcal{S}^0}|\gamma_{2,j}|^2<\frac{(q_n-s_{2,n}^0)\eta_1^2}{4q_n n^{\rho}}\right)
&\geq \left(1-\frac{E[\sum_{j\not\in\mathcal{S}^0}|\gamma_{2,j}|^2]4q_nn^{\rho}}{(q_n-s_{2,n}^0)\eta_1^2}\right)\nonumber\\
& = \left(1-\frac{2\theta_nq_nn^{\rho}}{\eta_1^2}\right).
\end{align}

\begin{align*}
&\Pi\left(|\gamma_{2,j}-\gamma_{2,j}^{(0)}|<\frac{\eta_1}{2\sqrt{q_n}n^{\rho/2}},j\in\mathcal{S}^0\right)
=E\left[\Pi\left(|\gamma_{2,j}-\gamma_{2,j}^{(0)}|<\frac{\eta_1}{2\sqrt{q_n}n^{\rho/2}},j\in\mathcal{S}^0|\bs_{\mathcal{S}^0}\right)\right]\\
&=E\left[\prod_{j\in\mathcal{S}^0}\Pi\left(|\gamma_{2,j}-\gamma_{2,j}^{(0)}|<\frac{\eta_1}{2\sqrt{q_n}n^{\rho/2}}|\bs_{\mathcal{S}^0}\right)\right].
\end{align*}
Using the fact that $\int_{a}^{b}e^{-x^2/2}dx\geq e^{-(a^2+b^2)/2}(b-a)$, one obtains
\begin{align*}
& \prod_{j\in\mathcal{S}^0}\Pi\left(|\gamma_{2,j}-\gamma_{2,j}^{(0)}|<\frac{\eta_1}{2\sqrt{q_n}n^{\rho/2}}|\bs_{\mathcal{S}^0}\right)
\geq \prod_{j\in\mathcal{S}^0}\left\{\left(\frac{\eta_1}{\sqrt{2q_nn^{\rho}\pi s_j^2}}\right)\exp\left(-\frac{|\gamma_{2,j}^{0}|^2+\eta_1^2/(4q_n n^{\rho})}{s_j^2}\right)\right\}.
\end{align*}
Thus
\begin{align*}
&\Pi\left(|\gamma_{2,j}-\gamma_{2,j}^{(0)}|<\frac{\eta_1}{2\sqrt{q_n}n^{\rho/2}},j\in\mathcal{S}^0\right)\\
& \geq E\left[\prod_{j\in\mathcal{S}^0}\left\{\left(\frac{\eta_1}{\sqrt{2q_nn^{\rho}\pi s_j^2}}\right)\exp\left(-\frac{|\gamma_{2,j}^{0}|^2+\eta_1^2/(4q_n n^{\rho})}{s_j^2}\right)\right\}\right]\nonumber\\
&\geq \left(\frac{\eta_1\theta_n}{\sqrt{2q_n n^{\rho}\pi}}\right)^{s_{2,n}^0}
\prod_{j\in\mathcal{S}^0}\int_{s_j}\left\{\frac{1}{\sqrt{s_j^2}}\exp\left(-\frac{|\gamma_{2,j}^{0}|^2+\eta_1^2/(4q_n n^{\rho})}{s_j^2}-\frac{\theta_n s_j^2}{2}\right)ds_j^2\right\}.
\end{align*}
Use the change of variable $\frac{1}{s_j^2}=z_j$ and the normalizing constant from the inverse Gaussian density to deduce
\begin{align*}
&\int_{s_j}\left\{\frac{1}{\sqrt{s_j^2}}\exp\left(-\frac{|\gamma_{2,j}^{0}|^2+\eta_1^2/(4q_n n^{\rho})}{s_j^2}-\frac{\theta_{n}s_j^2}{2}\right)ds_j^2\right\}\\
&=\int_{z_j}\left\{\frac{1}{\sqrt{z_j^3}}\exp\left(-(|\gamma_{2,j}^{0}|^2+\eta_1^2/(4q_n n^{\rho})z_{j}-\frac{\theta_{n}}{2z_{j}}\right)dz_{j}\right\}\\
&=\sqrt{\left(\frac{2\pi}{\theta_{n}}\right)}\exp\left(-\theta_{n}\sqrt{2\left(|\gamma_{2,j}^{0}|^2+\eta_1^2/(4q_n n^{\rho})\right)}\right).
\end{align*}
Therefore,
\begin{align}\label{second}
\Pi\left(|\gamma_{2,j}-\gamma_{2,j}^{(0)}|<\frac{\eta_1}{2\sqrt{q_n}n^{\rho/2}},j\in\mathcal{S}^0\right)
\geq \left(\frac{\eta_1\sqrt{\theta_n}}{\sqrt{q_n n^{\rho}}}\right)^{s_{2,n}^0} \exp\left(-\theta_{n}\sum_{j\in\mathcal{S}^0}\sqrt{2\left(|\gamma_{2,j}^{0}|^2+\eta_1^2/(4q_n n^{\rho})\right)}\right).
\end{align}
Combining results from (\ref{first}) and (\ref{second})
\begin{align*}
\Pi\left(||\bgamma_2-\bgamma_2^{(0)}||_2<\frac{\eta_1}{2n^{\rho/2}}\right)
&\geq \left(\frac{\eta_1\sqrt{\theta_n}}{\sqrt{q_n n^{\rho}}}\right)^{s_{2,n}^0} \exp\left(-\theta_{n}\sum_{j\in\mathcal{S}^0}\sqrt{2\left(|\gamma_{2,j}^{0}|^2+\eta_1^2/(4q_n n^{\rho})\right)}\right)\\
&\qquad \left(1-\frac{2\theta_nq_nn^{\rho/2}}{\eta_1^2}\right).
\end{align*}
Referring to Assumption (F),
\begin{align}\label{prob_n}
&-\log\Pi\left(||\bgamma_2-\bgamma_2^{(0)}||_2<\frac{\eta_1}{2n^{\rho/2}}\right)
\leq s_{2,n}^0[\eta+\log(q_n)+(3\rho/4)\log(n)+\log(\log(n))/2]\nonumber\\
&\qquad\qquad+\frac{\sqrt{2\left(|\gamma_{2,j}^{0}|^2+\eta_1^2/(4q_n n^{\rho})\right)}}{q_n n^{\rho/2}\log(n)}-\log\left(1-\frac{2}{\eta^2\log(n)}\right)=o(n),
\end{align}
under assumptions (B)-(F).

Finally,
\begin{align*}
-\log(\mathcal{D}_n)&\leq 2\sqrt{nq_n}\frac{\eta_1}{n^{\rho/2}}+\frac{n\eta^2}{8n^{\rho}}-\log\Pi\left(||\bgamma-\bgamma^{(0)}||_2<\frac{\eta_1}{n^{\rho/2}}\right)\\
&=2\eta\sqrt{q_n}n^{(1-\rho)/2}+\frac{\eta_1^2}{8}n^{1-\rho}-\log\Pi\left(||\bgamma-\bgamma^{(0)}||_2<\frac{\eta_1}{n^{\rho/2}}\right).
\end{align*}
Using (\ref{prob_n}), the fact that $(1-\rho)/2\in (-1/2,0)$ and assumption (B), we obtain $-\log(\mathcal{D}_n)=o(n)$. Thus
(c) follows.

\end{document}